\documentclass[11pt]{article}
\usepackage{amsmath,amssymb}
\usepackage{amsthm}
\usepackage{tikz}
\usepackage{mathrsfs}
\usetikzlibrary{arrows.meta, positioning, shapes.geometric,calc}
\usepackage{geometry}
\usepackage{titlesec}
\usepackage{setspace}
\usepackage{hyperref}
\hypersetup{hidelinks}
\usetikzlibrary{patterns}
\usepackage[capitalize,noabbrev,nameinlink]{cleveref}
\usepackage{graphicx}
\usepackage[utf8]{inputenc}
\usepackage{array}
\usepackage{tabularx}

\theoremstyle{plain}
\newtheorem{lemma}{Lemma}
\newtheorem{proposition}{Proposition}
\newtheorem{corollary}{Corollary}
\newtheorem{theorem}{Theorem}

\theoremstyle{definition}
\newtheorem{definition}{Definition}[section]

\theoremstyle{remark}
\newtheorem{remark}{Remark}

\usepackage{enumitem}

\theoremstyle{plain}

\begin{document}

\newcommand{\MiddleEar}{\mathrm{MiddleEar}}
\newcommand{\HumanEar}{\mathrm{HumanEar}}
\newcommand{\ZeroSymbol}{\mathrm{ZeroSymbol}}

\newcommand{\IntT}[1]{\operatorname{Int}_{\mathcal T}\!\left(#1\right)}
\newcommand{\downset}[1]{\mathord{\downarrow}\!#1}

\hypersetup{pageanchor=false}
\begin{titlepage}
    \centering
    \vspace*{2.5cm}

    {\huge\bfseries Heraclitean Dialectical Concept Space}
    \vspace{0.5cm}
    
    {\large\bfseries  An Attempt to Model Conceptual Evolution Using Topological Structures  \par}

    \vspace{4em}
    {\Large\scshape Volkan Yildiz\par}
    \vspace{0.3cm}
    {\large\texttt{volkan@hotmail.co.uk}\par}
    \vfill
    {\normalsize November, 2026  \par}

\begin{abstract}
We introduce Heraclitean Dialectical Concept Spaces (HDCS), a topological
framework for modelling stagewise conceptual change. At each stage, concepts
are atomic elements of a set equipped with a feasible family and a
neighbourhood assignment. The feasible family generates an external topology,
while neighbourhood overlaps define adjacency and channel structures through
which local conceptual interactions can be represented.

Emergence is modelled by the CCER construction. Given an interaction profile,
an uncovered remainder is combined with anchored neighbourhood witnesses to
produce a nonempty open candidate emergent region. Such a region is not itself
identified with a new concept: selected candidates may instead be reified as
atomic concepts at a subsequent stage, while earlier concepts may persist
through partial injective carry maps. This distinction allows persistence, disappearance,
reorganization, and emergence to be represented independently.

Stagewise development is assembled into an HDCS history space by quotienting
the disjoint union of stage concept sets by the identifications generated by
carry maps and equipping the quotient with the corresponding final topology.
Product concept spaces extend the same construction to interactions between
distinct conceptual domains. The framework also yields derived structural
consequences, including support characterizations of interaction, canonical
CCER envelopes, monotonicity and impossibility criteria, product-factorization
results, and invariance under concept-space isomorphism.

Applications to exchange and monetary transformation, the historical
development of zero and its representations, mammalian middle-ear evolution,
and AI-related labour transformation illustrate how the framework can be used
without assuming a universal linear trajectory or treating topology itself as
a causal theory of historical change.
    
\vfill



\end{abstract}

\end{titlepage}
\hypersetup{pageanchor=true}

\newpage

\tableofcontents 

\newpage

\newpage

\begin{flushleft}
{\em For Ares.}
\end{flushleft}

\newcommand
\HRule{\noindent\rule{\linewidth}{0.5pt}}

\HRule
{\small
\begin{flushright}
{\em  ``A man cannot step into the same river twice...''}
\end{flushright}
}
\begin{flushright}Heraclitus\end{flushright}
\HRule

\section*{Motivational Note}

The term \emph{dialectic} is used in this article in a broad structural sense,
to refer to change, interaction, transformation, persistence through change,
and the emergence of new forms from existing configurations. It does not refer
to Hegelian thesis--antithesis--synthesis, Hegelian logic, or a necessary
teleological philosophy of history.

HDCS instead asks a narrower mathematical question: whether aspects of
conceptual change can be represented using feasible regions, neighbourhoods,
overlaps, remainders, stagewise evolution, and reification. Topology provides
the structural language of the model; causal and historical interpretations
remain external to the formalism.

The framework is exploratory rather than definitive and is offered as one
possible mathematical approach to conceptual evolution.

This work may contain typos and errors. 
If you have read it and are not satisfied with any part of it, 
please do get in touch; it may well be something I have overlooked.

\vspace{1 em}
\noindent
\textit{Volkan Yildiz,\\
London,\\ 
November-2026.}

\newpage
\section*{Introduction}
\addcontentsline{toc}{section}{Introduction}

HDCS is intended as a formal language for representing conceptual structures
that may persist, reorganize, disappear, or give rise to newly reified concepts
across stages. Its basic objects are deliberately modest: a set of atomic
concepts, a feasible family of regions, and a neighbourhood assignment. From
these data the framework derives an external topology, adjacency and overlap
relations, channel ideals, candidate emergent regions, and stagewise history.

A central modelling distinction is between the identity of a concept and its
structural role at a particular stage. A concept may persist from stage
\(i\) to stage \(i+1\) while its feasible contexts and neighbourhood relations
change. Thus stagewise change need not mean that an entirely new concept has
been invented; it may instead represent a reorganization of the relations in
which an existing concept participates. Reification is reserved for the
separate operation in which a selected structured region at one stage is
represented by a new atomic concept at a later stage.

The adjective \emph{dialectical} is used only in this structural sense of
interaction and change. No thesis--antithesis--synthesis scheme, Hegelian logic,
or teleological sequence is assumed.

The framework is topological rather than metric. It does not assign numerical
distances between concepts, and it does not treat topology as a causal theory.
Historical, biological, linguistic, economic, or technological evidence is
used to motivate and constrain particular feasible regions, neighbourhoods,
profiles, and stage transitions; the topology then records the structural
consequences of those modelling choices.

The single-stage construction is developed first, followed by mediated
channels and CCER candidate emergence. The stagewise theory then separates
persistence, reification, and transition, and assembles declared persistence
relations into a history space. Product concept spaces extend the same
single-stage machinery to interactions between distinct conceptual domains.
The applications are therefore examples of how the formal architecture can be
instantiated, not claims that all conceptual or historical change follows one
universal trajectory. A later section derives consequences that no longer
depend on a particular application, including canonical emergence envelopes,
monotonicity laws, impossibility criteria, product-factorization results, and
structural invariance under relabelling. The stage index records structurally distinguished
states of a model; it need not be interpreted as a uniform physical time
parameter. HDCS is therefore a stagewise theory of structural change rather
than a continuous-time dynamics.

\newpage

\newpage
\section*{Notation}

\begingroup
\small
\renewcommand{\arraystretch}{1.30}

\subsection*{1. Foundational and Topological Structure}

\begin{tabularx}{\textwidth}{@{}
>{\raggedright\arraybackslash}p{0.27\textwidth}
>{\raggedright\arraybackslash}X
@{}}

\textbf{Symbol} & \textbf{Meaning} \\
\hline

\(C\)
&
Set of atomic concepts in a single-stage concept space.
\\

\(\mathcal P(C)\)
&
Power set of \(C\).
\\

\(\mathcal F\subseteq\mathcal P(C)\)
&
Feasible family of regions.
\\

\((F1)\)
&
Covering axiom:
\(\displaystyle \bigcup\mathcal F=C.\)
\\

\((F_{\cap})\)
&
Overlap-intersection axiom:
if \(U,V\in\mathcal F\) and
\(U\cap V\neq\varnothing\), then
\(U\cap V\in\mathcal F\).
\\

\((F_{\varnothing})\)
&
Nonemptiness axiom:
\(\displaystyle \varnothing\notin\mathcal F.\)
\\

\(\mathcal T=\tau_{\mathcal F}\)
&
External topology generated by the feasible-family basis:
\(\displaystyle
\mathcal T
=
\left\{
\bigcup\mathcal A:
\mathcal A\subseteq\mathcal F
\right\}.
\)
\\

\(N:C\to\mathcal P(\mathcal F)\)
&
Neighbourhood assignment. For each \(x\in C\),
\(N(x)\) is a nonempty family of feasible regions assigned to \(x\).
\\

\((N0)\)
&
Neighbourhood anchoring:
\(\displaystyle
N(x)\neq\varnothing
\)
and
\(\displaystyle
U\in N(x)\Rightarrow x\in U.
\)
\\

\((N_{\downarrow})\)
&
Downward neighbourhood condition:
if
\(U\in N(x)\),
\(V\in\mathcal F\), and
\(x\in V\subseteq U\),
then \(V\in N(x)\).
\\

\(N(A)\)
&
Common-neighbourhood family of \(A\subseteq C\):
\(\displaystyle
N(A)
:=
\bigcap_{x\in A}N(x),
\qquad
N(\varnothing):=\mathcal F.
\)
\\

\(\mathcal T|_D\)
&
Subspace topology on \(D\subseteq C\):
\(\displaystyle
\mathcal T|_D
=
\{\,U\cap D:U\in\mathcal T\,\}.
\)
\\

\(\operatorname{Int}_{\mathcal T}(A)\)
&
Interior of \(A\subseteq C\) with respect to the external topology
\(\mathcal T\).
\\

\end{tabularx}

\endgroup


\newpage
\section*{Notation (continued)}

\begingroup
\small
\renewcommand{\arraystretch}{1.30}

\subsection*{2. Adjacency, Overlaps, and Channels}

\begin{tabularx}{\textwidth}{@{}
>{\raggedright\arraybackslash}p{0.27\textwidth}
>{\raggedright\arraybackslash}X
@{}}

\textbf{Symbol} & \textbf{Meaning} \\
\hline

\(X\sim Y\)
&
Direct adjacency of concepts \(X,Y\in C\). It holds when there exist
\(U\in N(X)\) and \(V\in N(Y)\) such that
\(U\cap V\neq\varnothing\).
\\

\(O(X,Y)\)
&
Overlap family:
\(\displaystyle
O(X,Y)
=
\{
U\cap V:
U\in N(X),\
V\in N(Y),\
U\cap V\neq\varnothing
\}.
\)
\\

\(I_{XY}\) or \(I_{ij}\)
&
Direct channel ideal:
\(\displaystyle
I_{XY}
=
\left\{
W\in\mathcal T:
\exists\text{ finite }\mathcal A\subseteq O(X,Y)
\text{ such that }
W\subseteq\bigcup_{R\in\mathcal A}R
\right\}.
\)
\\

\(|I_{XY}|\)
&
Support of the direct channel ideal:
\(\displaystyle
|I_{XY}|
:=
\bigcup I_{XY}.
\)
\\

\(X_0\sim X_1\sim\cdots\sim X_k\)
&
Finite adjacency walk. Repeated concepts are allowed unless an additional
distinctness condition is imposed.
\\

\(W_r\)
&
Overlap witness for the step
\(\displaystyle X_{r-1}\sim X_r.\)
\\

\(I^{(k)}_{XY}\)
&
Exact \(k\)-step mediated channel ideal, when exact walk length is being
distinguished.
\\

\(I^{(\leq k)}_{XY}\)
&
Cumulative mediated channel ideal generated by witnessed adjacency walks
from \(X\) to \(Y\) of length at most \(k\).
\\

\(I^{(*)}_{XY}\)
&
Comprehensive mediated channel ideal generated by all finite witnessed
adjacency walks from \(X\) to \(Y\).
\\

\(X\sim^{(*)}Y\)
&
Mediated adjacency: there exists a finite adjacency walk from \(X\) to \(Y\).
\\

\(\mathrm{Ch}(X\Rightarrow Y)\)
&
Channel region associated with a specified mediated adjacency construction
from \(X\) to \(Y\), when this notation is used.
\\

\end{tabularx}

\endgroup


\newpage
\section*{Notation (continued)}

\begingroup
\small
\renewcommand{\arraystretch}{1.30}

\subsection*{3. Profiles, Remainders, and CCER}

\begin{tabularx}{\textwidth}{@{}
>{\raggedright\arraybackslash}p{0.27\textwidth}
>{\raggedright\arraybackslash}X
@{}}

\textbf{Symbol} & \textbf{Meaning} \\
\hline

\(p=(R,U,V_i,V_j)\)
&
Interaction profile between parent concepts \(C_i,C_j\).
Here \(R\) is a selected overlap region, \(U\in\mathcal T\) is a
background region, and \(V_i,V_j\) are selected parent neighbourhoods
containing \(R\).
\\

\(N_R(C_i)\)
&
Family of \(R\)-compatible neighbourhoods:
\(\displaystyle
N_R(C_i)
=
\{\,V\in N(C_i):R\subseteq V\,\}.
\)
\\

\(R_E(p)\)
&
Remainder associated with a profile:
\(\displaystyle
R_E(p)
=
\operatorname{Int}_{\mathcal T}
\bigl(
U\setminus(V_i\cup V_j)
\bigr).
\)
\\

\(E_{ij}\)
&
Remainder ideal associated with \(C_i,C_j\), generated by open subsets
of finite unions of productive remainders.
\\

\(K_i,K_j\)
&
Finite cores used in a CCER construction.
\\

\(N(K_i\cup\{C_i\})\neq\varnothing\)
&
Local admissibility condition for the core \(K_i\).
The analogous condition is imposed on \(K_j\).
\\

\(W_i,W_j\)
&
Anchored witnesses:
\(\displaystyle
W_i\in N(K_i\cup\{C_i\}),
\qquad
W_j\in N(K_j\cup\{C_j\}).
\)
\\

\(E\)
&
Candidate emergent region:
\(\displaystyle
E
=
W_i\cap W_j\cap R_E(p).
\)
The witness pair is productive when \(E\neq\varnothing\).
\\

\(\mathcal T|_E\)
&
Internal topology of \(E\), namely the subspace topology inherited from
the ambient stage:
\(\displaystyle
\mathcal T|_E
=
\{\,G\in\mathcal T:G\subseteq E\,\}
\)
when \(E\) is open.
\\

\(\mathscr E_i^{\mathrm{sel}}\)
&
Selected family of stage-\(i\) candidate regions designated for
reification.
\\

\(\rho_i\)
&
Injective reification map:
\(\displaystyle
\rho_i:
\mathscr E_i^{\mathrm{sel}}
\longrightarrow
C^{(i+1)}.
\)
\\

\(e_E=\rho_i(E)\)
&
Reified next-stage concept corresponding to the selected candidate region
\(E\).
\\

\(\operatorname{Real}_{i+1}(e_E)\)
&
Realization/provenance record:
\(\displaystyle
\operatorname{Real}_{i+1}(e_E)
=
\bigl(
E,\mathcal T^{(i)}|_E
\bigr).
\)
\\

\end{tabularx}

\endgroup


\newpage
\section*{Notation (continued)}

\begingroup
\small
\renewcommand{\arraystretch}{1.30}

\subsection*{4. Stage Dynamics and History Space}

\begin{tabularx}{\textwidth}{@{}
>{\raggedright\arraybackslash}p{0.27\textwidth}
>{\raggedright\arraybackslash}X
@{}}

\textbf{Symbol} & \textbf{Meaning} \\
\hline

\(\mathfrak C^{(i)}\)
&
Stage-\(i\) concept space:
\(\displaystyle
\mathfrak C^{(i)}
=
\bigl(
C^{(i)},\mathcal F^{(i)},N^{(i)}
\bigr).
\)
\\

\(C^{(i)}\)
&
Set of atomic concepts present at stage \(i\).
\\

\(\mathcal F^{(i)}\)
&
Feasible family at stage \(i\).
\\

\(N^{(i)}\)
&
Neighbourhood assignment at stage \(i\).
\\

\(\mathcal T^{(i)}\)
&
External topology at stage \(i\):
\(\displaystyle
\mathcal T^{(i)}
=
\tau_{\mathcal F^{(i)}}.
\)
\\

\(\mathrm{Events}^{(i)}\)
&
Event data supplied to the stage transition mechanism, including selected
reifications and other permitted structural changes.
\\

\(\operatorname{Supp}_i\)
&
Optional support assignment used when locality is invoked:
\(\displaystyle
\operatorname{Supp}_i:\mathcal V_i\to\bigl(\mathcal T^{(i)}\setminus\{\varnothing\}\bigr)
\), where \(\mathcal V_i\) is the set of elementary declared events at stage
\(i\).
\\

\(\Phi_i\)
&
Stage transition mechanism:
\(\displaystyle
\mathfrak C^{(i+1)}
=
\Phi_i
\bigl(
\mathfrak C^{(i)},
\mathrm{Events}^{(i)}
\bigr).
\)
Its output must again be a valid concept space.
\\

\(D_i\subseteq C^{(i)}\)
&
Domain of the stage-\(i\) carry map.
\\

\(\sigma_i:D_i\to C^{(i+1)}\)
&
Injective partial carry map representing declared persistence or identity
through change from stage \(i\) to stage \(i+1\). Distinct concepts are
never identified merely because they later contribute to one new structure.
\\

\(\sigma_i(D_i)\cap\rho_i(\mathscr E_i^{\mathrm{sel}})=\varnothing\)
&
Separation of persistence and reification: no next-stage atom is both a carried
successor and a newly reified concept.
\\

\(\mathfrak U\)
&
Timeline set:
\(\displaystyle
\mathfrak U
=
\bigsqcup_i
\bigl(
\{i\}\times C^{(i)}
\bigr).
\)
\\

\(\approx\)
&
Smallest equivalence relation on \(\mathfrak U\) generated by
\(\displaystyle
(i,x)
\approx
(i+1,\sigma_i(x))
\)
for \(x\in D_i\).
\\

\(C_{\mathrm{hist}}\)
&
History space:
\(\displaystyle
C_{\mathrm{hist}}
=
\mathfrak U/\!\approx.
\)
\\

\(q\)
&
Quotient map:
\(\displaystyle
q:
\mathfrak U
\longrightarrow
C_{\mathrm{hist}}.
\)
\\

\(q_i\)
&
Canonical stage map:
\(\displaystyle
q_i:
C^{(i)}
\longrightarrow
C_{\mathrm{hist}},
\qquad
q_i(x)=q(i,x).
\)
\\

\(\mathcal T_{\mathrm{hist}}\)
&
Final history topology characterized by
\(\displaystyle
O\in\mathcal T_{\mathrm{hist}}
\iff
q_i^{-1}(O)\in\mathcal T^{(i)}
\text{ for every }i.
\)
\\

\(\operatorname{Sat}_i(A)\)
&
Saturation of a stage subset \(A\subseteq C^{(i)}\):
\(\displaystyle
\operatorname{Sat}_i(A)
=
q^{-1}\bigl(q_i(A)\bigr).
\)
\\

\end{tabularx}

\endgroup


\newpage
\section*{Notation (continued)}

\begingroup
\small
\renewcommand{\arraystretch}{1.30}

\subsection*{5. Product and Cross-Space Structure}

\begin{tabularx}{\textwidth}{@{}
>{\raggedright\arraybackslash}p{0.27\textwidth}
>{\raggedright\arraybackslash}X
@{}}

\textbf{Symbol} & \textbf{Meaning} \\
\hline

\(C_1\times C_2\)
&
Underlying set of the product concept space formed from two concept spaces.
\\

\(\mathcal F_{\times}\)
&
Rectangular feasible family:
\(\displaystyle
\mathcal F_{\times}
=
\{
U\times V:
U\in\mathcal F_1,\;
V\in\mathcal F_2
\}.
\)
\\

\(N_{\times}\)
&
Product neighbourhood assignment:
\(\displaystyle
N_{\times}(x_1,x_2)
=
\{
U\times V:
U\in N_1(x_1),\
V\in N_2(x_2)
\}.
\)
\\

\(\mathcal T_{\times}\)
&
Product topology:
\(\displaystyle
\mathcal T_{\times}
=
\mathcal T_1\times\mathcal T_2.
\)
\\

\(X_a,X_b\)
&
Concepts in the product space, i.e.
\(\displaystyle X_a,X_b\in C_1\times C_2.\)
\\

\(O_{\times}(X_a,X_b)\)
&
Overlap family between product-space concepts, computed using product
neighbourhoods.
\\

\(p_{\times}\)
&
Product-space interaction profile:
\(\displaystyle
p_{\times}
=
(R,U,V_a,V_b).
\)
\\

\(R_{E,\times}(p_{\times})\)
&
Product-space remainder:
\(\displaystyle
R_{E,\times}(p_{\times})
=
\operatorname{Int}_{\mathcal T_{\times}}
\bigl(
U\setminus(V_a\cup V_b)
\bigr).
\)
\\

\(E_{\times}\)
&
Candidate emergent region obtained by applying ordinary CCER inside a
product concept space.
\\

\(I^{\times}_{ab}\)
&
Channel ideal associated with the product concepts \(X_a,X_b\).
\\

\(\pi_1,\pi_2\)
&
Coordinate projections:
\(\displaystyle
\pi_1:C_1\times C_2\to C_1,
\qquad
\pi_2:C_1\times C_2\to C_2.
\)
\\

\(E_1,\ldots,E_m\)
&
Already valid CCER candidate regions used in a finite multi-profile
construction.
\\

\(E_{\mathrm{multi}}\)
&
Finite multi-profile candidate:
\(\displaystyle
E_{\mathrm{multi}}
=
\bigcap_{\ell=1}^{m}E_\ell
\neq\varnothing.
\)
\\

\(\mathfrak C_{\times}^{(i,j)}\)
&
Product concept space formed from factor stages \(i\) and \(j\):
\(\displaystyle
\mathfrak C_{\times}^{(i,j)}
=
\mathfrak C_1^{(i)}
\times
\mathfrak C_2^{(j)}.
\)
\\

\(s:n\mapsto(i_n,j_n)\)
&
Optional coordinatewise nondecreasing synchronization schedule used to represent
asynchronously evolving factor spaces as a single sequence of joint stages.
\\

\end{tabularx}

\endgroup

\pagebreak

\section{Neighbourhoods and feasible families}

We introduce \emph{feasible families} as the basic structure used to define
neighbourhood assignments and adjacency relations. They will later also
constrain the construction of emergent conceptual regions.

Throughout, concepts are treated as elements of an abstract set. No intrinsic
geometric or metric structure is assumed. The initial structure is relational
and is expressed through containment and overlap among regions, together with
the neighbourhood assignments introduced below. Regions are defined
extensionally as subsets of the concept set.

\medskip

Let \(C\) be a nonempty set whose elements are called \emph{concepts}.
A \emph{region} is any subset \(U \subseteq C\).

A feasible family specifies which regions are admitted as coherent contexts
for structural and inferential purposes. Such regions may subsequently be used
in the definitions of neighbourhoods, adjacency relations, and emergent
structures.

\begin{definition}\label{def:feasible}
Let \(C\) be a nonempty set. A \emph{feasible family} on \(C\) is a family
\[
\mathcal{F} \subseteq \mathcal{P}(C)
\]
satisfying:
\[
\textnormal{(F1)}\qquad
\forall x \in C\ \exists\, U \in \mathcal{F}
\text{ such that } x \in U;
\]
\[
\textnormal{(F$\cap$)}\qquad
U,V \in \mathcal{F},\quad U \cap V \neq \varnothing
\ \Longrightarrow\ U \cap V \in \mathcal{F};
\]
\[
\textnormal{(F$\varnothing$)}\qquad
\varnothing \notin \mathcal{F}.
\]
\end{definition}

Condition \textup{(F1)} says that \(\mathcal{F}\) covers \(C\), so every concept
belongs to at least one feasible region. Condition \textup{(F$\cap$)} says that
whenever two feasible regions overlap, their common part remains feasible.
Condition \textup{(F$\varnothing$)} stipulates that every feasible region is
nonempty.

These conditions make \(\mathcal{F}\) a basis for a topology on \(C\), namely
\[
\tau_{\mathcal{F}}
:=
\left\{
\bigcup \mathcal{A}
:
\mathcal{A} \subseteq \mathcal{F}
\right\}.
\]
Nevertheless, \(\mathcal{F}\) is not itself assumed to be a topology: it
excludes the empty set and need not be closed under unions.

We interpret feasibility as a stronger requirement than topological openness.
A union of feasible regions may be open in \(\tau_{\mathcal{F}}\) without
itself being feasible, since arbitrary unions need not preserve semantic
coherence.

\begin{remark}
The empty set is excluded from the feasible family. Once the neighbourhood
assignment is introduced below, condition \textup{(N0)} also implies that
every assigned neighbourhood contains its concept and is therefore nonempty.
Thus empty sets cannot serve as overlap witnesses.
\end{remark}

\newpage

\begin{definition}\label{def:neigh}
Let \(C\) be a nonempty set and let \(\mathcal F\) be a feasible family on
\(C\). A map
\[
N:C\longrightarrow\mathcal P(\mathcal F),
\qquad
P\longmapsto N(P),
\]
is called a \emph{neighbourhood assignment} if, for every \(P\in C\), the
following conditions hold:
\[
\textnormal{(N0)}\qquad
N(P)\neq\varnothing
\quad\text{and}\quad
U\in N(P)\Longrightarrow P\in U;
\]
\[
\textnormal{(N$\downarrow$)}\qquad
U\in N(P),\quad V\in\mathcal F,\quad P\in V\subseteq U
\Longrightarrow V\in N(P).
\]
\end{definition}

Thus, \(N(P)\) consists of feasible regions in which \(P\) is locally
situated. The second condition says that any feasible restriction of such a
region that still contains \(P\) remains a neighbourhood of \(P\).

\begin{lemma}\label{lem:neigh-derived}
Let \(N\) be a neighbourhood assignment. Then, for every \(P\in C\):

\begin{enumerate}
\item every \(U\in N(P)\) is nonempty;
\item if \(U,V\in N(P)\), then \(U\cap V\in N(P)\).
\end{enumerate}
\end{lemma}

\begin{proof}
If \(U\in N(P)\), then \(P\in U\) by \textup{(N0)}, and hence
\(U\neq\varnothing\).

Now let \(U,V\in N(P)\). By \textup{(N0)}, \(P\in U\cap V\), so
\(U\cap V\neq\varnothing\). Since \(U,V\in\mathcal F\), condition
\textup{(F$\cap$)} gives
\[
U\cap V\in\mathcal F.
\]
Moreover,
\[
P\in U\cap V\subseteq U.
\]
Therefore \textup{(N$\downarrow$)} implies
\[
U\cap V\in N(P).
\]
\end{proof}

For any subset \(A\subseteq C\), define its family of common neighbourhoods by
\[
N(A):=\bigcap_{P\in A}N(P),
\qquad
N(\varnothing):=\mathcal F.
\]

Thus, \(U\in N(A)\) precisely when \(U\) is a neighbourhood of every element
of \(A\). In particular, if \(A\neq\varnothing\) and \(U\in N(A)\), then
\[
A\subseteq U.
\]

\noindent\textbf{Antitone behaviour.}\;
If \(A\subseteq B\subseteq C\), then
\[
N(B)\subseteq N(A).
\]
Indeed, every neighbourhood common to all elements of \(B\) is also common to
all elements of \(A\).

\begin{remark}
Unlike classical topological neighbourhood systems, the families \(N(P)\) are
not assumed to be upward closed. If \(U\in N(P)\) and
\[
U\subseteq W\subseteq C,
\]
we do not require \(W\in N(P)\), even when \(W\in\mathcal F\). Enlarging a
region may introduce conceptual incompatibilities, whereas feasible
restriction is required to preserve the local situation of \(P\).
\end{remark}

\begin{remark}\label{rem:finite-common}
For a finite set \(A\subseteq C\), the family \(N(A)\) may be empty. When a
particular finite set is intended to form a coherent core, we shall require
explicitly that
\[
N(A)\neq\varnothing.
\]

One could impose the stronger global condition
\[
\tag{\(N_{\mathrm{fin}}\)}
N(A)\neq\varnothing
\qquad
\text{for every nonempty finite }A\subseteq C.
\]
We do not assume \textup{\((N_{\mathrm{fin}})\)} by default, since it has strong
consequences for the adjacency relation introduced below.
\end{remark}

\begin{definition}\label{def:adj}
Let \(X,Y\subseteq C\) be nonempty subsets. We say that \(X\) and \(Y\) are
\emph{adjacent}, and write
\[
X\sim Y,
\]
if there exist \(U\in N(X)\) and \(V\in N(Y)\) such that 
\(
U\cap V\neq\varnothing.
\)
\end{definition}

\begin{proposition}\label{prop:adj-basic}
The adjacency relation is symmetric. Moreover, for every nonempty
\(X\subseteq C\),
\[
X\sim X
\quad\Longleftrightarrow\quad
N(X)\neq\varnothing.
\]
In particular, every individual concept is adjacent to itself.
\end{proposition}

\begin{proof}
Symmetry follows from
\[
U\cap V\neq\varnothing
\quad\Longleftrightarrow\quad
V\cap U\neq\varnothing.
\]

If \(N(X)\neq\varnothing\), choose \(U\in N(X)\). Since \(X\neq\varnothing\)
and \(X\subseteq U\), the set \(U\) is nonempty. Hence
\[
U\cap U=U\neq\varnothing,
\]
which proves \(X\sim X\).

Conversely, if \(X\sim X\), then the definition supplies elements of \(N(X)\),
so \(N(X)\neq\varnothing\).

For a concept \(P\in C\), one has
\[
N(\{P\})=N(P)\neq\varnothing
\]
by \textup{(N0)}.
\end{proof}

\begin{proposition}\label{prop:Nfin-complete}
Assume \textup{\((N_{\mathrm{fin}})\)}. Then every two nonempty finite subsets
\(X,Y\subseteq C\) are adjacent:
\[
X\sim Y.
\]
\end{proposition}

\begin{proof}
The set \(X\cup Y\) is finite and nonempty, so
\[
N(X\cup Y)\neq\varnothing.
\]
Choose \(U\in N(X\cup Y)\). By antitonicity,
\[
U\in N(X)
\qquad\text{and}\qquad
U\in N(Y).
\]
Since \(U\neq\varnothing\), the intersection
\[
U\cap U=U
\]
witnesses \(X\sim Y\).
\end{proof}

\begin{proposition}[Basic properties of adjacency]\label{prop:adj-properties}
Let \(\sim\) be the adjacency relation from Definition~\ref{def:adj}.
Then for all nonempty \(X,Y,Z \subseteq C\):
\begin{enumerate}[label=\arabic*.]
  \item \emph{Reflexive on its domain.}
  If \(N(X) \neq \varnothing\), then \(X \sim X\).
  In particular, for every \(P \in C\) one has \(\{P\} \sim \{P\}\).
  \item \emph{Symmetric.}
  \(X \sim Y\) if and only if \(Y \sim X\).
  \item \emph{Not necessarily transitive.}
  There is an example with concepts \(P,Q,R\in C\) such that
  \[
  \begin{aligned}
  \{P\}&\sim\{Q\}, \qquad \{Q\}\sim\{R\},\\
  \{P\}&\not\sim\{R\}.
  \end{aligned}
  \]
\end{enumerate}
\end{proposition}

\begin{proof}
\emph{(1)}
If \(N(X) \neq \varnothing\), choose \(U \in N(X)\).
Then \(U \cap U \neq \varnothing\), so \(X \sim X\).
For a singleton \(\{P\}\), nonemptiness of \(N(\{P\})\) follows directly from
\textup{(N0)}, hence \(\{P\} \sim \{P\}\).

\medskip

\emph{(2)}
If \(U \in N(X)\) and \(V \in N(Y)\) satisfy \(U \cap V \neq \varnothing\), then
the same pair witnesses \(Y \sim X\).
Thus \(X \sim Y \Leftrightarrow Y \sim X\).

\medskip

\emph{(3)}
Let
\[
C=\{P,Q,R,a,b\},
\qquad
\mathcal F=\mathcal P(C)\setminus\{\varnothing\}.
\]
Define
\[
N(P)=\{\{P\},\{P,a\}\},
\qquad
N(R)=\{\{R\},\{R,b\}\},
\]
and
\[
N(Q)=\{\{Q\},\{Q,a\},\{Q,b\}\}.
\]
These assignments satisfy \textup{(N0)} and \textup{(N$\downarrow$)}.
Now
\[
\{P,a\}\cap\{Q,a\}=\{a\}\neq\varnothing,
\]
so \(\{P\}\sim\{Q\}\), and
\[
\{Q,b\}\cap\{R,b\}=\{b\}\neq\varnothing,
\]
so \(\{Q\}\sim\{R\}\). However every neighbourhood of \(P\) is
contained in \(\{P,a\}\), whereas every neighbourhood of \(R\) is
contained in \(\{R,b\}\); hence their intersections are empty and
\[
\{P\}\not\sim\{R\}.
\]
\end{proof}

\begin{remark}
Adjacency is intended to represent meaningful interaction between nonempty
collections of concepts.
Accordingly, adjacency involving the empty set is not regarded as meaningful,
and we restrict attention throughout to nonempty subsets of \(C\).
\end{remark}

\begin{definition}[Strict adjacency]\label{def:strictadj}
The \emph{strict adjacency} relation on nonempty subsets of \(C\) is defined by
\[
X \sim_{\mathrm{str}} Y
\iff
X \neq Y
\ \text{and}\ 
\exists\,U \in N(X),\ V \in N(Y)
\ \text{such that}\ 
U \cap V \neq \varnothing .
\]
\end{definition}

\begin{definition}\label{def:opens}
Let \(C\) be a set of concepts, \(\mathcal F \subseteq \mathcal P(C)\) a feasible
family, and
\(N : C \to \mathcal P(\mathcal F)\) a neighbourhood assignment satisfying
\textup{(N0)} and \textup{(N$\downarrow$)}.
A subset \(U \subseteq C\) is called \emph{open} if
\[
\forall P \in U\ \exists\,V \in N(P)\ \text{such that}\ V \subseteq U.
\]
Let \(\mathcal T\) denote the family of all open subsets of \(C\).
(No assumption is made that \(U \in \mathcal F\).)
\end{definition}

\begin{proposition}[Neighbourhood topology]\label{prop:topology}
Under \textup{(N0)} and \textup{(N$\downarrow$)},
the family \(\mathcal T\) is a topology on \(C\).
\end{proposition}

\begin{proof}
We verify the topology axioms.

\medskip

\noindent\emph{(i) The empty set and the whole space.}
The empty set is vacuously open.
For \(C\), let \(P \in C\).
By \textup{(N0)} there exists \(V \in N(P)\), and clearly \(V \subseteq C\).

\medskip

\noindent\emph{(ii) Arbitrary unions.}
Let \(\{U_i\}_{i \in I} \subseteq \mathcal T\) and set \(U := \bigcup_{i \in I} U_i\).
If \(P \in U\), then \(P \in U_i\) for some \(i\).
Since \(U_i\) is open, there exists \(V \in N(P)\) with \(V \subseteq U_i \subseteq U\).
Hence \(U\) is open.

\medskip

\noindent\emph{(iii) Finite intersections.}
Let \(U,W \in \mathcal T\) and \(P \in U \cap W\).
There exist \(V_1,V_2 \in N(P)\) with \(V_1 \subseteq U\) and \(V_2 \subseteq W\).
By Lemma~\ref{lem:neigh-derived}, \(V_1 \cap V_2 \in N(P)\), and
\(V_1 \cap V_2 \subseteq U \cap W\).
Thus \(U \cap W\) is open.

\medskip

Therefore, \(\mathcal T\) is a topology on \(C\).
\end{proof}

\newpage

\begin{proposition}[The external topology]\label{prop:external-topology}
Let \(C\) be a nonempty set, let \(\mathcal F\) be a feasible family on \(C\),
and let \(N\) be a neighbourhood assignment. Define
\[
\mathcal T_N
:=
\left\{
U\subseteq C:
\forall P\in U\ \exists\,V\in N(P)
\text{ such that }V\subseteq U
\right\}.
\]
Then:

\begin{enumerate}[label=\arabic*.]
\item every feasible region \(U\in\mathcal F\) belongs to \(\mathcal T_N\);
\item every neighbourhood region \(U\in N(P)\) belongs to \(\mathcal T_N\);
\item the cross-point property
\[
Q\in U\in N(P)
\Longrightarrow
\exists\,W\in N(Q)\text{ such that }W\subseteq U
\]
holds automatically;
\item
\[
\mathcal T_N
=
\tau_{\mathcal F}
:=
\left\{
\bigcup\mathcal A:
\mathcal A\subseteq\mathcal F
\right\}.
\]
\end{enumerate}
\end{proposition}

\begin{proof}
Let \(U\in\mathcal F\), and let \(P\in U\). By \textup{(N0)}, choose
\(V\in N(P)\). Since \(P\in U\cap V\), the intersection is nonempty, and
\textup{(F$\cap$)} gives
\[
U\cap V\in\mathcal F.
\]
Moreover,
\[
P\in U\cap V\subseteq V.
\]
Hence \textup{(N$\downarrow$)} implies
\[
U\cap V\in N(P).
\]
Since \(U\cap V\subseteq U\), the set \(U\) is open. This proves~\textup{(1)}.

Now let \(U\in N(P)\), and let \(Q\in U\). By \textup{(N0)}, choose
\(V\in N(Q)\). Since \(Q\in U\cap V\), condition
\textup{(F$\cap$)} gives
\[
U\cap V\in\mathcal F.
\]
Because
\[
Q\in U\cap V\subseteq V,
\]
condition \textup{(N$\downarrow$)} yields
\[
U\cap V\in N(Q).
\]
Thus every point \(Q\in U\) has a neighbourhood contained in \(U\), proving
that \(U\in\mathcal T_N\). This proves~\textup{(2)} and also establishes the
cross-point property in~\textup{(3)}.

To prove~\textup{(4)}, first let \(U\in\mathcal T_N\). For each \(P\in U\),
choose \(V_P\in N(P)\) such that \(V_P\subseteq U\). Since \(P\in V_P\),
\[
U=\bigcup_{P\in U}V_P.
\]
Each \(V_P\) belongs to \(\mathcal F\), so \(U\in\tau_{\mathcal F}\). Hence
\[
\mathcal T_N\subseteq\tau_{\mathcal F}.
\]

Conversely, every member of \(\mathcal F\) is open by~\textup{(1)}, and
arbitrary unions of open sets are open. Therefore
\[
\tau_{\mathcal F}\subseteq\mathcal T_N.
\]
Thus
\[
\mathcal T_N=\tau_{\mathcal F}.
\]
\end{proof}

\begin{definition}\label{def:external-topology}
The topology
\[
\mathcal T
:=
\mathcal T_N
=
\tau_{\mathcal F}
\]
is called the \emph{external topology} of the concept space.
\end{definition}

\begin{remark}[The role of the neighbourhood assignment]
Although the external topology is determined by the feasible family
\(\mathcal F\), the neighbourhood assignment \(N\) remains additional
structure. It specifies which feasible regions are associated locally with
each concept and therefore contributes to later notions such as common
neighbourhoods, adjacency, overlap witnesses, and interaction channels.
Different neighbourhood assignments on the same feasible family may therefore
induce the same external topology while producing different relational
structures.
\end{remark}

\begin{remark}[The empty set]
The empty set belongs to \(\mathcal T\), since it is the union of the empty
subfamily of \(\mathcal F\):
\[
\varnothing=\bigcup\varnothing.
\]
Nevertheless,
\[
\varnothing\notin\mathcal F
\]
by \textup{(F$\varnothing$)}, and no neighbourhood region is empty, since
\(P\in U\) whenever \(U\in N(P)\).

Thus, the empty set is an open set but is neither a feasible region nor a
neighbourhood region.
\end{remark}

We refer to \(\mathcal T\) as the \emph{external topology} because it describes
the organisation of the concept set \(C\) at a given stage. Since
\(\mathcal T\) is a topology on the set of concepts itself, it should not be
confused with a topology placed on the objects, interpretations, or
realizations represented by those concepts.

When a region \(U\subseteq C\) is under consideration, we equip it with the
subspace topology
\[
\mathcal T|_U
:=
\{\,O\cap U:O\in\mathcal T\,\}.
\]
We call this the \emph{internal topology of \(U\)}. It records the open-set
structure inherited by \(U\) from the ambient concept space.

\newpage

\begin{proposition}[Restriction of the external topology]
\label{prop:subspace-base}
Let \(C'\subseteq C\). Then:

\begin{enumerate}[label=\arabic*.]
\item the subspace topology on \(C'\) is
\[
\mathcal T|_{C'}
=
\{\,U\cap C':U\in\mathcal T\,\};
\]

\item if \(\mathcal B\) is a basis for \(\mathcal T\), then
\[
\mathcal B|_{C'}
:=
\{\,B\cap C':B\in\mathcal B,\ B\cap C'\neq\varnothing\,\}
\]
is a basis for \(\mathcal T|_{C'}\).
\end{enumerate}
\end{proposition}

\begin{proof}
The first statement is the definition of the subspace topology.

For the second, every member of \(\mathcal B|_{C'}\) is open in the subspace
topology. If \(O\in\mathcal T\), write
\[
O=\bigcup_{\alpha\in I}B_\alpha,
\qquad B_\alpha\in\mathcal B.
\]
Then
\[
O\cap C'
=
\bigcup_{\alpha\in I}(B_\alpha\cap C').
\]
After omitting empty intersections, this expresses every subspace-open set as
a union of members of \(\mathcal B|_{C'}\).
\end{proof}

\begin{lemma}\label{lem:open-characterizations}
For \(U\subseteq C\), the following are equivalent:
\begin{align*}
\textup{(i)}\;& U\in\mathcal T;\\
\textup{(ii)}\;&
\forall P\in U\ \exists\,V_P\in N(P)
\text{ such that }V_P\subseteq U;\\
\textup{(iii)}\;&
U \text{ is a union of neighbourhood regions};\\
\textup{(iv)}\;&
U \text{ is a union of feasible regions}.
\end{align*}
\end{lemma}

\begin{proof}
The equivalence of \textup{(i)} and \textup{(ii)} is the definition of
\(\mathcal T\).

If \textup{(ii)} holds, then
\[
U=\bigcup_{P\in U}V_P,
\]
so \textup{(iii)} holds. Every neighbourhood region belongs to
\(\mathcal F\), so \textup{(iii)} implies \textup{(iv)}. Finally, every
feasible region is open by Proposition~\ref{prop:external-topology}, and hence
every union of feasible regions is open. Thus \textup{(iv)} implies
\textup{(i)}.
\end{proof}

\begin{remark}[Open regions and common neighbourhoods]
If
\[
U\in\bigcap_{P\in U}N(P),
\]
then \(U\) is open, since \(U\) itself is a neighbourhood of every one of its
points. The converse need not hold.

For example, let \(C\) contain at least two elements, let
\[
\mathcal F=\mathcal P(C)\setminus\{\varnothing\},
\]
and define
\[
N(P):=\{\{P\}\}.
\]
Then the induced topology is discrete, so every subset of \(C\) is open.
However, a set \(U\) containing two or more points does not belong to
\(N(P)\) for any \(P\in U\).
\end{remark}

\begin{corollary}[Neighbourhood basis]\label{cor:neighbourhood-base}
Let
\[
\mathcal B_N:=\bigcup_{P\in C}N(P)
\]
be the collection of all neighbourhood regions. Then \(\mathcal B_N\) is a
basis for \(\mathcal T\). In particular, every \(U\in\mathcal T\) satisfies
\[
U
=
\bigcup_{P\in U}
\bigcup
\{\,V\in N(P):V\subseteq U\,\}.
\]
\end{corollary}

\begin{proof}
Every member of \(\mathcal B_N\) is open by
Proposition~\ref{prop:external-topology}. Moreover, the defining condition for
openness ensures that every open set is a union of members of
\(\mathcal B_N\).
\end{proof}

\newpage

\providecommand{\IntT}[1]{%
  \operatorname{Int}_{\mathcal T}\!\left(#1\right)}
\providecommand{\clT}[1]{%
  \operatorname{cl}_{\mathcal T}\!\left(#1\right)}

\section{Structural adjacency via neighbourhood overlap}

Throughout this section, let \(C\) be a nonempty concept set,
\(\mathcal F\) a feasible family on \(C\), \(N\) a neighbourhood assignment,
and
\[
\mathcal T=\tau_{\mathcal F}
\]
the external topology.

When no confusion can arise, we abbreviate
\[
\{C_i\}\sim\{C_j\}
\]
by
\[
C_i\sim C_j.
\]

\begin{definition}[Overlap family]\label{def:overlap-family}
Let \(C_i,C_j\in C\). Their \emph{overlap family} is
\[
O(C_i,C_j)
:=
\left\{
U\cap V
\;\middle|\;
U\in N(C_i),\
V\in N(C_j),\
U\cap V\neq\varnothing
\right\}.
\]
\end{definition}

\begin{remark}
Every member of \(O(C_i,C_j)\) is a nonempty feasible region and therefore
a nonempty open set. Indeed, if \(U\in N(C_i)\) and \(V\in N(C_j)\), then
\(U,V\in\mathcal F\), and
\[
U\cap V\neq\varnothing
\Longrightarrow
U\cap V\in\mathcal F
\]
by \textup{(F$\cap$)}.
\end{remark}

Unlike an upward-closed topological neighbourhood system, the present
neighbourhood assignment need not give two distinct concepts a common
neighbourhood or even intersecting neighbourhoods. Thus
\(O(C_i,C_j)\) may be empty.

\begin{definition}[Channel ideal]\label{def:overlap-channel}
The \emph{channel ideal} generated by the overlap family is
\[
I_{ij}
:=
\left\{
W\in\mathcal T
\;\middle|\;
\begin{array}{l}
\text{there exist \(n\geq 0\) and }
R_1,\ldots,R_n\in O(C_i,C_j)\\[1mm]
\text{such that }
W\subseteq\displaystyle\bigcup_{k=1}^{n}R_k
\end{array}
\right\},
\]
where the union for \(n=0\) is understood to be \(\varnothing\).

Equivalently,
\[
I_{ij}
=
\downarrow
\left\{
\bigcup_{R\in\mathcal A}R
\;\middle|\;
\mathcal A\subseteq O(C_i,C_j)
\text{ is finite}
\right\},
\]
where, for \(\mathcal S\subseteq\mathcal T\),
\[
\downarrow\mathcal S
:=
\left\{
W\in\mathcal T
\;\middle|\;
\exists A\in\mathcal S
\text{ such that }W\subseteq A
\right\}.
\]
\end{definition}

The channel ideal consists of all open regions contained in finite unions of
direct overlap regions. Since every overlap region is already open, no
interior operator is needed in this definition.

\begin{remark}
The word \emph{ideal} is used in the order-theoretic sense. The topology
\(\mathcal T\), ordered by inclusion, is a distributive lattice under finite
unions and intersections. The family \(I_{ij}\) is a lower set closed under
finite unions.
\end{remark}

\begin{proposition}[Adjacency, overlap, and nontrivial channels]
\label{prop:adj-channel-equivalence}
For \(C_i,C_j\in C\), the following statements are equivalent:
\begin{enumerate}[label=\arabic*.]
\item \(C_i\sim C_j\);
\item \(O(C_i,C_j)\neq\varnothing\);
\item \(I_{ij}\neq\{\varnothing\}\);
\item \(I_{ij}\) contains a nonempty open set.
\end{enumerate}
\end{proposition}

\begin{proof}
The equivalence of \textup{(1)} and \textup{(2)} follows immediately from the
definitions of adjacency and the overlap family.

Suppose \(O(C_i,C_j)\neq\varnothing\), and choose
\[
R\in O(C_i,C_j).
\]
The set \(R\) is nonempty and open. Taking the finite family
\(\mathcal A=\{R\}\) shows that
\[
R\in I_{ij}.
\]
Hence \(I_{ij}\neq\{\varnothing\}\), and \(I_{ij}\) contains a nonempty open
set.

Conversely, suppose that \(I_{ij}\) contains a nonempty open set \(W\).
By definition, there exist finitely many overlaps
\(R_1,\ldots,R_n\in O(C_i,C_j)\) such that
\[
W\subseteq\bigcup_{k=1}^{n}R_k.
\]
Since \(W\neq\varnothing\), one must have \(n\geq 1\). Therefore
\(O(C_i,C_j)\neq\varnothing\).
\end{proof}

\begin{definition}[Direct channel]\label{def:direct-channel}
The channel between \(C_i\) and \(C_j\) is called \emph{nontrivial} when
\[
I_{ij}\neq\{\varnothing\}.
\]
When
\[
I_{ij}=\{\varnothing\},
\]
we say that \(C_i\) and \(C_j\) are \emph{directly structurally disconnected}.
\end{definition}

\begin{remark}
Direct structural disconnection does not exclude the possibility that
\(C_i\) and \(C_j\) may be connected indirectly through one or more
intermediate concepts. Mediated channels will therefore require a separate
definition.
\end{remark}

\begin{proposition}[Basic properties of the channel ideal]
\label{prop:ideal-props}
For all \(C_i,C_j\in C\):
\begin{enumerate}[label=\arabic*.]
\item \(\varnothing\in I_{ij}\);
\item \(I_{ij}\) is downward closed in \(\mathcal T\);
\item \(I_{ij}\) is closed under finite unions;
\item \(I_{ij}=I_{ji}\).
\end{enumerate}
\end{proposition}

\begin{proof}
The empty set belongs to \(I_{ij}\) by taking \(n=0\).

Suppose \(W\in I_{ij}\) and \(W'\in\mathcal T\) satisfies
\[
W'\subseteq W.
\]
Choose overlaps \(R_1,\ldots,R_n\) such that
\[
W\subseteq\bigcup_{k=1}^{n}R_k.
\]
Then
\[
W'\subseteq\bigcup_{k=1}^{n}R_k,
\]
so \(W'\in I_{ij}\).

Now suppose \(W_1,W_2\in I_{ij}\). Choose finite overlap families
\(\mathcal A_1,\mathcal A_2\subseteq O(C_i,C_j)\) such that
\[
W_a\subseteq\bigcup_{R\in\mathcal A_a}R,
\qquad a=1,2.
\]
Then
\[
W_1\cup W_2
\subseteq
\bigcup_{R\in\mathcal A_1\cup\mathcal A_2}R.
\]
Since \(\mathcal A_1\cup\mathcal A_2\) is finite,
\[
W_1\cup W_2\in I_{ij}.
\]

Finally,
\[
O(C_i,C_j)=O(C_j,C_i)
\]
because intersection is symmetric. Therefore
\[
I_{ij}=I_{ji}.
\]
\end{proof}

\begin{definition}[\(R\)-compatible neighbourhoods and witnesses]
\label{def:R-compatible}
Let \(R\in O(C_i,C_j)\). Define
\[
N_R(C_i)
:=
\left\{
U\in N(C_i)
\;\middle|\;
R\subseteq U
\right\},
\]
and
\[
N_R(C_j)
:=
\left\{
V\in N(C_j)
\;\middle|\;
R\subseteq V
\right\}.
\]

If \(R\in O(C_i,C_j)\), a \emph{witness pair} for \(R\) is a pair
\[
(U,V)\in N(C_i)\times N(C_j)
\]
such that
\[
R=U\cap V.
\]
\end{definition}

\begin{lemma}\label{lem:R-compat-nonempty}
If \(R\in O(C_i,C_j)\), then
\[
N_R(C_i)\neq\varnothing
\qquad\text{and}\qquad
N_R(C_j)\neq\varnothing.
\]
\end{lemma}

\begin{proof}
Since \(R\in O(C_i,C_j)\), there exist
\[
U_0\in N(C_i)
\qquad\text{and}\qquad
V_0\in N(C_j)
\]
such that
\[
R=U_0\cap V_0\neq\varnothing.
\]
Thus
\[
R\subseteq U_0
\qquad\text{and}\qquad
R\subseteq V_0.
\]
Hence
\[
U_0\in N_R(C_i)
\qquad\text{and}\qquad
V_0\in N_R(C_j).
\]
\end{proof}

\begin{definition}[Interaction profile and uncovered remainder]
\label{def:profile}
Let \(R\in O(C_i,C_j)\). An \emph{interaction profile} based on \(R\) is a
tuple
\[
p=(R,U,V_i,V_j)
\]
such that
\[
U\in\mathcal T,
\qquad
R\subseteq U,
\qquad
V_i\in N_R(C_i),
\qquad
V_j\in N_R(C_j).
\]

Define the restricted parent regions
\[
R_i(p):=U\cap V_i,
\qquad
R_j(p):=U\cap V_j,
\]
and define the \emph{uncovered remainder} of the profile by
\[
R_E(p)
:=
\IntT{U\setminus(V_i\cup V_j)}.
\]

Let \(\operatorname{Prof}_{ij}\) denote the family of all such profiles.
A profile \(p\) is called \emph{productive} when
\[
R_E(p)\neq\varnothing.
\]
\end{definition}

The expression \(R_E(p)\) records the open portion of the context \(U\) that
is not covered by either selected parent neighbourhood. By itself, this is a
set-theoretic notion of uncovered structure. Any interpretation of this
region as failure, contradiction, tension, or historical insufficiency
requires additional domain-specific assumptions.

\begin{proposition}[Basic properties of an interaction profile]
\label{prop:profile-basic}
Let
\[
p=(R,U,V_i,V_j)\in\operatorname{Prof}_{ij}.
\]
Then:
\begin{enumerate}[label=\arabic*.]
\item
\[
R\subseteq R_i(p)\cap R_j(p);
\]

\item \(R_i(p)\), \(R_j(p)\), and \(R_i(p)\cap R_j(p)\) are open;

\item
\[
R_i(p)\cap R_j(p)\in I_{ij};
\]

\item
\[
R_E(p)\subseteq U;
\]

\item
\[
R_E(p)\cap(V_i\cup V_j)=\varnothing;
\]

\item \(R_E(p)\) is open.
\end{enumerate}
\end{proposition}

\begin{proof}
Since
\[
R\subseteq U,\qquad R\subseteq V_i,\qquad R\subseteq V_j,
\]
we have
\[
R\subseteq U\cap V_i\cap V_j
=
R_i(p)\cap R_j(p).
\]

The sets \(U,V_i,V_j\) are open. Hence their finite intersections are open,
which proves~\textup{(2)}.

Since
\[
V_i\cap V_j\in O(C_i,C_j)
\]
and
\[
R_i(p)\cap R_j(p)
=
U\cap V_i\cap V_j
\subseteq
V_i\cap V_j,
\]
downward closure of \(I_{ij}\) gives
\[
R_i(p)\cap R_j(p)\in I_{ij}.
\]

The remaining statements follow directly from
\[
R_E(p)
=
\IntT{U\setminus(V_i\cup V_j)}.
\]
\end{proof}

\begin{lemma}[Separation of the remainder from parent-contained witnesses]
\label{lem:remainder-separation}
Let
\[
p=(R,U,V_i,V_j)\in\operatorname{Prof}_{ij}.
\]
If
\[
W_i\subseteq V_i
\qquad\text{and}\qquad
W_j\subseteq V_j,
\]
then
\[
W_i\cap W_j\cap R_E(p)=\varnothing.
\]
\end{lemma}

\begin{proof}
Since
\[
W_i\cap W_j
\subseteq
V_i\cap V_j
\subseteq
V_i\cup V_j
\]
and
\[
R_E(p)\cap(V_i\cup V_j)=\varnothing,
\]
the stated intersection is empty.
\end{proof}

\begin{remark}[Consequence for later emergence rules]
Lemma~\ref{lem:remainder-separation} shows that no nonempty emergent region can
be defined by
\[
\IntT{W_i\cap W_j\cap R_E(p)}
\]
while simultaneously requiring
\[
W_i\subseteq V_i
\qquad\text{and}\qquad
W_j\subseteq V_j.
\]
Any later emergence rule using this expression must therefore either allow
the witnesses to extend beyond \(V_i\) and \(V_j\), or replace the uncovered
remainder by a different interface construction.
\end{remark}

\begin{definition}[Remainder ideal]\label{def:ext-ideal}
The \emph{remainder ideal} associated with \(C_i\) and \(C_j\) is
\[
E_{ij}
:=
\left\{
W\in\mathcal T
\;\middle|\;
\begin{array}{l}
\text{there exist \(n\geq 0\) and }
p_1,\ldots,p_n\in\operatorname{Prof}_{ij}\\[1mm]
\text{such that }
W\subseteq\displaystyle\bigcup_{k=1}^{n}R_E(p_k)
\end{array}
\right\}.
\]
As before, the union for \(n=0\) is understood to be \(\varnothing\).
\end{definition}

\begin{proposition}[Basic properties of the remainder ideal]
\label{prop:ext-ideal}
For all \(C_i,C_j\in C\):
\begin{enumerate}[label=\arabic*.]
\item \(\varnothing\in E_{ij}\);
\item \(E_{ij}\) is downward closed;
\item \(E_{ij}\) is closed under finite unions;
\item \(R_E(p)\in E_{ij}\) for every
\(p\in\operatorname{Prof}_{ij}\);
\item \(E_{ij}=E_{ji}\);
\item
\[
E_{ij}\neq\{\varnothing\}
\]
if and only if there exists a productive profile
\(p\in\operatorname{Prof}_{ij}\).
\end{enumerate}
\end{proposition}

\begin{proof}
The first four statements follow directly from the definition, using the same
finite-union argument as for \(I_{ij}\).

For symmetry, if
\[
p=(R,U,V_i,V_j)\in\operatorname{Prof}_{ij},
\]
then
\[
p^{\mathrm{op}}=(R,U,V_j,V_i)\in\operatorname{Prof}_{ji},
\]
and
\[
R_E(p^{\mathrm{op}})
=
\IntT{U\setminus(V_j\cup V_i)}
=
R_E(p).
\]
Hence \(E_{ij}=E_{ji}\).

If a productive profile exists, then its nonempty remainder belongs to
\(E_{ij}\), so \(E_{ij}\neq\{\varnothing\}\).

Conversely, suppose \(E_{ij}\) contains a nonempty open set \(W\). Then
\[
W\subseteq\bigcup_{k=1}^{n}R_E(p_k)
\]
for some finite collection of profiles. Since \(W\neq\varnothing\), at least
one remainder \(R_E(p_k)\) must be nonempty. Thus at least one of the profiles
is productive.
\end{proof}

\newpage

\section{Multi-step adjacency and mediated channels}

Throughout this section, let
\[
(C,\mathcal F,N)
\]
be a concept space with external topology
\[
\mathcal T=\tau_{\mathcal F}.
\]

\begin{definition}[Adjacency walks]\label{def:adjacency-walk}
Let \(k\geq1\), and let \(C_i,C_j\in C\).

A \emph{length-\(k\) adjacency walk} from \(C_i\) to \(C_j\) is a sequence
\[
P=(X_0,\ldots,X_k)
\]
such that
\[
X_0=C_i,
\qquad
X_k=C_j,
\]
and
\[
X_{r-1}\sim X_r
\qquad
\text{for }r=1,\ldots,k.
\]

Repeated concepts are permitted. We denote the family of all such walks by
\[
\mathcal P_k(C_i,C_j).
\]
\end{definition}

\begin{definition}[Witnessed walk and its support]
\label{def:witnessed-walk}
Let
\[
P=(X_0,\ldots,X_k)\in\mathcal P_k(C_i,C_j).
\]

A \emph{witnessing family} for \(P\) is a sequence
\[
(W_1,\ldots,W_k)
\]
such that
\[
W_r\in O(X_{r-1},X_r)
\qquad
\text{for }r=1,\ldots,k.
\]

The \emph{support} of the witnessed walk is
\[
S(P;W_1,\ldots,W_k)
:=
\bigcup_{r=1}^{k}W_r.
\]

The \emph{\(k\)-step overlap-support family} is
\[
\mathcal O^{(k)}(C_i,C_j)
:=
\left\{
S(P;W_1,\ldots,W_k)
\;\middle|\;
\begin{array}{l}
P\in\mathcal P_k(C_i,C_j),\\
W_r\in O(X_{r-1},X_r)
\text{ for }r=1,\ldots,k
\end{array}
\right\}.
\]
\end{definition}

\begin{lemma}\label{lem:k-support-open}
Every member of
\[
\mathcal O^{(k)}(C_i,C_j)
\]
is a nonempty open subset of \(C\).
\end{lemma}

\begin{proof}
Each witness
\[
W_r\in O(X_{r-1},X_r)
\]
is a nonempty feasible region and therefore a nonempty open set.

Since \(k\geq1\), the union
\[
\bigcup_{r=1}^{k}W_r
\]
is nonempty. Since arbitrary unions of open sets are open, the support is
open.
\end{proof}

\begin{remark}[Aggregate support versus connected mediation]
The support
\[
S(P;W_1,\ldots,W_k)
\]
records all overlap regions used by the witnessed walk. It need not be
connected, since the definitions do not imply
\[
W_r\cap W_{r+1}\neq\varnothing.
\]

Thus, the support represents aggregate mediated structure rather than,
by itself, a continuous topological channel.
\end{remark}

\begin{definition}[Linked witnessed walk]\label{def:linked-walk}
A witnessed adjacency walk is called \emph{linked} if
\[
W_r\cap W_{r+1}\neq\varnothing
\qquad
\text{for }r=1,\ldots,k-1.
\]

The family of supports of linked length-\(k\) walks from \(C_i\) to \(C_j\)
is denoted by
\[
\mathcal O_{\mathrm{link}}^{(k)}(C_i,C_j).
\]
\end{definition}

\begin{remark}
The linked condition is optional. It should be imposed when the intended
interpretation requires successive overlap regions to meet. An ordinary
adjacency walk need not admit linked witnesses.
\end{remark}

\begin{definition}[\(k\)-step channel ideal]
\label{def:k-step-ideal}
For \(k\geq1\), define
\[
I^{(k)}_{ij}
:=
\left\{
W\in\mathcal T
\;\middle|\;
\begin{array}{l}
\text{there exists a finite family }
\mathcal A\subseteq\mathcal O^{(k)}(C_i,C_j)\\[1mm]
\text{such that }
W\subseteq\displaystyle\bigcup_{S\in\mathcal A}S
\end{array}
\right\}.
\]

The empty finite family is permitted, with
\[
\bigcup\varnothing=\varnothing.
\]

For \(K\geq1\), define the cumulative channel ideal
\[
I^{(\leq K)}_{ij}
:=
\left\{
W\in\mathcal T
\;\middle|\;
\begin{array}{l}
\text{there exists a finite family}\\
\mathcal A\subseteq
\displaystyle\bigcup_{k=1}^{K}
\mathcal O^{(k)}(C_i,C_j)\\[1mm]
\text{such that }
W\subseteq\displaystyle\bigcup_{S\in\mathcal A}S
\end{array}
\right\},
\]
and define
\[
I^{(*)}_{ij}
:=
\left\{
W\in\mathcal T
\;\middle|\;
\begin{array}{l}
\text{there exists a finite family}\\
\mathcal A\subseteq
\displaystyle\bigcup_{k\geq1}
\mathcal O^{(k)}(C_i,C_j)\\[1mm]
\text{such that }
W\subseteq\displaystyle\bigcup_{S\in\mathcal A}S
\end{array}
\right\}.
\]
\end{definition}

The ideal \(I^{(k)}_{ij}\) collects open regions supported by finitely many
length-\(k\) witnessed adjacency walks. The cumulative ideal
\(I^{(\leq K)}_{ij}\) records mediation requiring at most \(K\) adjacency
steps, while \(I^{(*)}_{ij}\) allows arbitrary finite walk length.

\begin{proposition}[Properties of multi-step channel ideals]
\label{prop:multi-step-ideals}
Let \(C_i,C_j\in C\). Then:

\begin{enumerate}[label=\arabic*.]

\item For every \(k\geq1\), the families
\[
I^{(k)}_{ij},
\qquad
I^{(\leq K)}_{ij},
\qquad
I^{(*)}_{ij}
\]
are ideals of \(\mathcal T\).

\item The constructions are symmetric:
\[
I^{(k)}_{ij}=I^{(k)}_{ji},
\]
\[
I^{(\leq K)}_{ij}=I^{(\leq K)}_{ji},
\]
and
\[
I^{(*)}_{ij}=I^{(*)}_{ji}.
\]

\item The cumulative ideals are monotone:
\[
I^{(\leq1)}_{ij}
\subseteq
I^{(\leq2)}_{ij}
\subseteq
\cdots
\subseteq
I^{(*)}_{ij}.
\]

Moreover, whenever \(k\leq K\),
\[
I^{(k)}_{ij}\subseteq I^{(\leq K)}_{ij}.
\]

\item The following are equivalent:
\[
\mathcal P_k(C_i,C_j)\neq\varnothing;
\]
\[
\mathcal O^{(k)}(C_i,C_j)\neq\varnothing;
\]
\[
I^{(k)}_{ij}\neq\{\varnothing\}.
\]

\item The one-step construction recovers the direct channel:
\[
I^{(1)}_{ij}=I_{ij}.
\]

\end{enumerate}
\end{proposition}

\begin{proof}
For~\textup{(1)}, the empty set belongs to each ideal by taking the empty
finite family.

Suppose \(W\) belongs to one of these ideals and \(W'\in\mathcal T\) satisfies
\[
W'\subseteq W.
\]
Any finite support family covering \(W\) also covers \(W'\), so the ideal is
downward closed.

If \(W_1\) and \(W_2\) belong to the same ideal, choose finite support
families \(\mathcal A_1\) and \(\mathcal A_2\) such that
\[
W_a\subseteq\bigcup_{S\in\mathcal A_a}S,
\qquad a=1,2.
\]
Then
\[
W_1\cup W_2
\subseteq
\bigcup_{S\in\mathcal A_1\cup\mathcal A_2}S.
\]
The union \(\mathcal A_1\cup\mathcal A_2\) is finite, so the ideals are closed
under finite unions.

For~\textup{(2)}, reverse an adjacency walk:
\[
(X_0,\ldots,X_k)
\longmapsto
(X_k,\ldots,X_0).
\]
Since
\[
O(X_{r-1},X_r)=O(X_r,X_{r-1}),
\]
the same overlap regions, taken in reverse order, witness the reversed walk.
The resulting support is unchanged.

Statement~\textup{(3)} follows from
\[
\bigcup_{k=1}^{K}
\mathcal O^{(k)}(C_i,C_j)
\subseteq
\bigcup_{k=1}^{K+1}
\mathcal O^{(k)}(C_i,C_j).
\]

For~\textup{(4)}, if a length-\(k\) adjacency walk exists, each of its
adjacent pairs has at least one overlap witness. Choosing one witness for each
step produces a nonempty member of
\(\mathcal O^{(k)}(C_i,C_j)\), which in turn gives a nonempty member of
\(I^{(k)}_{ij}\).

Conversely, if \(I^{(k)}_{ij}\) contains a nonempty set, its definition
requires at least one member of
\(\mathcal O^{(k)}(C_i,C_j)\), and hence at least one length-\(k\)
adjacency walk.

For~\textup{(5)}, a length-one walk from \(C_i\) to \(C_j\) consists of the
pair
\[
(C_i,C_j),
\]
and its support consists of one element of
\[
O(C_i,C_j).
\]
Hence
\[
\mathcal O^{(1)}(C_i,C_j)=O(C_i,C_j),
\]
and the generated ideals coincide.
\end{proof}

\begin{proposition}[Terminal padding]\label{prop:terminal-padding}
Under the convention that adjacency walks may repeat concepts,
\[
I^{(k)}_{ij}\subseteq I^{(k+1)}_{ij}
\qquad
\text{for every }k\geq1.
\]
\end{proposition}

\begin{proof}
Let
\[
S=S(P;W_1,\ldots,W_k)
\in\mathcal O^{(k)}(C_i,C_j).
\]
Choose a witness pair
\[
(A_k,B_k)
\]
for the final overlap \(W_k\), where
\[
A_k\in N(X_{k-1}),
\qquad
B_k\in N(C_j),
\qquad
W_k=A_k\cap B_k.
\]

Extend the walk by repeating its terminal concept:
\[
P'=(X_0,\ldots,X_k,C_j).
\]
Since
\[
B_k=B_k\cap B_k\in O(C_j,C_j),
\]
it may be used as the final overlap witness. The resulting support is
\[
S'=S\cup B_k
\in\mathcal O^{(k+1)}(C_i,C_j),
\]
and
\[
S\subseteq S'.
\]

Applying this construction to every support in a finite generating family
shows that every element of \(I^{(k)}_{ij}\) belongs to
\(I^{(k+1)}_{ij}\).
\end{proof}

\begin{remark}
Terminal padding does not imply
\[
\mathcal O^{(k)}(C_i,C_j)
\subseteq
\mathcal O^{(k+1)}(C_i,C_j).
\]
It proves only that each \(k\)-step support is contained in some
\((k+1)\)-step support.

For interpretations concerned with the minimum number of genuine mediating
steps, the cumulative ideals \(I^{(\leq K)}_{ij}\) are preferable to padded
exact-step ideals.
\end{remark}
\newpage

\begin{definition}[Concept space]\label{def:concept-space}
A \emph{concept space} is a triple
\[
(C,\mathcal F,N),
\]
where:

\begin{enumerate}[label=\arabic*.]

\item \(C\) is a nonempty set of concepts;

\item \(\mathcal F\) is a feasible family on \(C\);

\item
\[
N:C\longrightarrow\mathcal P(\mathcal F)
\]
is a neighbourhood assignment satisfying
\textup{(N0)} and \textup{(N$\downarrow$)}.

\end{enumerate}

The associated external topology is
\[
\mathcal T
=
\tau_{\mathcal F}
=
\left\{
\bigcup\mathcal A:
\mathcal A\subseteq\mathcal F
\right\}.
\]
Equivalently,
\[
\mathcal T
=
\left\{
U\subseteq C:
\forall P\in U\
\exists V\in N(P)
\text{ such that }V\subseteq U
\right\}.
\]

For \(D\subseteq C\), its internal topology is the subspace topology
\[
\mathcal T|_D
=
\left\{
U\cap D:
U\in\mathcal T
\right\}.
\]
\end{definition}

\begin{remark}
The finite-intersection property for neighbourhoods follows from the feasible
family axioms and \textup{(N$\downarrow$)}; it is therefore not imposed as a
separate axiom.

When a particular finite core \(A\subseteq C\) requires a common
neighbourhood, we impose the local condition
\[
N(A)\neq\varnothing.
\]
The global condition \textup{\((N_{\mathrm{fin}})\)} is not assumed, since it
would make every two nonempty finite subsets adjacent.
\end{remark}

\newpage

\section{Candidate emergent regions and internal topology}

Throughout this section, let
\[
(C,\mathcal F,N)
\]
be a concept space with external topology \(\mathcal T\). Let
\(C_i,C_j\in C\) be distinct concepts and let
\[
p=(R,U,V_i,V_j)
\]
be an interaction profile. Its uncovered remainder is
\[
R_E(p)
:=
\operatorname{Int}_{\mathcal T}\bigl(U\setminus(V_i\cup V_j)\bigr).
\]

\subsection{Finite cores and witnesses}

\begin{definition}\label{ccer:def:core}
A \emph{finite core for \(C_i\) relative to \(p\)} is a finite nonempty set
\[
K_i\subseteq U\cap V_i
\]
such that
\[
N\bigl(K_i\cup\{C_i\}\bigr)\neq\varnothing.
\]
A finite core for \(C_j\) is defined analogously. A pair \((K_i,K_j)\) of
such cores is called an \emph{admissible core pair}.
\end{definition}

\begin{remark}
The condition
\[
N\bigl(K_i\cup\{C_i\}\bigr)\neq\varnothing
\]
is equivalent to
\[
N(K_i)\cap N(C_i)\neq\varnothing.
\]
Only this local common-neighbourhood condition is used; no global
\(N_{\mathrm{fin}}\) assumption is required.
\end{remark}

\begin{definition}\label{ccer:def:witness}
Let \((K_i,K_j)\) be an admissible core pair. An \emph{anchored witness} for
\(K_i\) is a set
\[
W_i\in N\bigl(K_i\cup\{C_i\}\bigr),
\]
and an anchored witness for \(K_j\) is a set
\[
W_j\in N\bigl(K_j\cup\{C_j\}\bigr).
\]
\end{definition}

\begin{lemma}\label{ccer:lem:contained}
If \(W_i\subseteq V_i\) or \(W_j\subseteq V_j\), then
\[
W_i\cap W_j\cap R_E(p)=\varnothing.
\]
\end{lemma}

\begin{proof}
If \(W_i\subseteq V_i\), then
\[
W_i\cap R_E(p)
\subseteq
V_i\cap R_E(p)
=
\varnothing.
\]
Hence \(W_i\cap W_j\cap R_E(p)=\varnothing\). The argument for
\(W_j\subseteq V_j\) is the same.
\end{proof}

\begin{remark}
A productive witness cannot be confined to its selected parent region. It
must extend far enough that the two witnesses meet inside the uncovered
remainder.
\end{remark}

\subsection{Candidate emergent regions}

\begin{definition}\label{ccer:def:productive}
Let \((K_i,K_j)\) be an admissible core pair. A pair of anchored witnesses
\((W_i,W_j)\) is called \emph{productive} when
\[
W_i\cap W_j\cap R_E(p)\neq\varnothing.
\]
For every productive pair, define
\[
E(W_i,W_j;p)
:=
W_i\cap W_j\cap R_E(p).
\]
The set \(E(W_i,W_j;p)\) is called a \emph{candidate emergent region}.
\end{definition}

\begin{remark}
No additional interior operator is needed. Neighbourhood regions are open,
\(R_E(p)\) is open, and finite intersections of open sets are open. Hence
\[
\operatorname{Int}_{\mathcal T}
\bigl(W_i\cap W_j\cap R_E(p)\bigr)
=
W_i\cap W_j\cap R_E(p).
\]
\end{remark}

\begin{definition}\label{ccer:def:admission}
The profile \(p\) \emph{admits CCER} if there exists an admissible core pair
and a productive witness pair for that core pair.

The profile \(p\) \emph{robustly admits CCER} if it admits CCER and, in addition, every admissible core pair
has at least one productive witness pair. Thus robust admission is never satisfied
vacuously when no admissible core pair exists.
\end{definition}

\begin{remark}
The construction is generally non-deterministic. Different profiles, cores,
or witnesses may produce different candidate emergent regions. CCER therefore
specifies structurally permitted candidates rather than a unique outcome.
\end{remark}

\subsection{Basic properties}

\begin{proposition}\label{ccer:prop:basic}
Let
\[
E=E(W_i,W_j;p)
\]
be a candidate emergent region. Then:
\begin{enumerate}
\item \(E\) is a nonempty open subset of \(U\);
\item \(E\cap(V_i\cup V_j)=\varnothing\);
\item \(E\cap\{C_i,C_j\}=\varnothing\);
\item \(E\subseteq W_i\cap W_j\);
\item \(W_i\not\subseteq V_i\) and \(W_j\not\subseteq V_j\);
\item \(E\in I_{ij}\);
\item \(E\in E_{ij}\);
\item every open set \(O\subseteq E\) belongs to \(I_{ij}\cap E_{ij}\).
\end{enumerate}
\end{proposition}

\begin{proof}
The sets \(W_i\), \(W_j\), and \(R_E(p)\) are open, so their intersection
\(E\) is open. It is nonempty by productivity, and it is contained in \(U\)
because \(R_E(p)\subseteq U\).

Since
\[
R_E(p)\cap(V_i\cup V_j)=\varnothing,
\]
we have \(E\cap(V_i\cup V_j)=\varnothing\). Moreover,
\(C_i\in V_i\) and \(C_j\in V_j\), so
\[
E\cap\{C_i,C_j\}=\varnothing.
\]
The inclusion \(E\subseteq W_i\cap W_j\) is immediate.

If \(W_i\subseteq V_i\), then Lemma~\ref{ccer:lem:contained} would give
\(E=\varnothing\), a contradiction. Hence \(W_i\not\subseteq V_i\), and
similarly \(W_j\not\subseteq V_j\).

Because \(E\neq\varnothing\), the intersection \(W_i\cap W_j\) is a
nonempty overlap of a neighbourhood of \(C_i\) with a neighbourhood of
\(C_j\). Thus \(W_i\cap W_j\in O(C_i,C_j)\). Since \(E\) is open and
\[
E\subseteq W_i\cap W_j,
\]
downward closure of \(I_{ij}\) gives \(E\in I_{ij}\).

Also, \(E\subseteq R_E(p)\), and \(R_E(p)\) is one of the generators of the
remainder ideal. Hence \(E\in E_{ij}\). The final statement follows from
downward closure of both ideals.
\end{proof}

\begin{remark}
Every candidate emergent region satisfies
\[
E\in I_{ij}\cap E_{ij}.
\]
Thus it is simultaneously supported by an overlap channel and localized
inside uncovered remainder structure. The converse need not hold.
\end{remark}

\subsection{Internal topology}

\begin{definition}\label{ccer:def:internal}
Let \(E\) be a candidate emergent region. Its \emph{internal topology} is the
subspace topology
\[
\mathcal T_E
:=
\mathcal T|_E
=
\{U\cap E:U\in\mathcal T\}.
\]
\end{definition}

\begin{proposition}\label{ccer:prop:internal}
For every candidate emergent region \(E\),
\[
\mathcal T_E
=
\{O\in\mathcal T:O\subseteq E\}.
\]
Moreover,
\[
\mathcal T_E\subseteq I_{ij}\cap E_{ij}.
\]
\end{proposition}

\begin{proof}
If \(U\in\mathcal T\), then \(U\cap E\) is open in \(C\) because \(E\) is
open, and \(U\cap E\subseteq E\). Conversely, if \(O\in\mathcal T\) and
\(O\subseteq E\), then \(O=O\cap E\), so \(O\in\mathcal T_E\). This proves
the equality. The final inclusion follows from
Proposition~\ref{ccer:prop:basic}.
\end{proof}

\begin{remark}
The topology \(\mathcal T_E\) is the canonical topology inherited from the
ambient concept space. Equivalently, it is the coarsest topology on \(E\) for
which the inclusion map
\[
E\hookrightarrow C
\]
is continuous.
\end{remark}

\begin{proposition}\label{ccer:prop:basis}
Let \(E\subseteq W_i\) be a candidate emergent region, and let
\(\mathcal B_i\) be a basis for the subspace topology on \(W_i\). Then
\[
\mathcal B_i|_E
:=
\{B\cap E:B\in\mathcal B_i,\ B\cap E\neq\varnothing\}
\]
is a basis for \(\mathcal T_E\). The analogous statement holds for a basis
of the subspace topology on \(W_j\).
\end{proposition}

\begin{proof}
Because \(E\subseteq W_i\), transitivity of subspace topologies gives
\[
(\mathcal T|_{W_i})|_E
=
\mathcal T|_E.
\]
The standard basis-restriction theorem therefore applies.
\end{proof}

\subsection{Reification and interpretation}

\begin{remark}
A candidate emergent region
\[
E\subseteq C^{(n)}
\]
is not thereby treated as an atomic concept of \(C^{(n)}\). It is a region
(a subset of the stage concept set). If it is selected for reification, a
later-stage map may assign to it a new atomic concept
\[
\rho_n(E)=e_E\in C^{(n+1)}.
\]
The earlier region and its inherited topology may be retained as a realization
of that atom:
\[
\operatorname{Real}_{n+1}(e_E)
=
(E,\mathcal T^{(n)}|_E).
\]
Thus the region \(E\) and the later-stage atom \(e_E\) occupy distinct formal
roles in HDCS, and reification does not identify them. If literal set-theoretic
separation between region objects and concept-atoms is desired, one may use
tagged representatives without changing the construction.
\end{remark}

\begin{remark}
CCER supplies candidate regions only. It does not by itself determine whether
a candidate is novel, historically realized, stable, dominant, or suitable
for reification. Those conclusions require additional dynamical or
application-specific criteria.
\end{remark}

\begin{remark}
The remainder
\[
R_E(p)
=
\operatorname{Int}_{\mathcal T}
\bigl(U\setminus(V_i\cup V_j)\bigr)
\]
is mathematically the open portion of \(U\) not covered by the selected
parent neighbourhoods. Interpreting it as failure, contradiction, tension, or
historical insufficiency requires separate domain-specific justification.
\end{remark}

\section{Stages and Heraclitean development}
\label{sec:hdcs-stages}

The preceding constructions describe the structure of a concept space at a
single stage. We now introduce a stagewise framework in which concepts,
feasible regions, neighbourhood assignments, and selected emergent structures
may change over time.

Throughout this section, a stage is indexed by \(i\). At stage \(i\) we write
\[
\mathfrak C^{(i)}
:=
\bigl(
C^{(i)},
\mathcal F^{(i)},
N^{(i)}
\bigr)
\]
for the corresponding concept space.

Its external topology is
\[
\mathcal T^{(i)}
:=
\tau_{\mathcal F^{(i)}}
=
\left\{
\bigcup\mathcal A
\;\middle|\;
\mathcal A\subseteq\mathcal F^{(i)}
\right\}.
\]

Equivalently,
\[
\mathcal T^{(i)}
=
\left\{
U\subseteq C^{(i)}
\;\middle|\;
\forall x\in U\
\exists V\in N^{(i)}(x)
\text{ such that }V\subseteq U
\right\}.
\]

All constructions introduced in the preceding sections may be applied
stagewise by replacing
\[
(C,\mathcal F,N,\mathcal T)
\]
with
\[
\bigl(
C^{(i)},
\mathcal F^{(i)},
N^{(i)},
\mathcal T^{(i)}
\bigr).
\]

\subsection{Stagewise adjacency, overlap, and restriction}

For \(C_a,C_b\in C^{(i)}\), write
\[
C_a\sim_i C_b
\]
when
\[
\exists U\in N^{(i)}(C_a),\
\exists V\in N^{(i)}(C_b)
\quad\text{such that}\quad
U\cap V\neq\varnothing.
\]

The corresponding overlap family is
\[
O_i(C_a,C_b)
:=
\left\{
U\cap V
\;\middle|\;
U\in N^{(i)}(C_a),\
V\in N^{(i)}(C_b),\
U\cap V\neq\varnothing
\right\}.
\]

Every member of \(O_i(C_a,C_b)\) is a nonempty feasible region and therefore
a nonempty open set in \(\mathcal T^{(i)}\).

\begin{definition}[Restriction of a stage]
\label{def:hdcs-stage-restriction}
Let \(U\subseteq C^{(i)}\) be nonempty. Define the restricted feasible family
by
\[
\mathcal F^{(i)}|_U
:=
\left\{
F\cap U
\;\middle|\;
F\in\mathcal F^{(i)},\
F\cap U\neq\varnothing
\right\}.
\]

For \(x\in U\), define the restricted neighbourhood assignment by
\[
N^{(i)}|_U(x)
:=
\left\{
W\cap U
\;\middle|\;
W\in N^{(i)}(x)
\right\}.
\]

The restricted stage is
\[
\mathfrak C^{(i)}|_U
:=
\bigl(
U,
\mathcal F^{(i)}|_U,
N^{(i)}|_U
\bigr).
\]
\end{definition}

\begin{proposition}[Stage restriction]
\label{prop:hdcs-stage-restriction}
The restricted structure
\[
\mathfrak C^{(i)}|_U
\]
is a concept space on \(U\). Its external topology is precisely the subspace
topology
\[
\mathcal T^{(i)}|_U
=
\left\{
G\cap U
\;\middle|\;
G\in\mathcal T^{(i)}
\right\}.
\]
\end{proposition}

\begin{proof}
First, \(\mathcal F^{(i)}|_U\) covers \(U\). Indeed, if \(x\in U\), then
there exists \(F\in\mathcal F^{(i)}\) with \(x\in F\), and hence
\[
x\in F\cap U\in\mathcal F^{(i)}|_U.
\]

The empty set is excluded by definition. If
\[
F_1\cap U,\ F_2\cap U
\in
\mathcal F^{(i)}|_U
\]
have nonempty intersection, then
\[
F_1\cap F_2\neq\varnothing,
\]
so
\[
F_1\cap F_2\in\mathcal F^{(i)}
\]
by \textup{(F\(\cap\))}. Therefore
\[
(F_1\cap U)\cap(F_2\cap U)
=
(F_1\cap F_2)\cap U
\in
\mathcal F^{(i)}|_U.
\]

Now let \(x\in U\). Since \(N^{(i)}(x)\neq\varnothing\), the family
\(N^{(i)}|_U(x)\) is nonempty. Moreover, every
\[
W\cap U\in N^{(i)}|_U(x)
\]
contains \(x\).

To verify downward closure, suppose
\[
W\cap U\in N^{(i)}|_U(x)
\]
and
\[
F\cap U\in\mathcal F^{(i)}|_U
\]
satisfy
\[
x\in F\cap U\subseteq W\cap U.
\]
Then
\[
x\in F\cap W\subseteq W.
\]
Since \(F,W\in\mathcal F^{(i)}\), the nonempty intersection
\(F\cap W\) lies in \(\mathcal F^{(i)}\), and
\textup{(N\(\downarrow\))} gives
\[
F\cap W\in N^{(i)}(x).
\]
Furthermore,
\[
(F\cap W)\cap U
=
F\cap U.
\]
Hence
\[
F\cap U\in N^{(i)}|_U(x).
\]

Thus \(\mathfrak C^{(i)}|_U\) is a concept space.

Finally, since \(\mathcal F^{(i)}\) is a basis for
\(\mathcal T^{(i)}\), the family
\[
\mathcal F^{(i)}|_U
\]
is the standard restricted basis for the subspace topology on \(U\).
Therefore the external topology of the restricted concept space is
\[
\mathcal T^{(i)}|_U.
\]
\end{proof}

\subsection{Stagewise channels and profiles}

For \(C_a,C_b\in C^{(i)}\), define the stage-\(i\) channel ideal by
\[
I^{(i)}_{ab}
:=
\left\{
W\in\mathcal T^{(i)}
\;\middle|\;
\begin{array}{l}
\text{there exists a finite family }
\mathcal A\subseteq O_i(C_a,C_b)\\[1mm]
\text{such that }
W\subseteq\displaystyle\bigcup_{R\in\mathcal A}R
\end{array}
\right\}.
\]

The empty finite family is permitted, with
\[
\bigcup\varnothing=\varnothing.
\]

No interior operator is required, since every overlap region is already open.

For \(R\in O_i(C_a,C_b)\), define
\[
N^{(i)}_R(C_a)
:=
\left\{
V\in N^{(i)}(C_a)
\;\middle|\;
R\subseteq V
\right\},
\]
and similarly for \(C_b\).

\begin{definition}[Stage profile]
\label{def:hdcs-stage-profile}
Let
\[
R\in O_i(C_a,C_b).
\]
A \emph{stage-\(i\) interaction profile} is a tuple
\[
p^{(i)}
=
(R,U,V_a,V_b)
\]
such that
\[
U\in\mathcal T^{(i)},
\qquad
R\subseteq U,
\]
and
\[
V_a\in N^{(i)}_R(C_a),
\qquad
V_b\in N^{(i)}_R(C_b).
\]

Its \emph{stage remainder} is
\[
R_E^{(i)}(p^{(i)})
:=
\operatorname{Int}_{\mathcal T^{(i)}}
\bigl(
U\setminus(V_a\cup V_b)
\bigr).
\]
\end{definition}

The remainder is an open subset of \(U\) and satisfies
\[
R_E^{(i)}(p^{(i)})
\cap
(V_a\cup V_b)
=
\varnothing.
\]

It records the open portion of the selected interaction context not covered
by the two selected parent neighbourhoods.

\begin{remark}[Aggregate remainder support]
For fixed \(C_a,C_b\), one may form
\[
\mathcal R^{(i)}_{ab}
:=
\bigcup_{p^{(i)}}
R_E^{(i)}(p^{(i)}),
\]
where the union ranges over all stage-\(i\) profiles for the pair.

This set is open. It records all points that occur in the remainder of at
least one profile. It need not, however, constitute one coherent interaction
context, since distinct profiles may represent different structural or
domain-specific situations.
\end{remark}

\subsection{Stagewise candidate emergence}

The CCER construction of the preceding section is applied independently at
each stage.

Let
\[
p^{(i)}
=
(R,U,V_a,V_b)
\]
be a stage-\(i\) profile.

An admissible finite core pair consists of finite nonempty sets
\[
K_a\subseteq U\cap V_a,
\qquad
K_b\subseteq U\cap V_b
\]
such that
\[
N^{(i)}
\bigl(
K_a\cup\{C_a\}
\bigr)
\neq\varnothing,
\]
and
\[
N^{(i)}
\bigl(
K_b\cup\{C_b\}
\bigr)
\neq\varnothing.
\]

Choose anchored witnesses
\[
W_a
\in
N^{(i)}
\bigl(
K_a\cup\{C_a\}
\bigr),
\qquad
W_b
\in
N^{(i)}
\bigl(
K_b\cup\{C_b\}
\bigr).
\]

The pair \((W_a,W_b)\) is called productive when
\[
W_a\cap W_b\cap R_E^{(i)}(p^{(i)})
\neq\varnothing.
\]

\begin{definition}[Stagewise candidate emergent region]
\label{def:hdcs-stage-candidate}
For a productive witness pair, define
\[
E^{(i)}
=
E^{(i)}
(W_a,W_b;p^{(i)})
:=
W_a
\cap
W_b
\cap
R_E^{(i)}(p^{(i)}).
\]

The nonempty open set \(E^{(i)}\) is called a
\emph{stage-\(i\) candidate emergent region}.
\end{definition}

No condition
\[
W_a\subseteq V_a
\qquad\text{or}\qquad
W_b\subseteq V_b
\]
is imposed. Indeed, either containment would force
\[
E^{(i)}=\varnothing.
\]

Thus productive witnesses must extend beyond their selected parent
neighbourhoods into the uncovered remainder.

\begin{proposition}[Stagewise candidate emergence]
\label{prop:hdcs-stage-candidate}
Let
\[
E^{(i)}
=
W_a\cap W_b\cap R_E^{(i)}(p^{(i)})
\]
be a stage-\(i\) candidate emergent region. Then:

\begin{enumerate}[label=\arabic*.]

\item
\[
E^{(i)}
\]
is a nonempty open subset of \(U\).

\item
\[
E^{(i)}
\cap
(V_a\cup V_b)
=
\varnothing.
\]

\item
\[
C_a\notin E^{(i)},
\qquad
C_b\notin E^{(i)}.
\]

\item
\[
E^{(i)}
\subseteq
W_a\cap W_b.
\]

\item
\[
E^{(i)}
\in
I^{(i)}_{ab}.
\]

\item Every open subset
\[
G\subseteq E^{(i)}
\]
belongs to \(I^{(i)}_{ab}\).

\item The internal topology of \(E^{(i)}\) is
\[
\mathcal T^{(i)}|_{E^{(i)}}
=
\left\{
G\in\mathcal T^{(i)}
\;\middle|\;
G\subseteq E^{(i)}
\right\}.
\]

\end{enumerate}
\end{proposition}

\begin{proof}
The sets
\[
W_a,\qquad
W_b,\qquad
R_E^{(i)}(p^{(i)})
\]
are open. Hence their finite intersection \(E^{(i)}\) is open, and it is
nonempty by productivity.

Since
\[
R_E^{(i)}(p^{(i)})\subseteq U,
\]
we have
\[
E^{(i)}\subseteq U.
\]

Also,
\[
R_E^{(i)}(p^{(i)})
\cap
(V_a\cup V_b)
=
\varnothing,
\]
so
\[
E^{(i)}
\cap
(V_a\cup V_b)
=
\varnothing.
\]

Because
\[
C_a\in V_a,
\qquad
C_b\in V_b,
\]
it follows that
\[
C_a\notin E^{(i)}
\qquad\text{and}\qquad
C_b\notin E^{(i)}.
\]

Since \(E^{(i)}\neq\varnothing\),
\[
W_a\cap W_b\neq\varnothing.
\]
Moreover,
\[
W_a\in N^{(i)}(C_a),
\qquad
W_b\in N^{(i)}(C_b),
\]
so
\[
W_a\cap W_b
\in
O_i(C_a,C_b).
\]
As \(E^{(i)}\) is open and
\[
E^{(i)}
\subseteq
W_a\cap W_b,
\]
downward closure gives
\[
E^{(i)}
\in
I^{(i)}_{ab}.
\]
The same argument applies to every open
\[
G\subseteq E^{(i)}.
\]

Finally, since \(E^{(i)}\) is itself open,
\[
\mathcal T^{(i)}|_{E^{(i)}}
=
\left\{
G\in\mathcal T^{(i)}
\;\middle|\;
G\subseteq E^{(i)}
\right\}.
\]
\end{proof}

\begin{remark}
Stagewise CCER produces candidate regions rather than automatically producing
new atomic concepts. The decision to stabilize or reify one of these regions
belongs to the transition between stages.
\end{remark}

\subsection{Reification and persistence}

Let
\[
\mathscr E_i
\]
denote the family of all stage-\(i\) candidate emergent regions, and let
\[
\mathscr E_i^{\mathrm{sel}}
\subseteq
\mathscr E_i
\]
be the subfamily selected for reification.

\begin{definition}[Reification map]
\label{def:hdcs-reification}
A \emph{stage-\(i\) reification map} is an injective map
\[
\rho_i:
\mathscr E_i^{\mathrm{sel}}
\longrightarrow
C^{(i+1)}
\]
such that, for every
\[
E\in\mathscr E_i^{\mathrm{sel}},
\]
the element
\[
e_E:=\rho_i(E)
\]
is regarded as the next-stage atomic concept representing the selected
stage-\(i\) region \(E\).
\end{definition}

The region \(E\) and the atom \(e_E\) are not identified as the same
mathematical object:
\[
E\subseteq C^{(i)},
\qquad
e_E\in C^{(i+1)}.
\]

For the HDCS stage transitions defined below, the internal structure of the
earlier region is retained by the realization record
\[
\operatorname{Real}_{i+1}(e_E)
:=
\bigl(
E,\mathcal T^{(i)}|_E
\bigr).
\]

Thus reification converts a structured region at stage \(i\) into an atomic
concept at stage \(i+1\), while retaining the earlier region as its recorded
internal realization. Injectivity ensures that this single-valued provenance
record is not ambiguous: distinct selected candidate regions are not assigned
the same reified atom within one reification map.

\subsection{Carry maps}

Not every concept need have a successor at the next stage.

\begin{definition}[Carry map]
\label{def:hdcs-carry}
For each stage \(i\) having a successor, let
\[
D_i\subseteq C^{(i)}
\]
be the set of concepts that persist into stage \(i+1\).

A \emph{carry map} is an injective map
\[
\sigma_i:
D_i
\longrightarrow
C^{(i+1)}.
\]

For \(x\in D_i\), the concept
\[
\sigma_i(x)
\]
is interpreted as the stage-\((i+1)\) successor of \(x\).
\end{definition}

No successor is required for concepts outside \(D_i\). Likewise, concepts in
\(C^{(i+1)}\) may have no predecessor, for example when they arise through
reification.

Injectivity is part of the meaning of carry: \(\sigma_i\) records persistence
of an already distinguished identity. If two distinct concepts
\[
x\neq y
\]
contribute to one later concept \(z\), that event is not represented by
\(\sigma_i(x)=\sigma_i(y)=z\). Instead, \(z\) is treated as a new concept,
typically through reification of a candidate region whose provenance records
both earlier concepts. The earlier concepts may then persist separately or
may disappear from the next stage.

\begin{remark}[Identity without literal equality]
If a concept is literally retained as the same set-theoretic object in both
stages, one may impose
\[
\sigma_i(x)=x.
\]
More generally, \(\sigma_i(x)\) records the identified successor without
requiring literal equality of stage elements. Injectivity preserves the
distinction between identity-through-change and convergence of distinct
histories.
\end{remark}

\subsection{Heraclitean Dialectical Concept Spaces}

Let \(I\) be either
\[
\mathbb N
\]
or a finite initial segment
\[
\{0,1,\ldots,m\}.
\]

This ensures that whenever \(i\in I\) is not terminal, the successor
\(i+1\) is also in \(I\).

\begin{definition}[Heraclitean Dialectical Concept Space]
\label{def:hdcs-dynamic}
A \emph{Heraclitean Dialectical Concept Space} (HDCS) consists of:

\begin{enumerate}[label=\arabic*.]

\item a family of concept spaces
\[
\bigl(
\mathfrak C^{(i)}
\bigr)_{i\in I},
\qquad
\mathfrak C^{(i)}
=
\bigl(
C^{(i)},
\mathcal F^{(i)},
N^{(i)}
\bigr);
\]

\item for every nonterminal \(i\in I\), a carry domain
\[
D_i\subseteq C^{(i)}
\]
and injective carry map
\[
\sigma_i:
D_i\to C^{(i+1)};
\]

\item for every nonterminal \(i\in I\), a selected family
\[
\mathscr E_i^{\mathrm{sel}}
\]
of candidate emergent regions and an injective reification map
\[
\rho_i:
\mathscr E_i^{\mathrm{sel}}
\to
C^{(i+1)};
\]

\item separation of persistence and reification:
\[
\sigma_i(D_i)
\cap
\rho_i\bigl(\mathscr E_i^{\mathrm{sel}}\bigr)
=
\varnothing;
\]
thus no next-stage atom is simultaneously declared to be the carried identity
of an earlier concept and a newly reified concept;

\item stagewise event data
\[
\mathrm{Events}^{(i)}
\]
recording selected emergences, reifications, carry decisions, and any other
declared structural changes; when structural locality is invoked, the
elementary event set \(\mathcal V_i\) carried by this data is equipped with a
support assignment
\[
\operatorname{Supp}_i:\mathcal V_i\to\bigl(\mathcal T^{(i)}\setminus\{\varnothing\}\bigr);
\]

\item a transition mechanism
\[
\Phi_i
\]
satisfying
\[
\mathfrak C^{(i+1)}
=
\Phi_i
\bigl(
\mathfrak C^{(i)},
\mathrm{Events}^{(i)}
\bigr).
\]

\end{enumerate}

Each output
\[
\mathfrak C^{(i+1)}
\]
is required to satisfy the axioms of a concept space.
\end{definition}

The transition mechanism \(\Phi_i\) is not assumed to be a unique or fixed
algorithm. It specifies a constrained stage transition. In particular, the
present framework does not prescribe a universal procedure for constructing
\[
\mathcal F^{(i+1)}
\quad\text{or}\quad
N^{(i+1)}
\]
from the previous stage.

This avoids imposing transport rules that need not preserve the feasible-family
or neighbourhood axioms.

\subsection{Heraclitean flux conditions}

The following conditions describe the intended coherence of stagewise
development.

\begin{description}

\item[\textup{(H1) Nonstationarity of a change model.}]
A transition from stage \(i\) to stage \(i+1\) is called \emph{stationary} when
all of the following hold:
\[
\mathfrak C^{(i+1)}=\mathfrak C^{(i)},
\qquad
D_i=C^{(i)},
\qquad
\sigma_i=\operatorname{id}_{C^{(i)}},
\qquad
\mathscr E_i^{\mathrm{sel}}=\varnothing.
\]
Stationary intervals are permitted. However, an HDCS presented as a nontrivial
model of change is required to contain at least one nonstationary transition.

Thus genuine change may consist of alteration of the concept set, feasible
family, neighbourhood assignment, persistence structure, reification, or some
combination of these. A fully stationary sequence remains a mathematically
valid degenerate HDCS, but it is not a nontrivial Heraclitean development.

\item[\textup{(H2) Declared structural support.}]
When locality is invoked, write \(\mathcal V_i\) for the set of elementary
declared events contained in \(\mathrm{Events}^{(i)}\). Each such event is
assigned an open structural support by a map
\[
\operatorname{Supp}_i:
\mathcal V_i
\longrightarrow
\mathcal T^{(i)}\setminus\{\varnothing\}.
\]

The support is required to be nonempty because an elementary declared event is
assumed to occur somewhere in the stage structure. This support assignment localizes the event within the stage. The topology
does not by itself impose a law by which effects propagate outside
\(\operatorname{Supp}_i(\varepsilon)\). If an application claims an effect
outside that support, the relevant direct or multi-step channel relations must
be supplied explicitly as additional event data. Thus locality is represented
formally, while propagation remains application-specific.

\item[\textup{(H3) Reification with provenance.}]
If
\[
E\in\mathscr E_i^{\mathrm{sel}},
\]
then
\[
e_E=\rho_i(E)
\in C^{(i+1)}
\]
is the corresponding reified concept, and the stage-\(i\) realization
\[
\bigl(
E,\mathcal T^{(i)}|_E
\bigr)
\]
is retained as its provenance data.

\item[\textup{(H4) Provenance.}]
Every reified concept records the data from which its candidate region was
constructed.

For
\[
e_E=\rho_i(E),
\]
one may record
\[
\operatorname{Parents}(e_E)
=
\{C_a,C_b\},
\]
\[
\operatorname{Profile}(e_E)
=
p^{(i)},
\]
\[
\operatorname{Cores}(e_E)
=
(K_a,K_b),
\]
and
\[
\operatorname{Witnesses}(e_E)
=
(W_a,W_b).
\]

\item[\textup{(H5) Identity through change.}]
Concepts selected for persistence are tracked by the injective carry map
\[
\sigma_i:
D_i\to C^{(i+1)}.
\]
Concepts outside \(D_i\) need not have successors, while reified or otherwise
new concepts may appear without predecessors. Because carry maps are
injective, two distinct earlier identities cannot be collapsed into one later
identity by persistence alone. Moreover,
\[
\sigma_i(D_i)
\cap
\rho_i\bigl(\mathscr E_i^{\mathrm{sel}}\bigr)
=
\varnothing,
\]
so a reified atom is not simultaneously treated as the carried identity of an
earlier concept.

\end{description}

\begin{remark}
The conditions above distinguish three logically separate possibilities:

\[
\text{persistence},
\qquad
\text{disappearance},
\qquad
\text{emergence}.
\]

In particular, the appearance of a new concept does not imply the disappearance
of its historical predecessors. Earlier concepts may persist through the carry
map while a newly reified concept is adjoined at the next stage. The image
separation condition ensures that these are genuinely different operations:
reification cannot silently double as persistence.
\end{remark}

\subsection{The transition mechanism}

\begin{definition}[Abstract evolution mechanism]
\label{def:hdcs-evolution}
For each nonterminal stage \(i\), the evolution mechanism
\[
\Phi_i
\]
takes the current concept space together with the declared event data and
produces the next valid concept space:
\[
\Phi_i:
\bigl(
\mathfrak C^{(i)},
\mathrm{Events}^{(i)}
\bigr)
\longmapsto
\mathfrak C^{(i+1)}.
\]

The event data may include:

\begin{itemize}

\item concepts selected for persistence through \(\sigma_i\);

\item candidate emergent regions selected for reification through \(\rho_i\);

\item provenance information associated with those candidates;

\item declared structural supports \(\operatorname{Supp}_i(\varepsilon)\)
when locality is being modelled;

\item declared additions, removals, or modifications of feasible and
neighbourhood structure.

\end{itemize}

The output is required to satisfy the feasible-family and neighbourhood axioms
of a concept space.
\end{definition}

\begin{remark}
No claim is made that every admissible transition is uniquely determined by
the previous stage.

The role of \(\Phi_i\) is to delimit structurally permissible stage
transformations, while domain-specific applications may impose additional
dynamical, empirical, or algorithmic rules.
\end{remark}

\subsection{The HDCS history space}

The carry maps allow persistent concepts to be tracked across stages.

Let
\[
\mathfrak U
:=
\bigsqcup_{i\in I}
\bigl(
\{i\}\times C^{(i)}
\bigr)
\]
be the disjoint union of all stage concept sets.

Define \(\approx\) to be the smallest equivalence relation on
\(\mathfrak U\) such that
\[
(i,x)
\approx
(i+1,\sigma_i(x))
\]
whenever
\[
x\in D_i.
\]

\begin{definition}[HDCS history space]
\label{def:hdcs-history}
The \emph{HDCS history set} is
\[
C_{\mathrm{hist}}
:=
\mathfrak U/\!\approx.
\]

For each stage \(i\), define the canonical map
\[
q_i:
C^{(i)}
\longrightarrow
C_{\mathrm{hist}}
\]
by
\[
q_i(x)
:=
[(i,x)]_{\approx}.
\]

The \emph{history topology}
\[
\mathcal T_{\mathrm{hist}}
\]
is the final topology with respect to the family
\[
(q_i)_{i\in I}.
\]
Explicitly,
\[
O\subseteq C_{\mathrm{hist}}
\]
is open if and only if
\[
q_i^{-1}(O)
\in
\mathcal T^{(i)}
\]
for every \(i\in I\).
\end{definition}

Thus
\[
\bigl(
C_{\mathrm{hist}},
\mathcal T_{\mathrm{hist}}
\bigr)
\]
records the stagewise history of concepts under the identifications declared
by the carry maps.

\begin{proposition}[Stage identities remain distinct in the history set]
\label{prop:history-stage-injective}
For every stage \(i\), the canonical map
\[
q_i:C^{(i)}\longrightarrow C_{\mathrm{hist}}
\]
is injective.
\end{proposition}

\begin{proof}
Consider the undirected graph on the timeline set \(\mathfrak U\) whose edges
are precisely the generating carry identifications
\[
(i,x)\;---\;(i+1,\sigma_i(x)).
\]
Because each \(\sigma_i\) is a partial injective map, every timeline vertex has
at most one carry edge to the next stage and at most one carry edge from the
previous stage. Hence each connected component of this graph is a single
trajectory and contains at most one vertex from any fixed stage.

If \(q_i(x)=q_i(y)\), then \((i,x)\) and \((i,y)\) lie in the same equivalence
class, hence in the same connected component of the carry graph. Since that
component contains at most one stage-\(i\) vertex, \(x=y\). Therefore \(q_i\)
is injective.
\end{proof}

\begin{remark}[Convergence is not identity]
The injectivity above is intentional. If two distinct earlier concepts later
contribute to one newly formed concept, HDCS records that event by reification
and provenance rather than by identifying their carry trajectories. Thus a
pattern
\[
x\longrightarrow z\longleftarrow y
\]
with \(x\neq y\) is treated as convergence into a new structure, not as a
retroactive assertion that \(x=y\).

No continuity assumption on the partial carry maps is required for the final
topology itself to exist.
\end{remark}

\begin{remark}[Relation to categorical colimits]
The history-space construction is defined here directly as a quotient of the
disjoint union together with its final topology.

When the carry maps are total continuous maps and form an ordinary sequential
diagram in the category of topological spaces, the corresponding quotient
construction may be identified with the usual topological colimit.

For genuinely partial carry maps, no such categorical identification is
claimed in the present work.
\end{remark}

\subsection{Interpretation}

The stagewise framework separates four distinct operations:

\[
\boxed{
\text{local structural interaction}
}
\]

\[
\Downarrow
\]

\[
\boxed{
\text{candidate emergent region}
}
\]

\[
\Downarrow
\]

\[
\boxed{
\text{selection and reification}
}
\]

\[
\Downarrow
\]

\[
\boxed{
\text{next-stage concept}
}
\]

while pre-existing concepts may independently persist through the carry maps.

Consequently, emergence does not entail replacement. A concept may become
newly available or structurally prominent while earlier concepts continue to
exist at the same or later stages.

The HDCS dynamics therefore represent conceptual development as a combination
of persistence, transformation, disappearance, and emergence rather than as a
simple sequence of mutually exclusive stages.

\section*{Additional topological properties}

The core HDCS construction does not require continuity, convergence,
connectedness, or compactness. Nevertheless, these standard topological
notions provide useful tools for analysing stagewise structure and the
behaviour of selected transition maps.

Throughout this section, let
\[
\mathfrak C^{(i)}
=
\bigl(
C^{(i)},\mathcal F^{(i)},N^{(i)}
\bigr)
\]
be a stage concept space with external topology
\[
\mathcal T^{(i)}
=
\tau_{\mathcal F^{(i)}}.
\]

\subsection*{Continuity across stages}

Let
\[
D_f\subseteq C^{(i)}
\]
and let
\[
f:
D_f
\longrightarrow
C^{(i+1)}
\]
be a transition map defined on \(D_f\).

The domain \(D_f\) is equipped with the subspace topology
\[
\mathcal T^{(i)}|_{D_f}
=
\left\{
U\cap D_f:
U\in\mathcal T^{(i)}
\right\}.
\]

\begin{definition}[Continuity of a partial stage map]
\label{def:hdcs-partial-continuity}
The map
\[
f:D_f\to C^{(i+1)}
\]
is \emph{continuous} if
\[
V\in\mathcal T^{(i+1)}
\quad\Longrightarrow\quad
f^{-1}(V)
\in
\mathcal T^{(i)}|_{D_f}.
\]
\end{definition}

Thus continuity has its standard topological meaning: every open region at
the target stage has an open preimage in the domain of the transition map.

\begin{remark}[Continuity of carry maps]
For a carry map
\[
\sigma_i:
D_i
\longrightarrow
C^{(i+1)},
\]
continuity is an optional regularity condition on the evolution mechanism.
It is not implied by the HDCS axioms.

When imposed, it expresses compatibility between the inherited topology on
persisting concepts and the external topology of the next stage.
\end{remark}

\begin{proposition}[Neighbourhood criterion for continuity]
\label{prop:hdcs-neighbourhood-continuity}
Let
\[
f:D_f\to C^{(i+1)}.
\]
Suppose that for every \(x\in D_f\) and every
\[
V\in N^{(i+1)}(f(x)),
\]
there exists
\[
U\in N^{(i)}(x)
\]
such that
\[
f(U\cap D_f)\subseteq V.
\]
Then \(f\) is continuous.
\end{proposition}

\begin{proof}
Let
\[
O\in\mathcal T^{(i+1)}.
\]
We show that
\[
f^{-1}(O)
\]
is open in \(D_f\).

Let
\[
x\in f^{-1}(O).
\]
Then
\[
f(x)\in O.
\]
Since \(O\) is open, there exists
\[
V\in N^{(i+1)}(f(x))
\]
such that
\[
V\subseteq O.
\]

By hypothesis, there exists
\[
U\in N^{(i)}(x)
\]
such that
\[
f(U\cap D_f)\subseteq V.
\]
Hence
\[
x\in U\cap D_f
\subseteq
f^{-1}(O).
\]

Because \(U\) is open in \(C^{(i)}\), the set
\[
U\cap D_f
\]
is open in the subspace \(D_f\). Thus every point of
\(f^{-1}(O)\) has a subspace-open neighbourhood contained in
\(f^{-1}(O)\), so
\[
f^{-1}(O)\in\mathcal T^{(i)}|_{D_f}.
\]
Therefore \(f\) is continuous.
\end{proof}

\begin{remark}
Because the neighbourhood regions form a basis for the external topology,
the preceding criterion is also necessary for continuity.

Thus continuity may equivalently be expressed locally as
\[
\forall x\in D_f,\
\forall V\in N^{(i+1)}(f(x)),\
\exists U\in N^{(i)}(x)
\]
such that
\[
f(U\cap D_f)\subseteq V.
\]
\end{remark}

\begin{remark}[Continuity and the history space]
Continuity of the carry maps is not required for the HDCS history quotient
or its final topology to exist.

The final topology on
\[
C_{\mathrm{hist}}
\]
is always defined by
\[
O\in\mathcal T_{\mathrm{hist}}
\quad\Longleftrightarrow\quad
q_i^{-1}(O)\in\mathcal T^{(i)}
\quad
\text{for every }i.
\]

If the carry maps are total and continuous, however, the stage sequence may
be regarded as a diagram in the category of topological spaces, and the
history-space construction may then be compared with the corresponding
topological colimit.
\end{remark}

\subsection*{Convergence within a stage}

\begin{definition}[Net convergence in a concept space]
\label{def:hdcs-net-convergence}
Let
\[
(C,\mathcal F,N)
\]
be a concept space with external topology
\[
\mathcal T=\tau_{\mathcal F}.
\]

A net
\[
(x_\alpha)_{\alpha\in A}
\]
in \(C\) is said to \emph{converge} to
\[
x\in C,
\]
written
\[
x_\alpha\longrightarrow x,
\]
if for every
\[
U\in N(x)
\]
there exists
\[
\alpha_0\in A
\]
such that
\[
\alpha\geq\alpha_0
\quad\Longrightarrow\quad
x_\alpha\in U.
\]
\end{definition}

\begin{proposition}
The convergence notion of
Definition~\ref{def:hdcs-net-convergence}
coincides with ordinary topological convergence in
\[
(C,\mathcal T).
\]
\end{proposition}

\begin{proof}
The family \(N(x)\) is a neighbourhood basis at \(x\).

Hence a net is eventually contained in every topological neighbourhood of
\(x\) if and only if it is eventually contained in every member of \(N(x)\).
\end{proof}

\begin{remark}
Net convergence is entirely stage-internal. It describes limiting behaviour
inside one fixed external topology.

It should not by itself be interpreted as temporal evolution from stage
\(i\) to stage \(i+1\). Cross-stage evolution is represented separately by
carry maps, reification maps, and the transition mechanism.
\end{remark}

\subsection*{Openness}

At every stage,
\[
U\in\mathcal T^{(i)}
\]
if and only if
\[
\forall x\in U\
\exists V\in N^{(i)}(x)
\quad\text{such that}\quad
V\subseteq U.
\]

Equivalently,
\[
U
\]
is open precisely when it is a union of feasible regions.

Thus openness expresses local representability by coherent feasible
contexts.

A continuous carry map preserves openness under inverse image:
if
\[
O\in\mathcal T^{(i+1)},
\]
then
\[
\sigma_i^{-1}(O)
\]
is open in the carry domain \(D_i\).

\subsection*{Connectedness}

\begin{definition}[Connected region]
A subset
\[
A\subseteq C^{(i)}
\]
is \emph{connected} if the subspace
\[
\bigl(
A,\mathcal T^{(i)}|_A
\bigr)
\]
cannot be written as the union of two disjoint nonempty open subsets of
\(A\).
\end{definition}

Connectedness is not required by the HDCS axioms. A stage may contain many
disconnected components.

Likewise, graph-theoretic adjacency connectivity and topological
connectedness should not be identified without additional hypotheses.
An adjacency walk between concepts need not produce a connected union of
overlap witnesses.

\subsection*{Compactness}

\begin{definition}[Compact region]
A subset
\[
A\subseteq C^{(i)}
\]
is \emph{compact} if every family
\[
\{U_\lambda\}_{\lambda\in\Lambda}
\subseteq
\mathcal T^{(i)}
\]
satisfying
\[
A
\subseteq
\bigcup_{\lambda\in\Lambda}U_\lambda
\]
contains a finite subfamily
\[
U_{\lambda_1},
\ldots,
U_{\lambda_n}
\]
such that
\[
A
\subseteq
U_{\lambda_1}
\cup\cdots\cup
U_{\lambda_n}.
\]
\end{definition}

No separation axiom is assumed.

\begin{remark}[Interpretation of compactness]
Compactness is used here only in its standard topological sense. It should not
be identified with metric boundedness, finiteness, or conceptual stability.

In an application, compactness may nevertheless have a useful structural
interpretation: whenever a compact region is described by an arbitrary open
cover, finitely many of those local open contexts already suffice to cover
the region.

Any stronger interpretation---for example, that compactness represents
historical stabilization or semantic finiteness---requires additional
domain-specific justification.
\end{remark}

\subsection*{Role of the additional properties}

\begin{remark}
These properties are analytical tools rather than foundational HDCS axioms.

They serve different purposes:

\begin{itemize}

\item continuity compares selected structures across stages;

\item convergence describes limiting behaviour within a fixed stage;

\item connectedness measures whether a region decomposes topologically into
separated parts;

\item compactness expresses the finite-subcover property of a region.

\end{itemize}

None of these notions is required for the definition of feasible families,
neighbourhood assignments, adjacency, channel ideals, candidate emergence,
or reification.
\end{remark}


\newpage

\section{Applications to exchange and monetary transformation}
\label{sec:applications-exchange}

The purpose of the examples in this section is not to propose a universal
historical sequence of monetary development. In particular, HDCS does not
assume that all societies pass through a fixed progression of the form
\[
\text{gift}
\longrightarrow
\text{barter}
\longrightarrow
\text{commodity money}
\longrightarrow
\text{coinage}
\longrightarrow
\text{fiat money}.
\]

Instead, the examples illustrate a common structural pattern:

\[
\begin{array}{c}
\text{inherited organization}
\\[1mm]
\Downarrow
\\[1mm]
\text{local insufficiency or uncovered structure}
\\[1mm]
\Downarrow
\\[1mm]
\text{candidate emergent region}
\\[1mm]
\Downarrow
\\[1mm]
\text{possible reification}
\end{array}
\]

Different historical cases instantiate this pattern under different empirical
conditions.

The choice of concept set, feasible family, neighbourhood assignment, profiles,
and witnesses is therefore part of the modelling process. Historical evidence
constrains the interpretation of these objects, but does not uniquely determine
a topology.

In particular, HDCS distinguishes between:

\begin{enumerate}[label=\arabic*.]
\item an empirically documented historical configuration;
\item its representation by a concept space;
\item candidate emergence permitted by the HDCS structure; and
\item the later reification of a selected candidate as an atomic concept.
\end{enumerate}

\subsection[Russian enterprise barter under institutional deformation]{Institutional deformation and the reactivation of barter: Russian enterprise exchange}

Russian industrial enterprise exchange during the 1990s provides a useful
example of stagewise structural reorganization. The historical claim is not
that barter was invented in this setting. Rather, nonmonetary and barter-like
settlement became far more prominent inside an already monetized economy when
ordinary monetary settlement became difficult for many firms.

For HDCS, the significant change is therefore the structural role of direct or
nonmonetary exchange within the stage, not the first historical appearance of
barter as a human practice.

Mumssen et al. report that non-monetary transactions in
Russian industrial sales rose from less than ten percent in 1992 to more than
fifty percent in 1998. Their analysis links barter and offsets particularly to
liquidity problems, credit constraints, and enterprise arrears
\cite{CommanderDolinskayaMumssen2000}.

\paragraph{Initial stage.}

Let
\[
C_{\mathrm R}^{(0)}
=
\{
M,L,A,Q,T
\},
\]
where
\[
\begin{aligned}
M &:= \text{monetary settlement},\\
L &:= \text{liquidity constraint},\\
A &:= \text{enterprise arrears},\\
Q &:= \text{restricted access to bank credit},\\
T &:= \text{tax and utility offsets}.
\end{aligned}
\]

The absence of a barter atom from \(C_{\mathrm R}^{(0)}\) is stage-relative:
it means that enterprise barter is not represented as an operative atomic
concept in this chosen initial model. It does not assert that barter was absent
from earlier human history or unknown to the agents involved.

A feasible family
\[
\mathcal F_{\mathrm R}^{(0)}
\]
is chosen to represent empirically coherent combinations of these structures.
For example, one may distinguish regions describing:

\[
\{M,L,A\},
\qquad
\{L,A,Q\},
\qquad
\{A,Q,T\},
\]
together with the nonempty intersections required by the feasible-family
axioms.

The exact choice of
\[
\mathcal F_{\mathrm R}^{(0)}
\quad\text{and}\quad
N_{\mathrm R}^{(0)}
\]
is a modelling decision; the empirical claim is only that liquidity
constraints, arrears, credit conditions, and noncash settlement interacted
strongly in the documented Russian case.

\paragraph{Candidate barter region.}

Consider a stage-\(0\) profile
\[
p_{\mathrm B}^{(0)}
=
(R,U,V_L,V_A)
\]
associated with the concepts \(L\) and \(A\), where \(U\) represents a broader
enterprise-settlement context.

Its uncovered remainder is
\[
R_E^{(0)}(p_{\mathrm B}^{(0)})
=
\operatorname{Int}_{\mathcal T_{\mathrm R}^{(0)}}
\bigl(
U\setminus(V_L\cup V_A)
\bigr).
\]

Choose admissible finite cores and anchored witnesses
\[
W_L\in N_{\mathrm R}^{(0)}(K_L\cup\{L\}),
\qquad
W_A\in N_{\mathrm R}^{(0)}(K_A\cup\{A\}).
\]

If
\[
W_L
\cap
W_A
\cap
R_E^{(0)}(p_{\mathrm B}^{(0)})
\neq\varnothing,
\]
then
\[
E_{\mathrm B}^{(0)}
:=
W_L
\cap
W_A
\cap
R_E^{(0)}(p_{\mathrm B}^{(0)})
\]
is a stage-\(0\) candidate emergent region.

The interpretation is not that money has disappeared, nor that barter has been
historically invented. Rather, \(E_{\mathrm B}^{(0)}\) represents a region of
enterprise practice in which the existing monetary-settlement structure no
longer covers the operative conditions and nonmonetary settlement becomes
structurally prominent or newly organized within the modelled stage.

If this candidate is selected for reification, define
\[
B^{(1)}
:=
\rho_0(E_{\mathrm B}^{(0)}),
\]
where
\[
B^{(1)}
=
\text{enterprise barter}.
\]

At the same time, monetary settlement may persist:
\[
\sigma_0(M)=M^{(1)}.
\]

Thus the next stage may contain both
\[
M^{(1)}
\qquad\text{and}\qquad
B^{(1)}.
\]

This is important: the increased structural role of barter does not imply the
disappearance of money, and stagewise emergence need not mean historical
novelty.

The Russian case therefore illustrates the HDCS pattern
\[
\boxed{
\text{persistence}
+
\text{institutional deformation}
+
\text{reorganization of available exchange forms}.
}
\]
Here the reified atom \(B^{(1)}\) represents enterprise barter as a newly
salient or stabilized structural role in this particular stage model, not the
first invention of barter.

\subsection{From barter to commodity currency:
cigarettes in prisoner-of-war camps}

A second and unusually well-documented example is R.~A. Radford's account of
economic organization in prisoner-of-war camps during the Second World War
\cite{Radford1945}.

Radford describes an economy in which direct barter initially played an
important role. Over time, cigarettes became a widely accepted unit of
account and medium of exchange. Importantly, barter was not eliminated:
cigarette exchange and direct barter continued to coexist.

This makes the example especially useful for HDCS.

\paragraph{Initial exchange structure.}

Let
\[
C_{\mathrm P}^{(0)}
=
\{
D,C,A,P
\},
\]
where
\[
\begin{aligned}
D &:= \text{direct barter},\\
C &:= \text{cigarettes as an ordinary consumption good},\\
A &:= \text{general acceptability in exchange},\\
P &:= \text{prices expressed in a common unit}.
\end{aligned}
\]

Initially, \(C\) refers to cigarettes as a commodity rather than as a monetary
concept.

As trading becomes more organized, feasible regions may represent
configurations in which cigarettes are simultaneously:

\[
\text{portable},
\qquad
\text{divisible},
\qquad
\text{widely accepted},
\qquad
\text{used to quote prices}.
\]

Consider a profile
\[
p_{\mathrm C}^{(0)}
=
(R,U,V_D,V_C)
\]
associated with direct barter and the increasingly exchangeable cigarette
commodity.

With suitable admissible cores and productive witnesses,
\[
W_D
\cap
W_C
\cap
R_E^{(0)}(p_{\mathrm C}^{(0)})
\neq\varnothing.
\]

Define
\[
E_{\mathrm C}^{(0)}
:=
W_D
\cap
W_C
\cap
R_E^{(0)}(p_{\mathrm C}^{(0)}).
\]

This candidate region may be interpreted as the structural configuration in
which cigarettes no longer function merely as one traded good but also act as
a common exchange reference.

If selected for reification,
\[
C_{\mathrm{currency}}^{(1)}
:=
\rho_0(E_{\mathrm C}^{(0)})
\]
is the concept
\[
\text{cigarette currency}.
\]

Direct barter may nevertheless persist:
\[
\sigma_0(D)=D^{(1)}.
\]

Thus
\[
D^{(1)},
\qquad
C_{\mathrm{currency}}^{(1)}
\in
C_{\mathrm P}^{(1)}.
\]

This accurately represents Radford's observation that cigarettes became the
normal currency without extinguishing barter.

The example therefore demonstrates that HDCS emergence should not be read as
replacement:
\[
\boxed{
\text{new organizing form}
\neq
\text{elimination of earlier forms}.
}
\]

\subsection{Standardization and early coinage: the Lydian case}

The emergence of early coinage provides a different kind of application.

The British Museum describes electrum coins produced in Lydia from about
650~BC as among the earliest coins known. These pieces were produced according
to regulated weight standards, and some carry designs such as the lion's head,
which has been interpreted as a possible royal symbol
\cite{BritishMuseumMoney}.

The historical purpose for which the earliest coins were first produced is
not known with certainty. HDCS should therefore not treat a particular causal
account of the invention of coinage as established fact.

Instead, the framework may model the documented structural conjunction.

Let
\[
C_{\mathrm L}^{(0)}
=
\{
E,W,M,A
\},
\]
where
\[
\begin{aligned}
E &:= \text{electrum as valued metal},\\
W &:= \text{regulated weight standard},\\
M &:= \text{stamped or punched marking},\\
A &:= \text{institutional or political authorization}.
\end{aligned}
\]

The empirical evidence supports the coexistence of at least the first three
structures in early Lydian coinage, while the interpretation of particular
marks as royal authorization should be treated cautiously.

Suppose a profile
\[
p_{\mathrm K}^{(0)}
\]
and productive witnesses identify a nonempty candidate region
\[
E_{\mathrm K}^{(0)}
=
W_1
\cap
W_2
\cap
R_E^{(0)}(p_{\mathrm K}^{(0)}).
\]

If selected for reification, define
\[
K^{(1)}
:=
\rho_0(E_{\mathrm K}^{(0)}),
\]
where
\[
K^{(1)}
=
\text{coinage}.
\]

The HDCS claim is deliberately limited:

\begin{quote}
Coinage may be represented as a stabilized concept arising within a documented
configuration of valued metal, regulated weight, marking, and institutional
practice.
\end{quote}

The model does not claim that topology determines why coinage was invented.

\subsection{From convertibility to fiat:
Britain and the end of the gold standard}

A fourth application concerns the transition from convertible banknotes to
fiat money.

For much of the history of Bank of England notes, banknotes could be converted
into gold. The Bank of England identifies 1931 as the end of the British gold
standard; after the end of convertibility, its banknotes constitute fiat money
\cite{BankOfEnglandMoney,BankOfEnglandFiat2013}.

This provides a concrete case in which an inherited monetary object persists
while the institutional structure determining its value changes.

Let
\[
C_{\mathrm F}^{(0)}
=
\{
N,G,L,S,T
\},
\]
where
\[
\begin{aligned}
N &:= \text{banknote circulation},\\
G &:= \text{gold convertibility},\\
L &:= \text{legal monetary denomination},\\
S &:= \text{state or central-bank backing},\\
T &:= \text{public acceptance}.
\end{aligned}
\]

Prior to the transition, one important feasible configuration contains
\[
\{N,G,L,S\}.
\]

After suspension of convertibility, the banknote concept \(N\) persists, but
the role of \(G\) changes fundamentally.

A suitable profile may therefore identify a remainder in which the established
convertibility structure no longer covers the operative monetary context,
while legal denomination, institutional backing, and continued circulation
remain active.

With productive witnesses, let
\[
E_{\mathrm F}^{(0)}
=
W_L
\cap
W_S
\cap
R_E^{(0)}(p_{\mathrm F}^{(0)})
\neq\varnothing.
\]

If this region is selected for reification,
\[
F^{(1)}
:=
\rho_0(E_{\mathrm F}^{(0)}),
\]
where
\[
F^{(1)}
=
\text{fiat monetary regime}.
\]

The banknote itself may persist under the carry map:
\[
\sigma_0(N)=N^{(1)}.
\]

Thus the transformation is not represented as
\[
\text{coin disappears}
\longrightarrow
\text{fiat appears},
\]
but rather as a change in the network of institutional relations supporting
monetary value.

This is closer to the actual historical transformation.

\subsection{Multi-profile emergence}

Some applications involve several distinct structural pressures. In such
cases, one should not define emergence merely by intersecting raw remainders.

Instead, let
\[
E_1,\ldots,E_m
\]
be already valid CCER candidate regions, with
\[
E_\ell
=
W_{\ell,a}
\cap
W_{\ell,b}
\cap
R_E(p_\ell).
\]

If
\[
E_{\mathrm{multi}}
:=
\bigcap_{\ell=1}^{m}E_\ell
\neq\varnothing,
\]
then \(E_{\mathrm{multi}}\) is a nonempty open region because it is a finite
intersection of open sets.

Moreover,
\[
E_{\mathrm{multi}}\subseteq E_\ell
\]
for every \(\ell\). Hence, if
\[
E_\ell\in I_{a_\ell b_\ell},
\]
then downward closure gives
\[
E_{\mathrm{multi}}
\in
I_{a_\ell b_\ell}
\]
for every \(\ell\).

Thus multi-parent emergence is represented by the common part of several
independently supported candidate regions rather than by the intersection of
remainders alone.

\subsection{What the applications establish}

The four cases above should not be interpreted as successive stages of one
universal monetary history.

They illustrate different manifestations of the same HDCS architecture:

\[
\begin{array}{c}
\text{Russian enterprise economy:}\\
\text{cash settlement persists while barter becomes widespread},
\\[3mm]
\text{POW camp economy:}\\
\text{barter persists while cigarette currency becomes dominant},
\\[3mm]
\text{early Lydia:}\\
\text{metal, weight standards, and marking combine in coinage},
\\[3mm]
\text{Britain in 1931:}\\
\text{banknotes persist while gold convertibility disappears}.
\end{array}
\]

The common structural form is
\[
\boxed{
\text{persistence}
+
\text{structural reorganization}
+
\text{candidate emergence}
+
\text{possible reification}.
}
\]

The applications therefore support the use of HDCS as a descriptive formal
language for conceptual and institutional transformation.

They do not establish that topology alone explains the historical causes of
barter, money, coinage, or fiat currency. Those causal claims must come from
the relevant historical and empirical literature.

\subsection{Stagewise history}

When the applications are represented dynamically, persisted concepts are
tracked by the carry maps
\[
\sigma_i
\]
and selected emergent regions are converted into next-stage concept-atoms by
the reification maps
\[
\rho_i.
\]

For example, the POW-camp case may contain
\[
D^{(0)}
\]
for barter and the candidate
\[
E_{\mathrm C}^{(0)}
\]
for cigarette currency.

At the next stage,
\[
\sigma_0(D^{(0)})
=
D^{(1)}
\]
may preserve barter, while
\[
\rho_0(E_{\mathrm C}^{(0)})
=
C_{\mathrm{currency}}^{(1)}
\]
introduces cigarette currency as a new atomic concept.

The history space then records both forms and their declared stagewise
relations.

No claim of topological path-connectedness follows merely from the existence
of these carry and reification relations. Such a conclusion would require
additional connectedness hypotheses on the relevant stagewise supports.


\section{Topology of the HDCS History Space}
\label{sec:hdcs-history-topology}

Let
\[
\bigl(
\mathfrak C^{(i)}
\bigr)_{i\in I},
\qquad
\mathfrak C^{(i)}
=
\bigl(
C^{(i)},\mathcal F^{(i)},N^{(i)}
\bigr),
\]
be an HDCS in the sense of
Definition~\ref{def:hdcs-dynamic}. For every stage \(i\), let
\[
\mathcal T^{(i)}
=
\tau_{\mathcal F^{(i)}}
\]
be the corresponding external topology.

The purpose of this section is to describe how the stage topologies interact
with candidate emergence, carry maps, and the global history space.

No claim is made that the resulting history space is automatically connected,
path connected, compact, or a categorical colimit. Such properties require
additional hypotheses.

\subsection{Stagewise candidate structure}

The results of the preceding sections apply independently within each stage.

\begin{proposition}[Topology of stagewise candidate regions]
\label{prop:hdcs-history-candidate}
Let
\[
E^{(i)}
=
W_a\cap W_b\cap R_E^{(i)}(p^{(i)})
\]
be a stage-\(i\) candidate emergent region arising from a productive witness
pair for concepts
\[
C_a,C_b\in C^{(i)}.
\]
Then:

\begin{enumerate}[label=\arabic*.]

\item
\[
E^{(i)}\in\mathcal T^{(i)};
\]

\item
\[
E^{(i)}\in I^{(i)}_{ab};
\]

\item every open subset
\[
G\subseteq E^{(i)}
\]
lies in \(I^{(i)}_{ab}\);

\item the internal topology of \(E^{(i)}\) is the subspace topology
\[
\mathcal T^{(i)}|_{E^{(i)}}
=
\left\{
G\in\mathcal T^{(i)}
\;\middle|\;
G\subseteq E^{(i)}
\right\}.
\]

\end{enumerate}
\end{proposition}

\begin{proof}
The sets
\[
W_a,\qquad
W_b,\qquad
R_E^{(i)}(p^{(i)})
\]
are open in \(\mathcal T^{(i)}\). Therefore their finite intersection
\(E^{(i)}\) is open.

Since
\[
E^{(i)}
\subseteq
W_a\cap W_b
\]
and
\[
W_a\in N^{(i)}(C_a),
\qquad
W_b\in N^{(i)}(C_b),
\]
the nonempty intersection
\[
W_a\cap W_b
\]
belongs to
\[
O_i(C_a,C_b).
\]
Hence \(E^{(i)}\), being an open subset of this overlap region, lies in the
downward-closed channel ideal
\[
I^{(i)}_{ab}.
\]

The same argument applies to every open
\[
G\subseteq E^{(i)}.
\]

Finally, because \(E^{(i)}\) is itself open,
\[
\mathcal T^{(i)}|_{E^{(i)}}
=
\left\{
G\in\mathcal T^{(i)}
:
G\subseteq E^{(i)}
\right\}.
\]
\end{proof}

\begin{remark}[Internal topology is inherited]
The topology
\[
\mathcal T^{(i)}|_{E^{(i)}}
\]
is the canonical topology inherited from the ambient stage.

It is not necessary, and in general is not appropriate, to describe it as
a ``maximal'' internal topology. Equivalently, it is the coarsest topology on
\(E^{(i)}\) for which the inclusion
\[
E^{(i)}
\hookrightarrow
C^{(i)}
\]
is continuous.
\end{remark}

\subsection{Structural locality}

The Heraclitean locality condition is a modelling constraint on the transition
mechanism rather than a theorem of the stage topology.

\begin{remark}[Declared structural support]
Suppose an elementary event \(\varepsilon\in\mathcal V_i\) at stage \(i\) has
structural support
\[
\operatorname{Supp}_i(\varepsilon)=U\in\mathcal T^{(i)}\setminus\{\varnothing\}.
\]
Condition \textup{(H2)} records only this localization. It does not imply that
an effect must, may, or will propagate beyond \(U\). If an application models
an effect outside \(U\), the relevant direct or multi-step channel relations
must be supplied explicitly as additional event data.

Thus topology provides the language in which support and possible mediation
can be stated, but it supplies no autonomous propagation law.
\end{remark}

\subsection{The timeline and history quotient}

Recall from Definition~\ref{def:hdcs-history} the timeline set
\[
\mathfrak U
:=
\bigsqcup_{i\in I}
\bigl(
\{i\}\times C^{(i)}
\bigr),
\]
the carry-generated equivalence relation \(\approx\), and the history set
\[
C_{\mathrm{hist}}
=
\mathfrak U/\!\approx.
\]
Equip \(\mathfrak U\) with the topological sum of the stage topologies, and let
\[
q:
\mathfrak U
\longrightarrow
C_{\mathrm{hist}}
\]
be the canonical quotient map. The canonical stage maps are the maps
\[
q_i:
C^{(i)}
\longrightarrow
C_{\mathrm{hist}},
\qquad
q_i(x)=q(i,x),
\]
already defined in Definition~\ref{def:hdcs-history}.

The history topology
\[
\mathcal T_{\mathrm{hist}}
\]
is the final topology with respect to the maps \(q_i\). Equivalently,
\[
O\in\mathcal T_{\mathrm{hist}}
\]
if and only if
\[
q_i^{-1}(O)
\in
\mathcal T^{(i)}
\]
for every stage \(i\). With the topological-sum topology on \(\mathfrak U\), this
is equivalently the quotient topology induced by \(q\).

\subsection{Saturation and stage images}

\newcommand{\Sat}{\operatorname{Sat}}

\begin{definition}[Saturation]
\label{def:hdcs-saturation}
For
\[
A\subseteq C^{(i)},
\]
define its saturation in the timeline by
\[
\Sat_i(A)
:=
q^{-1}\bigl(q_i(A)\bigr).
\]

Thus \(\Sat_i(A)\) is the union of all trajectory equivalence classes that
meet the stage-\(i\) subset \(A\).
\end{definition}

\begin{proposition}[Canonical stage maps]
\label{prop:hdcs-canonical-stage-maps}
For every stage \(i\):

\begin{enumerate}[label=\arabic*.]

\item the canonical map
\[
q_i:
\bigl(C^{(i)},\mathcal T^{(i)}\bigr)
\longrightarrow
\bigl(C_{\mathrm{hist}},\mathcal T_{\mathrm{hist}}\bigr)
\]
is continuous;

\item if
\[
A\in\mathcal T^{(i)}
\]
and
\[
\Sat_i(A)
\]
is open in the timeline space \(\mathfrak U\), then
\[
q_i(A)
\]
is open in \(C_{\mathrm{hist}}\).

\end{enumerate}
\end{proposition}

\begin{proof}
Continuity of \(q_i\) follows directly from the definition of the final
topology.

For the second claim,
\[
q^{-1}\bigl(q_i(A)\bigr)
=
\Sat_i(A).
\]
If this set is open in the timeline space, then the definition of the quotient
topology implies that
\[
q_i(A)
\]
is open in \(C_{\mathrm{hist}}\).
\end{proof}

\begin{remark}[Injective stage maps]
By Proposition~\ref{prop:history-stage-injective}, every canonical stage map
\(q_i\) is injective. Thus the history quotient identifies successive
representatives along one declared persistence trajectory without collapsing
distinct concepts that coexist at the same stage.

Injectivity does not imply that \(q_i\) is automatically a topological
embedding onto its image. Continuity follows from the final topology, while
agreement between the stage topology and the subspace topology on
\(q_i(C^{(i)})\) requires additional conditions.
\end{remark}

\begin{remark}
An open subset of a stage does not automatically have an open image in the
history space. Openness of the saturation is an additional condition.

In particular, openness of a candidate emergent region at stage \(i\) does
not by itself imply that its image under \(q_i\) is open in
\(C_{\mathrm{hist}}\).
\end{remark}

\subsection{Continuity of carry trajectories}

Let
\[
\sigma_i:
D_i
\longrightarrow
C^{(i+1)}
\]
be a carry map. The domain \(D_i\) carries the subspace topology
\[
\mathcal T^{(i)}|_{D_i}.
\]

\begin{definition}[Continuous carry map]
\label{def:hdcs-history-continuous-carry}
The carry map \(\sigma_i\) is \emph{continuous} if
\[
V\in\mathcal T^{(i+1)}
\quad\Longrightarrow\quad
\sigma_i^{-1}(V)
\in
\mathcal T^{(i)}|_{D_i}.
\]
\end{definition}

Define the associated history map
\[
\widehat{\sigma}_i:
D_i
\longrightarrow
C_{\mathrm{hist}}
\]
by
\[
\widehat{\sigma}_i
:=
q_{i+1}\circ\sigma_i.
\]

\begin{proposition}[Continuity of carried histories]
\label{prop:hdcs-history-carry-continuity}
If
\[
\sigma_i:
D_i\to C^{(i+1)}
\]
is continuous, then
\[
\widehat{\sigma}_i:
D_i\to C_{\mathrm{hist}}
\]
is continuous.
\end{proposition}

\begin{proof}
The map
\[
q_{i+1}:
C^{(i+1)}
\to
C_{\mathrm{hist}}
\]
is continuous by
Proposition~\ref{prop:hdcs-canonical-stage-maps}.
Hence
\[
\widehat{\sigma}_i
=
q_{i+1}\circ\sigma_i
\]
is a composite of continuous maps.
\end{proof}

\begin{remark}
No openness statement follows from continuity alone.

In particular, neither the carry maps \(\sigma_i\) nor the canonical maps
\(q_i\) are automatically open maps. Openness must be established separately
in any application where it is needed.
\end{remark}

\subsection{Persistence and trajectory classes}

Suppose a concept
\[
x^{(i)}\in C^{(i)}
\]
persists through a finite or infinite succession of carry maps.

Whenever the relevant compositions are defined, write
\[
x^{(j)}
=
\sigma_{j-1}\circ\cdots\circ\sigma_i(x^{(i)})
\]
for its carried representative at stage \(j\geq i\).

By construction,
\[
q_j(x^{(j)})
=
q_i(x^{(i)}).
\]

Thus all carried representatives belong to a single point of the history
space.

\begin{proposition}[Trajectory nets]
\label{prop:hdcs-trajectory-net}
Let
\[
(x_\alpha)_{\alpha\in A}
\]
be a net consisting entirely of carried representatives of one persistent
concept, where \(x_\alpha\in C^{(j_\alpha)}\). Then its image in the history
space is constant:
\[
q_{j_\alpha}(x_\alpha)
=
q_i(x^{(i)})
\]
for every \(\alpha\).

Consequently,
\[
q_{j_\alpha}(x_\alpha)
\longrightarrow
q_i(x^{(i)})
\]
in
\[
C_{\mathrm{hist}}.
\]
\end{proposition}

\begin{proof}
All carried representatives are identified by the equivalence relation
defining \(C_{\mathrm{hist}}\). Hence their quotient images are equal.

Every constant net converges to its constant value.
\end{proof}

\begin{remark}
This convergence is a consequence of quotient identification and should not
be interpreted as evidence of dynamical convergence in the original stage
spaces.

The history topology identifies successive representatives of one declared
trajectory; it does not imply that arbitrary conceptual sequences converge.
\end{remark}

\subsection{Mediated channels}

Let
\[
I^{(\leq k)}_{AB}
\]
denote the cumulative channel ideal generated by witnessed adjacency walks
from \(A\) to \(B\) of length at most \(k\).

The cumulative formulation avoids confusing exact \(k\)-step support with
support available through at most \(k\) steps.

\begin{proposition}[Monotonicity of cumulative mediated channels]
\label{prop:hdcs-cumulative-channel-monotonicity}
For every \(k\geq 1\),
\[
I^{(\leq k)}_{AB}
\subseteq
I^{(\leq k+1)}_{AB}.
\]

If
\[
I^{(*)}_{AB}
\]
denotes the ideal generated by all finite witnessed adjacency walks from
\(A\) to \(B\), then
\[
\bigcup_{k\geq 1}
I^{(\leq k)}_{AB}
\subseteq
I^{(*)}_{AB}.
\]
\end{proposition}

\begin{proof}
Every witnessed walk of length at most \(k\) is also a witnessed walk of
length at most \(k+1\). Hence every generator of
\[
I^{(\leq k)}_{AB}
\]
is available in
\[
I^{(\leq k+1)}_{AB}.
\]
Downward closure gives the first inclusion.

The second follows because
\[
I^{(*)}_{AB}
\]
contains the generators arising from every finite walk.
\end{proof}

\subsection{Connected witnessed channels}

Adjacency alone does not imply topological connectedness.

Let
\[
X_0\sim X_1\sim\cdots\sim X_k
\]
be a witnessed adjacency walk at one stage, and choose
\[
W_r
\in
O(X_{r-1},X_r),
\qquad
r=1,\dots,k.
\]

The sets \(W_r\) need not intersect one another merely because the underlying
concepts form an adjacency walk.

We therefore make the linking condition explicit.

\begin{definition}[Linked witnessed walk]
\label{def:hdcs-linked-walk}
A witnessed adjacency walk
\[
X_0\sim X_1\sim\cdots\sim X_k
\]
with witnesses
\[
W_1,\dots,W_k
\]
is called \emph{linked} if
\[
W_r\cap W_{r+1}\neq\varnothing
\]
for every
\[
r=1,\dots,k-1.
\]
\end{definition}

\begin{proposition}[Connected support of a linked witnessed walk]
\label{prop:hdcs-linked-connected}
Let
\[
W_1,\dots,W_k
\]
be the witnesses of a linked witnessed walk. If every \(W_r\) is connected,
then
\[
S
:=
\bigcup_{r=1}^{k}W_r
\]
is connected.
\end{proposition}

\begin{proof}
Each \(W_r\) is connected and
\[
W_r\cap W_{r+1}\neq\varnothing
\]
for every successive pair.

A finite union of connected subsets with nonempty successive intersections
is connected. Applying this fact inductively gives the result.
\end{proof}

\begin{corollary}[Path-connected support]
\label{cor:hdcs-linked-path-connected}
Under the hypotheses of
Proposition~\ref{prop:hdcs-linked-connected}, if every \(W_r\) is path
connected, then
\[
S
=
\bigcup_{r=1}^{k}W_r
\]
is path connected.
\end{corollary}

\begin{proof}
A union of two path-connected sets having nonempty intersection is path
connected. Induction on \(k\) gives the result.
\end{proof}

\begin{remark}[Endpoint anchoring]
Neither proposition implies that the endpoint concepts
\[
X_0
\quad\text{and}\quad
X_k
\]
belong to the witness support \(S\).

If one additionally assumes
\[
X_0\in W_1
\qquad\text{and}\qquad
X_k\in W_k,
\]
then the endpoint concepts themselves lie in the same connected, or
path-connected, subset \(S\).
\end{remark}

\subsection{Connectedness from a common core}

The following standard criterion is useful for candidate regions produced
from several compatible profiles.

\begin{proposition}[Connectedness from a common core]
\label{prop:hdcs-connected-core}
Let
\[
E\subseteq C^{(i)}
\]
and suppose
\[
K\subseteq E
\]
is nonempty and connected.

Assume that for every \(x\in E\) there exists a finite sequence of connected
sets
\[
W_1^{(x)},\dots,W_{n(x)}^{(x)}
\subseteq E
\]
such that
\[
x\in W_1^{(x)},
\]
\[
W_r^{(x)}
\cap
W_{r+1}^{(x)}
\neq\varnothing
\]
for every applicable \(r\), and
\[
W_{n(x)}^{(x)}
\cap
K
\neq\varnothing.
\]

Then \(E\) is connected.
\end{proposition}

\begin{proof}
For every \(x\in E\), define
\[
S_x
:=
K
\cup
\bigcup_{r=1}^{n(x)}W_r^{(x)}.
\]

By successive overlap, each \(S_x\) is connected. Moreover, every \(S_x\)
contains \(K\).

Hence the family
\[
\{S_x:x\in E\}
\]
consists of connected sets with a common nonempty intersection.

Therefore
\[
E
=
\bigcup_{x\in E}S_x
\]
is connected.
\end{proof}

\begin{remark}
This proposition is purely topological. If \(E\) is a multi-profile candidate,
the required common core and witness chains must be established separately;
they do not follow merely from the CCER axioms.
\end{remark}

\subsection{Channel supports and compactness}

A channel ideal is a family of open subsets, not itself a topological space.
Consequently, compactness should be applied to an associated subset of the
stage space rather than directly to the ideal.

\begin{definition}[Channel support]
\label{def:hdcs-channel-support}
For concepts \(A,B\in C^{(i)}\), define the support of their direct channel by
\[
|I^{(i)}_{AB}|
:=
\bigcup
I^{(i)}_{AB}.
\]

Equivalently,
\[
|I^{(i)}_{AB}|
=
\bigcup
O_i(A,B).
\]
\end{definition}

\begin{proof}[Verification of the equivalence]
Every element of
\[
I^{(i)}_{AB}
\]
is contained in a finite union of overlap regions. Hence
\[
\bigcup I^{(i)}_{AB}
\subseteq
\bigcup O_i(A,B).
\]

Conversely, every overlap region belongs to the channel ideal, so
\[
\bigcup O_i(A,B)
\subseteq
\bigcup I^{(i)}_{AB}.
\]
\end{proof}

\begin{definition}[Compact channel support]
\label{def:hdcs-compact-channel}
The direct channel between \(A\) and \(B\) is said to have
\emph{compact support} if
\[
|I^{(i)}_{AB}|
\]
is compact as a subspace of
\[
(C^{(i)},\mathcal T^{(i)}).
\]
\end{definition}

Thus every open cover
\[
|I^{(i)}_{AB}|
\subseteq
\bigcup_{\lambda\in\Lambda}U_\lambda,
\qquad
U_\lambda\in\mathcal T^{(i)},
\]
admits a finite subcover.

\begin{remark}
Compact support is an optional analytical condition. It is not implied by the
HDCS axioms.

Nor should compactness be interpreted automatically as boundedness,
finiteness, or historical stability. It means only that every open cover of
the channel support admits a finite subcover.
\end{remark}

\subsection{Candidate regions in the history space}

Let
\[
E^{(i)}
\subseteq C^{(i)}
\]
be a stage-\(i\) candidate region.

There are two distinct ways in which it may contribute to later HDCS
structure.

First, individual concepts contained in \(E^{(i)}\) may persist under the
carry map whenever they lie in \(D_i\).

Second, the region itself may be selected for reification:
\[
e_E
=
\rho_i(E^{(i)})
\in
C^{(i+1)}.
\]

These operations must not be confused.

The history quotient identifies concepts related by the carry maps
\[
\sigma_i.
\]
It does not identify the subset
\[
E^{(i)}\subseteq C^{(i)}
\]
with the next-stage atom
\[
e_E\in C^{(i+1)}.
\]

Instead, their relation is stored by the provenance record
\[
\operatorname{Real}_{i+1}(e_E)
=
\bigl(
E^{(i)},
\mathcal T^{(i)}|_{E^{(i)}}
\bigr).
\]

Thus the HDCS distinguishes:

\[
\boxed{
\text{trajectory identification by }\sigma_i
}
\]
from
\[
\boxed{
\text{region-to-concept reification by }\rho_i.
}
\]

\subsection{Relation to the applications}

The historical applications of the preceding section instantiate different
parts of this general structure.

In the Russian enterprise example, monetary settlement may persist through
the carry relation while barter is represented by a newly selected and
reified candidate.

In the prisoner-of-war example, direct barter may persist while cigarette
currency becomes a new organizing concept.

In the Lydian example, coinage is represented as a reified structure arising
from a particular configuration of metal, standardized weight, marking, and
institutional practice.

In the British fiat-money example, banknotes persist while the structural
role of gold convertibility changes.

These examples therefore need not lie on one connected topological path.
They are distinct applications of the same formal architecture.

No connectedness, path-connectedness, or continuity relation between these
historical cases is asserted unless the relevant linking hypotheses are
specified independently.

\subsection*{Summary and interpretation}

The topology of an HDCS operates at two distinct levels.

At a fixed stage,
\[
\mathcal T^{(i)}
\]
describes the open structure generated by the feasible family. Candidate
emergent regions are open subspaces supported simultaneously by channel
structure and uncovered remainder structure.

Across stages, the history space
\[
C_{\mathrm{hist}}
\]
records declared identity relations generated by the carry maps.

Continuity of a carry map is an additional compatibility condition. It is
not required for the history topology to exist.

Likewise:

\begin{itemize}

\item openness of a stage region does not automatically imply openness of its
image in the history space;

\item adjacency does not automatically imply topological connectedness;

\item a witnessed walk has connected support only under additional linking
and connectedness hypotheses;

\item compactness is properly a property of channel supports or other
subspaces, not of channel ideals as collections of open sets;

\item reification of an emergent region is distinct from carrying a persistent
concept through time.

\end{itemize}

Accordingly, the topology does not by itself establish a smooth or continuous
historical trajectory. Rather, it provides a rigorous language in which
persistence, structural interaction, emergence, connectedness, continuity,
and compactness can each be formulated and tested under explicit hypotheses.


\pagebreak

\section{Cross-Space Emergence: Zero, Number, and Representation}
\label{sec:cross-space-zero}

The history of zero provides a useful illustration of interaction between
distinct conceptual domains. It must, however, be treated with some care.
Historically, at least three structures should be distinguished:

\begin{enumerate}[label=\arabic*.]
\item the use of a mark as a placeholder in positional notation;
\item the treatment of zero as a number participating in arithmetic;
\item the linguistic and graphic representations used to denote that number.
\end{enumerate}

These developments are historically related, but they are not identical and
need not occur simultaneously. Accordingly, the present application does not
claim that the mathematical concept of zero arose from one unique historical
cause or one universal sequence.

Instead, the example illustrates how HDCS can represent the interaction of a
\emph{numerical concept space} with a \emph{representational concept space}.

\subsection{Product concept spaces}

Let
\[
\mathfrak C_{\mathrm{Num}}
=
\bigl(
C_{\mathrm{Num}},
\mathcal F_{\mathrm{Num}},
N_{\mathrm{Num}}
\bigr)
\]
be a numerical concept space, and let
\[
\mathfrak C_{\mathrm{Rep}}
=
\bigl(
C_{\mathrm{Rep}},
\mathcal F_{\mathrm{Rep}},
N_{\mathrm{Rep}}
\bigr)
\]
be a representational concept space.

Their external topologies are denoted by
\[
\mathcal T_{\mathrm{Num}}
=
\tau_{\mathcal F_{\mathrm{Num}}},
\qquad
\mathcal T_{\mathrm{Rep}}
=
\tau_{\mathcal F_{\mathrm{Rep}}}.
\]

The representational space may contain linguistic forms, written symbols,
notational conventions, and other structures involved in expressing
mathematical concepts.

\begin{definition}[Product concept space]
\label{def:zero-product-space}
Define
\[
C_{\times}
:=
C_{\mathrm{Num}}
\times
C_{\mathrm{Rep}}.
\]

Its feasible family is
\[
\mathcal F_{\times}
:=
\left\{
U\times V
\;\middle|\;
U\in\mathcal F_{\mathrm{Num}},
\;
V\in\mathcal F_{\mathrm{Rep}}
\right\}.
\]

For
\[
(x,y)\in C_{\times},
\]
define
\[
N_{\times}(x,y)
:=
\left\{
U\times V
\;\middle|\;
U\in N_{\mathrm{Num}}(x),
\;
V\in N_{\mathrm{Rep}}(y)
\right\}.
\]
\end{definition}

\begin{proposition}[Product construction]
\label{prop:zero-product-space}
The triple
\[
\mathfrak C_{\times}
=
\bigl(
C_{\times},
\mathcal F_{\times},
N_{\times}
\bigr)
\]
is a concept space.

Its external topology is the ordinary product topology
\[
\mathcal T_{\times}
=
\mathcal T_{\mathrm{Num}}
\times
\mathcal T_{\mathrm{Rep}}.
\]
\end{proposition}

\begin{proof}
Since
\(\mathcal F_{\mathrm{Num}}\) and
\(\mathcal F_{\mathrm{Rep}}\) cover their respective concept sets,
\(\mathcal F_{\times}\) covers
\(C_{\times}\).

Every product
\[
U\times V\in\mathcal F_{\times}
\]
is nonempty.

If
\[
(U_1\times V_1)\cap(U_2\times V_2)\neq\varnothing,
\]
then
\[
U_1\cap U_2\neq\varnothing
\qquad\text{and}\qquad
V_1\cap V_2\neq\varnothing.
\]
Hence
\[
U_1\cap U_2\in\mathcal F_{\mathrm{Num}}
\]
and
\[
V_1\cap V_2\in\mathcal F_{\mathrm{Rep}},
\]
so
\[
(U_1\times V_1)\cap(U_2\times V_2)
=
(U_1\cap U_2)\times(V_1\cap V_2)
\in
\mathcal F_{\times}.
\]

Thus \(\mathcal F_{\times}\) is a feasible family.

The neighbourhood axioms follow componentwise. In particular, if
\[
A\times B
\subseteq
U\times V
\]
contains \((x,y)\), with
\[
U\in N_{\mathrm{Num}}(x),
\qquad
V\in N_{\mathrm{Rep}}(y),
\]
then
\[
A\subseteq U,
\qquad
B\subseteq V.
\]
Downward closure in the factor spaces gives
\[
A\in N_{\mathrm{Num}}(x),
\qquad
B\in N_{\mathrm{Rep}}(y).
\]
Therefore
\[
A\times B
\in
N_{\times}(x,y).
\]

Finally, \(\mathcal F_{\times}\) is precisely a rectangular basis for the
ordinary product topology generated by
\(\mathcal T_{\mathrm{Num}}\) and
\(\mathcal T_{\mathrm{Rep}}\).
\end{proof}

\subsection{Placeholder and numerical zero}

We first consider the numerical space independently of its linguistic or
graphic realizations.

At an initial stage write
\[
\mathfrak C_{\mathrm{Num}}^{(0)}
=
\bigl(
C_{\mathrm{Num}}^{(0)},
\mathcal F_{\mathrm{Num}}^{(0)},
N_{\mathrm{Num}}^{(0)}
\bigr).
\]

Among its concepts we may distinguish structures such as
\[
P,\quad E,\quad A,\quad R,
\]
where
\[
\begin{aligned}
P &:= \text{positional numerical representation},\\
E &:= \text{an empty position in such a representation},\\
A &:= \text{arithmetic operations},\\
R &:= \text{recording or computational practice}.
\end{aligned}
\]

These are modelling abstractions. They are not intended as a claim that a
particular historical society explicitly possessed concepts with exactly
these boundaries.

Historically, the distinction between a placeholder and zero as a number is
important. Place-value systems can employ a device indicating an empty
position without thereby treating zero as a number participating fully in
arithmetic. Indian mathematical developments eventually supplied both
structures, and Brahmagupta's seventh-century work gives explicit arithmetic
rules involving zero
\cite{Plofker2009,Joseph2011,Ifrah2000}.

Let
\[
p_{\mathrm{Num}}^{(0)}
=
(R_0,U,V_P,V_A)
\]
be a numerical profile involving positional representation and arithmetic
structure.

Its remainder is
\[
R_E^{(0)}(p_{\mathrm{Num}}^{(0)})
=
\operatorname{Int}_{\mathcal T_{\mathrm{Num}}^{(0)}}
\bigl(
U\setminus(V_P\cup V_A)
\bigr).
\]

Choose admissible finite cores
\[
K_P\subseteq U\cap V_P,
\qquad
K_A\subseteq U\cap V_A
\]
and anchored witnesses
\[
W_P
\in
N_{\mathrm{Num}}^{(0)}
\bigl(
K_P\cup\{P\}
\bigr),
\]
\[
W_A
\in
N_{\mathrm{Num}}^{(0)}
\bigl(
K_A\cup\{A\}
\bigr).
\]

If
\[
W_P
\cap
W_A
\cap
R_E^{(0)}(p_{\mathrm{Num}}^{(0)})
\neq
\varnothing,
\]
then
\[
E_{\mathrm{Zero}}^{(0)}
:=
W_P
\cap
W_A
\cap
R_E^{(0)}(p_{\mathrm{Num}}^{(0)})
\]
is a candidate emergent region.

It is open and satisfies
\[
E_{\mathrm{Zero}}^{(0)}
\in
I_{PA}^{(0)}.
\]

The interpretation is that the candidate identifies configurations in which
existing positional and arithmetic structures admit a further organization
capable of treating absence as a determinate numerical value.

This is a structural interpretation of the HDCS model, not a historical
claim that a single ``contradiction'' between positional notation and
arithmetic caused the invention of zero.

If the candidate is selected for reification, define
\[
Z^{(1)}
:=
\rho_0
\bigl(
E_{\mathrm{Zero}}^{(0)}
\bigr),
\]
where
\[
Z^{(1)}
=
\text{zero as a numerical concept}.
\]

Thus
\[
E_{\mathrm{Zero}}^{(0)}
\subseteq
C_{\mathrm{Num}}^{(0)}
\]
is the candidate region, while
\[
Z^{(1)}
\in
C_{\mathrm{Num}}^{(1)}
\]
is its reified next-stage concept.

\subsection{Historical distinction between placeholder and number}

The distinction represented above has historical motivation.

The Bakhshali manuscript contains a dot used within positional notation as a
placeholder. Its dating has been the subject of continuing scholarly
discussion, so the present application does not rely on a claim that it is
the uniquely earliest surviving zero symbol.

More important for the present model is the functional distinction:
a placeholder marks an empty position in a numeral, whereas numerical zero
can itself enter arithmetic.

Brahmagupta's
\emph{Brahmasphutasiddhanta}, written in \(628\), explicitly gives rules
involving zero, including addition, subtraction, and multiplication
\cite{Plofker2009,Joseph2011}.

Thus the historical record itself provides reason to distinguish
\[
\boxed{\text{placeholder}}
\]
from
\[
\boxed{\text{zero as number}}.
\]

HDCS represents this distinction by allowing an earlier representational
structure to persist while a later numerical structure is selected and
reified.

\subsection{Names and representational transmission}

We next consider the representational concept space
\[
\mathfrak C_{\mathrm{Rep}}
=
\bigl(
C_{\mathrm{Rep}},
\mathcal F_{\mathrm{Rep}},
N_{\mathrm{Rep}}
\bigr).
\]

It is useful to distinguish two kinds of representation:

\begin{enumerate}[label=\arabic*.]
\item lexical representations, such as words meaning zero or emptiness;
\item graphic representations, such as positional dots and the later
      written sign \(0\).
\end{enumerate}

These should not be identified automatically.

For the lexical history, relevant forms include
\[
L_{\textsf{śūnya}},
\qquad
L_{\textsf{ṣifr}},
\qquad
L_{\textsf{zephirum}},
\qquad
L_{\textsf{zero}}.
\]

Here
\emph{śūnya} is the Sanskrit term associated with emptiness or void,
\emph{ṣifr} is the corresponding Arabic mathematical term, and the Medieval
Latin form \emph{zephirum} occurs in Fibonacci's
\emph{Liber Abaci}. Later Italian forms produced the modern word
\emph{zero}
\cite{Ifrah2000,Menninger1992,Plofker2009}.

This historical sequence should not be described as a purely phonetic
transformation. The transition from Sanskrit terminology to Arabic
terminology involves semantic translation as well as transmission, while
the later development from Arabic \emph{ṣifr} through Medieval Latin and
Italian has an etymological history of its own.

Nor does etymological succession by itself imply topological adjacency.

Rather, when historical evidence identifies a concrete context of translation,
bilingual mathematical practice, manuscript transmission, or scholarly
contact, that context may be represented by an overlap region in the
representational concept space.

For example, if
\[
L_{\textsf{ṣifr}}
\]
and
\[
L_{\textsf{zephirum}}
\]
are represented by concepts whose neighbourhoods contain a documented
translation or transmission context, then one may have
\[
O_{\mathrm{Rep}}
\bigl(
L_{\textsf{ṣifr}},
L_{\textsf{zephirum}}
\bigr)
\neq\varnothing.
\]

The resulting channel describes that particular representational relation.

No general inclusion such as
\[
I_{\textsf{India,Arabia}}
\subseteq
I_{\textsf{Arabia,Europe}}
\]
is asserted, since ideals associated with different parent pairs need not be
nested.

\subsection{Cross-space realization}

After numerical zero has been reified, let
\[
Z\in C_{\mathrm{Num}}
\]
denote the numerical concept of zero.

Let
\[
S\in C_{\mathrm{Rep}}
\]
denote a selected representation of zero, such as an established written
zero sign in a particular mathematical tradition.

Choose neighbourhoods
\[
U_Z\in N_{\mathrm{Num}}(Z),
\qquad
V_S\in N_{\mathrm{Rep}}(S).
\]

Then
\[
J
:=
U_Z\times V_S
\]
is a feasible and open region of the product concept space.

We call \(J\) a \emph{joint realization region}. It represents contexts in
which a numerical role and a representational form occur together.

\begin{remark}
The fact that \(J\) is open does not by itself make it a CCER candidate
emergent region.

Openness, product realization, and CCER emergence are distinct notions.
\end{remark}

\subsection{Cross-space CCER}

Because
\[
\mathfrak C_{\times}
\]
is itself a concept space, the ordinary CCER construction applies to it
without requiring a separate emergence rule.

Let
\[
X_a,X_b
\in
C_{\times}
\]
be product concepts, and let
\[
p_{\times}
=
(R_{\times},U_{\times},V_a,V_b)
\]
be a profile in the product concept space.

Choose admissible finite cores
\[
K_a\subseteq U_{\times}\cap V_a,
\qquad
K_b\subseteq U_{\times}\cap V_b
\]
and anchored witnesses
\[
W_a
\in
N_{\times}
\bigl(
K_a\cup\{X_a\}
\bigr),
\]
\[
W_b
\in
N_{\times}
\bigl(
K_b\cup\{X_b\}
\bigr).
\]

The product remainder is
\[
R_{E,\times}(p_{\times})
=
\operatorname{Int}_{\mathcal T_{\times}}
\bigl(
U_{\times}\setminus(V_a\cup V_b)
\bigr).
\]

If
\[
W_a
\cap
W_b
\cap
R_{E,\times}(p_{\times})
\neq\varnothing,
\]
then
\[
E_{\times}
:=
W_a
\cap
W_b
\cap
R_{E,\times}(p_{\times})
\]
is a cross-space candidate emergent region.

It is a nonempty open subset of
\[
C_{\mathrm{Num}}\times C_{\mathrm{Rep}}
\]
and lies in the corresponding product-space channel ideal.

If the application is intended specifically to represent stabilization of the
zero concept with a particular representation, one additionally requires
\[
E_{\times}\subseteq J.
\]

Thus the model records a cross-space candidate only when the ordinary CCER
conditions are actually satisfied inside the product concept space.

\subsection{Reification of a joint symbolic concept}

Suppose
\[
E_{\times}
\]
is a selected cross-space candidate satisfying
\[
E_{\times}\subseteq J.
\]

To reify it, a target next-stage concept space must first be specified. Let
\[
C_{\mathrm{Sym}}^{(1)}
\]
be such a target concept set, and let
\[
\rho_{\times}:
\mathscr E_{\times}^{\mathrm{sel}}
\longrightarrow
C_{\mathrm{Sym}}^{(1)}
\]
be the selected injective cross-space reification map. Define
\[
\rho_{\times}(E_{\times})
=
Z_{\mathrm{sym}}
\in
C_{\mathrm{Sym}}^{(1)},
\]
where \(Z_{\mathrm{sym}}\) represents a stabilized symbolic realization of
numerical zero.

The distinction of types remains essential:
\[
E_{\times}
\subseteq
C_{\mathrm{Num}}\times C_{\mathrm{Rep}}
\]
is a region, whereas
\[
Z_{\mathrm{sym}}
\]
is a next-stage concept.

Its provenance may retain
\[
\operatorname{Real}(Z_{\mathrm{sym}})
=
\bigl(
E_{\times},
\mathcal T_{\times}|_{E_{\times}}
\bigr).
\]

\subsection{Coordinate projections}

Let
\[
\pi_{\mathrm{Num}}
:
C_{\mathrm{Num}}\times C_{\mathrm{Rep}}
\longrightarrow
C_{\mathrm{Num}}
\]
and
\[
\pi_{\mathrm{Rep}}
:
C_{\mathrm{Num}}\times C_{\mathrm{Rep}}
\longrightarrow
C_{\mathrm{Rep}}
\]
be the coordinate projections.

\begin{proposition}[Projections of an open cross-space candidate]
\label{prop:zero-product-projections}
Let
\[
E_{\times}
\subseteq
C_{\mathrm{Num}}\times C_{\mathrm{Rep}}
\]
be open.

Then the restricted maps
\[
\pi_{\mathrm{Num}}|_{E_{\times}}
:
E_{\times}\to C_{\mathrm{Num}}
\]
and
\[
\pi_{\mathrm{Rep}}|_{E_{\times}}
:
E_{\times}\to C_{\mathrm{Rep}}
\]
are continuous and open.
\end{proposition}

\begin{proof}
Coordinate projections from a product topology are continuous and open.

Their restrictions to \(E_{\times}\) are continuous.

Since \(E_{\times}\) is open in the product, every open subset of
\(E_{\times}\) is also open in the ambient product space. Its coordinate
projection is therefore open in the corresponding factor.
\end{proof}

Thus both the numerical and representational aspects of an open cross-space
candidate remain visible through continuous open projections.

This does not imply that either projection uniquely determines the full
cross-space structure.

\subsection{Connectedness}

Cross-space emergence does not imply connectedness.

However, product structure supplies a simple sufficient condition.

\begin{proposition}[Connected joint realization]
\label{prop:zero-product-connected}
Let
\[
U\subseteq C_{\mathrm{Num}},
\qquad
V\subseteq C_{\mathrm{Rep}}
\]
be connected subspaces.

Then
\[
U\times V
\]
is connected in the product topology.

If both \(U\) and \(V\) are path connected, then
\[
U\times V
\]
is path connected.
\end{proposition}

\begin{proof}
These are standard properties of products of connected and path-connected
spaces.
\end{proof}

\begin{remark}
A general cross-space candidate
\[
E_{\times}
\subseteq U\times V
\]
need not itself be connected even when \(U\times V\) is connected.

Connectedness of \(E_{\times}\) must therefore be proved from additional
hypotheses, such as the linked-witness or common-core conditions developed
earlier.
\end{remark}

\subsection{Historical interpretation}

The HDCS representation separates several historical processes that should
not be collapsed into one event.

First, positional notation can contain a device for representing an empty
place. Such a device need not yet be a number.

Second, zero can acquire an arithmetic role. Brahmagupta's work provides a
particularly clear historical example of explicit operations involving zero.

Third, numerical concepts acquire lexical and written representations whose
forms may be transmitted, translated, and modified across linguistic and
scholarly communities.

Finally, numerical and representational structures can become stably coupled
in mathematical practice.

The HDCS model represents these as potentially distinct stages and domains:

\[
\boxed{
\text{placeholder structure}
}
\]

\[
\Downarrow
\]

\[
\boxed{
\text{zero as numerical concept}
}
\]

together with an independently developing

\[
\boxed{
\text{representational structure}
}
\]

whose interaction can support

\[
\boxed{
\text{joint symbolic realization}.
}
\]

This diagram represents a modelling architecture rather than a claim of
unique historical causation.

\subsection{Limits of the example}

The zero example illustrates the expressive capacity of HDCS, but several
limits should be explicit.

\begin{enumerate}[label=\arabic*.]

\item The topology does not prove that practical necessity, trade, accounting,
or a logical contradiction caused the historical development of zero.

\item A placeholder symbol and zero as a number are mathematically and
historically distinct structures.

\item The words
\emph{śūnya},
\emph{ṣifr},
\emph{zephirum},
and
\emph{zero}
belong to a history of lexical transmission and translation; they should not
be treated as successive phonetic variants of one unchanged word.

\item The graphic sign \(0\) has a notational history that must not be
identified automatically with the etymological history of the word
``zero.''

\item Historical succession does not by itself establish neighbourhood
overlap, channel membership, connectedness, or continuity. Those properties
must be represented by explicit HDCS data and justified in the application.

\item A product-open region is not automatically a CCER candidate.
Cross-space emergence requires the same productive witness and remainder
conditions as emergence in any other concept space.

\end{enumerate}

Accordingly, the contribution of the example is not the claim that HDCS
reconstructs the unique origin of zero. Rather, it shows how the framework
can distinguish and relate numerical structure, representational structure,
historical transmission, candidate emergence, and later reification without
identifying them prematurely.

\subsection*{General product notation for later applications}
\label{subsec:cross-space-product}

The product construction has already been established in
Definition~\ref{def:zero-product-space} and
Proposition~\ref{prop:zero-product-space}. For later applications we reuse it
with generic factor notation
\[
\mathfrak C_1=(C_1,\mathcal F_1,N_1),
\qquad
\mathfrak C_2=(C_2,\mathcal F_2,N_2),
\]
writing
\[
C_{\times}=C_1\times C_2,
\qquad
\mathcal T_{\times}=\mathcal T_1\times\mathcal T_2.
\]
All overlap families, profiles, remainders, channels, and CCER candidates in a
cross-space application are the ordinary single-stage constructions applied in
this product concept space. No second cross-space emergence operator is
introduced.

\subsection*{Asynchronous development of factor spaces}

The component spaces need not evolve at the same rate.

Suppose
\[
\mathfrak C_1^{(i)}
=
(C_1^{(i)},\mathcal F_1^{(i)},N_1^{(i)})
\]
and
\[
\mathfrak C_2^{(j)}
=
(C_2^{(j)},\mathcal F_2^{(j)},N_2^{(j)})
\]
are stages in two independently evolving HDCS systems.

Their joint structure may be formed at any selected pair of stages:
\[
\mathfrak C_{\times}^{(i,j)}
:=
\mathfrak C_1^{(i)}
\times
\mathfrak C_2^{(j)}.
\]

This pair-indexed notation avoids requiring the two domains to develop
synchronously.

If a single sequence of joint stages is desired, one may choose a
synchronization schedule
\[
s:n\longmapsto(i_n,j_n)
\]
subject to
\[
i_{n+1}\ge i_n,\qquad j_{n+1}\ge j_n
\qquad\text{for every applicable }n,
\]
and define
\[
\mathfrak C_{\times}^{[n]}
:=
\mathfrak C_1^{(i_n)}
\times
\mathfrak C_2^{(j_n)}.
\]

The synchronization schedule is additional modelling data; it is not imposed
by the topology itself. Coordinatewise monotonicity prevents a joint history from
moving backwards in either factor, while equality in one coordinate allows one
domain to remain unchanged as the other advances.

\section{Cross-Space Application: Evolution of the Mammalian Middle Ear}
\label{sec:mammalian-ear}

The evolution of the mammalian middle ear provides a useful biological
application of cross-space modelling because major changes in skeletal
morphology are closely associated with changes in mechanical sound
transmission.

The application must nevertheless distinguish fossil evidence from functional
interpretation. Morphological transformations can be reconstructed directly
from comparative anatomy, developmental evidence, and fossils. Auditory
performance in extinct organisms is inferred indirectly from those structures.

Accordingly, we use a \emph{morphological concept space} and an
\emph{auditory-functional concept space}, rather than treating perception
itself as directly observable in the fossil record.

\subsection{Morphological and auditory-functional spaces}

Let
\[
\mathfrak C_{\mathrm{Morph}}
=
(
C_{\mathrm{Morph}},
\mathcal F_{\mathrm{Morph}},
N_{\mathrm{Morph}}
)
\]
represent morphological configurations, and let
\[
\mathfrak C_{\mathrm{Aud}}
=
(
C_{\mathrm{Aud}},
\mathcal F_{\mathrm{Aud}},
N_{\mathrm{Aud}}
)
\]
represent functional configurations associated with mechanical sound
transmission.

Their product is
\[
\mathfrak C_{\mathrm{Ear}}
=
\mathfrak C_{\mathrm{Morph}}
\times
\mathfrak C_{\mathrm{Aud}},
\]
with topology
\[
\mathcal T_{\mathrm{Ear}}
=
\mathcal T_{\mathrm{Morph}}
\times
\mathcal T_{\mathrm{Aud}}.
\]

\paragraph{Morphological structure.}

Among the relevant morphological concepts are
\[
D,\quad A,\quad Q,\quad G,\quad S,
\]
where
\[
\begin{aligned}
D &:= \text{dentary},\\
A &:= \text{articular},\\
Q &:= \text{quadrate},\\
G &:= \text{angular},\\
S &:= \text{dentary--squamosal articulation}.
\end{aligned}
\]

Comparative anatomy, developmental biology, and the fossil record support
the homologies
\[
\text{articular}\longleftrightarrow\text{malleus},
\]
\[
\text{quadrate}\longleftrightarrow\text{incus},
\]
and
\[
\text{angular}\longleftrightarrow\text{ectotympanic}.
\]
The stapes has a different and older evolutionary history and should not be
treated as another postdentary jaw element
\cite{AllinHopson1992,Kemp2005,Luo2007}.

The transformation of the ancestral articular--quadrate jaw articulation
was associated with the appearance of a secondary
dentary--squamosal jaw joint. This allowed the former primary-joint elements
to become increasingly specialized in the auditory apparatus
\cite{AllinHopson1992,Kemp2005,Luo2007}.

\paragraph{Auditory-functional structure.}

The auditory-functional space may contain concepts such as
\[
F_{\mathrm{load}},
\qquad
F_{\mathrm{mixed}},
\qquad
F_{\mathrm{trans}},
\]
where
\[
\begin{aligned}
F_{\mathrm{load}}
&:= \text{predominantly load-bearing jaw function},\\
F_{\mathrm{mixed}}
&:= \text{combined mechanical and auditory participation},\\
F_{\mathrm{trans}}
&:= \text{specialized transmission toward the inner ear}.
\end{aligned}
\]

These functional concepts summarize biomechanical interpretations rather than
directly observed mental or perceptual states.

\subsection{Joint morphological--functional configurations}

A point of
\[
C_{\mathrm{Morph}}\times C_{\mathrm{Aud}}
\]
represents a joint morphological--functional configuration.

For example, a feasible region of the form
\[
U_{\mathrm{jaw}}
\times
V_{\mathrm{mixed}}
\]
may represent configurations in which postdentary elements remain attached
to the mandible while also participating in sound transmission.

Another region
\[
U_{\mathrm{detached}}
\times
V_{\mathrm{trans}}
\]
may represent configurations in which the relevant ossicular elements are
no longer load-bearing components of the primary jaw articulation and are
specialized for auditory transmission.

Such regions must be constrained by comparative anatomy, development, and
palaeontology.

\subsection{Transitional configurations}

The fossil record contains intermediate mammaliaform configurations in which
the postdentary complex retains a structural connection to the mandible while
the dentary--squamosal joint assumes increasing importance.

Such evidence is compatible with a transition from structures participating
in jaw mechanics toward increasingly specialized auditory configurations
\cite{AllinHopson1992,Kemp2005,Luo2007}.

In HDCS terms, let
\[
X_a,X_b
\in
C_{\mathrm{Ear}}
\]
represent two empirically motivated joint configurations.

Let
\[
p_{\mathrm{Ear}}
=
(R,U,V_a,V_b)
\]
be a product profile between them.

Its remainder is
\[
R_{E,\mathrm{Ear}}(p_{\mathrm{Ear}})
=
\operatorname{Int}_{\mathcal T_{\mathrm{Ear}}}
\bigl(
U\setminus(V_a\cup V_b)
\bigr).
\]

The mathematical remainder is simply the open part of the selected background
context not covered by the two selected parent neighbourhoods.

It should not by itself be interpreted as an evolutionary pressure,
adaptation, or transitional fossil. Such an interpretation requires
independent biological evidence.

\subsection{Candidate middle-ear configuration}

Choose admissible finite cores
\[
K_a\subseteq U\cap V_a,
\qquad
K_b\subseteq U\cap V_b,
\]
satisfying
\[
N_{\mathrm{Ear}}(K_a\cup\{X_a\})\neq\varnothing,
\]
\[
N_{\mathrm{Ear}}(K_b\cup\{X_b\})\neq\varnothing.
\]

Choose anchored witnesses
\[
W_a
\in
N_{\mathrm{Ear}}(K_a\cup\{X_a\}),
\]
\[
W_b
\in
N_{\mathrm{Ear}}(K_b\cup\{X_b\}).
\]

If
\[
W_a
\cap
W_b
\cap
R_{E,\mathrm{Ear}}(p_{\mathrm{Ear}})
\neq
\varnothing,
\]
define
\[
E_{\mathrm{MME}}
:=
W_a
\cap
W_b
\cap
R_{E,\mathrm{Ear}}(p_{\mathrm{Ear}}).
\]

Then \(E_{\mathrm{MME}}\) is a candidate emergent region in the product
concept space.

By the general CCER theory,
\[
E_{\mathrm{MME}}
\in
\mathcal T_{\mathrm{Ear}}
\]
and
\[
E_{\mathrm{MME}}
\in
I_{ab}^{\mathrm{Ear}}.
\]

No global finite-neighbourhood condition
\[
(N_{\mathrm{fin}})
\]
is required.

\subsection{Reification of the middle-ear concept}

The candidate region is not itself a next-stage concept.

If biological and modelling criteria select
\[
E_{\mathrm{MME}}
\]
for reification, specify a next-stage ear concept set \(C_{\mathrm{Ear}}^{(1)}\)
and an injective reification map
\[
\rho_{\mathrm{Ear}}:
\mathscr E_{\mathrm{Ear}}^{\mathrm{sel}}
\longrightarrow
C_{\mathrm{Ear}}^{(1)}.
\]
Define
\[
M_{\mathrm{MME}}
:=
\rho_{\mathrm{Ear}}(E_{\mathrm{MME}})
\in C_{\mathrm{Ear}}^{(1)}.
\]

Then
\[
E_{\mathrm{MME}}
\subseteq
C_{\mathrm{Morph}}\times C_{\mathrm{Aud}}
\]
is the earlier structured region, while
\[
M_{\mathrm{MME}}
\]
is the corresponding reified concept at the next stage.

Its provenance may retain
\[
\operatorname{Real}(M_{\mathrm{MME}})
=
\left(
E_{\mathrm{MME}},
\mathcal T_{\mathrm{Ear}}|_{E_{\mathrm{MME}}}
\right).
\]

Biologically, the reified concept represents the integrated mammalian
middle-ear configuration involving the malleus, incus, stapes, tympanic
structures, and their relation to the reorganized mammalian jaw.

The HDCS construction does not claim that all components appeared at once.
It represents their joint structural stabilization at the level chosen by
the model.

\subsection{Repeated and branching evolution}

The evolution of the definitive mammalian middle ear should not be represented
as one simple linear transformation shared identically by every mammaliaform
lineage.

Fossil and phylogenetic studies indicate that complete detachment of the
postdentary middle-ear complex from the mandible may have occurred
independently in several mammaliaform lineages.

The basic HDCS definition in this article is sequential: each individual HDCS
is indexed by \(\mathbb N\) or by a finite initial segment. Branching is
therefore not an additional operation inside a single sequential history.
Instead, distinct evolutionary lineages may be represented by distinct
descendant HDCS sequences that share a common precursor stage or precursor
model.

Schematically, if \(P\) is a common precursor configuration, three descendant
models may begin with corresponding successors
\[
P\longrightarrow P_1,
\qquad
P\longrightarrow P_2,
\qquad
P\longrightarrow P_3,
\]
where each arrow belongs to its own descendant stage sequence. The branches
may then contain distinct candidate regions
\[
E_1,\qquad E_2,\qquad E_3
\]
and distinct reifications
\[
\rho_1(E_1),
\qquad
\rho_2(E_2),
\qquad
\rho_3(E_3).
\]
These need not represent one historically identical event. HDCS can therefore
compare similar structural outcomes across separately modelled lineages without
claiming that the present sequential definition itself forms a branching
history. A genuinely branching index structure would require an explicit
extension of the basic definition.

This distinction is particularly important in evolutionary applications.

\subsection{Coordinate projections}

The product projections
\[
\pi_{\mathrm{Morph}}
:
C_{\mathrm{Morph}}\times C_{\mathrm{Aud}}
\longrightarrow
C_{\mathrm{Morph}}
\]
and
\[
\pi_{\mathrm{Aud}}
:
C_{\mathrm{Morph}}\times C_{\mathrm{Aud}}
\longrightarrow
C_{\mathrm{Aud}}
\]
are continuous and open.

Because
\[
E_{\mathrm{MME}}
\]
is open in the product topology, their restrictions
\[
\pi_{\mathrm{Morph}}|_{E_{\mathrm{MME}}}
\]
and
\[
\pi_{\mathrm{Aud}}|_{E_{\mathrm{MME}}}
\]
are also continuous and open onto their images.

Thus the joint candidate has visible morphological and functional
projections.

This does not imply that either projection uniquely determines the complete
cross-space candidate.

\subsection{Connectedness}

Cross-space emergence does not automatically imply connectedness.

If one wishes to interpret
\[
E_{\mathrm{MME}}
\]
as a connected structural region, connectedness must be established by
additional hypotheses.

For example, the linked-witness criterion of
Proposition~\ref{prop:hdcs-linked-connected}
may be applied when the relevant connected witnesses have nonempty successive
intersections.

Alternatively, the common-core criterion of
Proposition~\ref{prop:hdcs-connected-core}
may be used when a nonempty connected core is linked to every part of the
candidate through connected subsets.

Neither property follows merely from the existence of a CCER candidate.

\subsection{Later auditory specialization}

Once a mammalian middle-ear configuration is present, later lineages exhibit
substantial variation in auditory specialization.

Such variation may be represented by later morphological and
auditory-functional concept spaces, with candidate regions corresponding to
particular frequency sensitivities, cochlear structures, ecological
conditions, or communication systems.

The present framework does not, however, treat the modern human ear as a
further CCER emergent optimized specifically for speech.

Human auditory sensitivity, speech reception, outer-ear geometry,
middle-ear transmission, cochlear organization, and language evolution are
distinct phenomena. Establishing an adaptive causal relation among them would
require evidence beyond the topological model.

Accordingly, later human auditory specialization may be studied as a separate
application, but it is not taken here as the necessary endpoint of mammalian
middle-ear evolution.

\subsection{Interpretation and limitations}

The mammalian-ear example illustrates how HDCS can represent the interaction
of two kinds of structure:

\[
\boxed{
\text{morphological organization}
}
\]

and

\[
\boxed{
\text{auditory-functional organization}.
}
\]

The product concept space records joint configurations of these structures.

Within that product:

\begin{enumerate}[label=\arabic*.]

\item empirically motivated configurations determine feasible regions;

\item product profiles identify local structural contexts;

\item productive witnesses may define candidate emergent regions;

\item selected candidates may be reified as later concepts;

\item different evolutionary branches may independently produce structurally
similar outcomes.

\end{enumerate}

The application does not claim that topology explains natural selection or
causes anatomical transformation.

Nor does the remainder
\[
R_E(p)
\]
itself represent evolutionary pressure.

Rather, topology supplies a formal language for representing relationships
among documented or hypothesized configurations.

The historical and biological content must continue to be supplied by fossil
evidence, comparative anatomy, developmental biology, and biomechanical
analysis.

\subsection*{Biological summary}

The strongest empirically supported part of the example is the relationship
between jaw evolution and the mammalian middle ear.

The malleus and incus are homologous to the articular and quadrate of the
ancestral jaw articulation, while the ectotympanic is associated with the
angular. Their incorporation into the auditory apparatus was associated with
the development of the dentary--squamosal jaw articulation
\cite{AllinHopson1992,Kemp2005,Luo2007}.

Fossils document intermediate conditions between jaw-associated and fully
detached middle-ear configurations. Evidence also indicates that complete
detachment was not necessarily a single event inherited unchanged by every
mammalian lineage.

For HDCS, this makes the mammalian middle ear especially useful not as an
example of inevitable linear progress, but as an example of
\[
\boxed{
\text{reorganization, persistence, branching, and repeated emergence}.
}
\]

\section[Cross-Space Application: AI, Task Transformation, and Labour Displacement]{Cross-Space Application: AI, Task Transformation,\texorpdfstring{\\}{ } and Labour Displacement}
\label{sec:ai-labour}

The rapid development and deployment of artificial intelligence provides a
contemporary application of cross-space modelling. The relevant interaction is
not between ``humans'' and ``AI'' as undifferentiated entities, but between
economic structures and structures describing AI capabilities and deployment.

Let
\[
\mathfrak C_{\mathrm{Econ}}
=
\bigl(
C_{\mathrm{Econ}},
\mathcal F_{\mathrm{Econ}},
N_{\mathrm{Econ}}
\bigr)
\]
be an economic concept space, and let
\[
\mathfrak C_{\mathrm{AI}}
=
\bigl(
C_{\mathrm{AI}},
\mathcal F_{\mathrm{AI}},
N_{\mathrm{AI}}
\bigr)
\]
be an AI-capability concept space.

Their external topologies are
\[
\mathcal T_{\mathrm{Econ}}
=
\tau_{\mathcal F_{\mathrm{Econ}}},
\qquad
\mathcal T_{\mathrm{AI}}
=
\tau_{\mathcal F_{\mathrm{AI}}}.
\]

The associated product concept space is
\[
\mathfrak C_{\times}
=
\mathfrak C_{\mathrm{Econ}}
\times
\mathfrak C_{\mathrm{AI}},
\]
with
\[
C_{\times}
=
C_{\mathrm{Econ}}
\times
C_{\mathrm{AI}}
\]
and product topology
\[
\mathcal T_{\times}
=
\mathcal T_{\mathrm{Econ}}
\times
\mathcal T_{\mathrm{AI}}.
\]

\subsection{Economic and AI-capability structures}

Among the economic concepts one may distinguish, for example,
\[
L,\quad Y,\quad D,\quad Q,\quad R,
\]
where
\[
\begin{aligned}
L &:= \text{demand for human labour},\\
Y &:= \text{productive output},\\
D &:= \text{income distribution},\\
Q &:= \text{effective demand},\\
R &:= \text{occupational reallocation or retraining}.
\end{aligned}
\]

Among the AI-capability concepts one may distinguish
\[
A,\quad S,\quad M,\quad P,
\]
where
\[
\begin{aligned}
A &:= \text{task-level AI capability},\\
S &:= \text{substitution potential},\\
M &:= \text{human--AI complementarity},\\
P &:= \text{organizational deployment}.
\end{aligned}
\]

These concepts are deliberately more limited than claims about complete
occupations. An occupation normally consists of many heterogeneous tasks, and
the ability of an AI system to perform some of those tasks does not by itself
imply elimination of the corresponding job.

Recent labour-market research supports this distinction. The refined
ILO--NASK global exposure index estimates that approximately one quarter of
global employment lies in occupations with some exposure to generative AI,
while emphasizing that transformation of occupations is generally more likely
than full replacement \cite{GmyrekEtAl2025}. Subsequent reviews of observed
adoption likewise report heterogeneous productivity effects and, so far,
limited evidence of large-scale aggregate displacement
\cite{ILOGenAI2026}.

Accordingly, the present application treats labour displacement as a
\emph{possible structural outcome}, not as an inevitable consequence of AI.

\subsection{Joint economic--AI configurations}

A point
\[
(e,a)
\in
C_{\mathrm{Econ}}
\times
C_{\mathrm{AI}}
\]
represents a joint economic--technological configuration.

For example,
\[
(L,M)
\]
may represent a configuration in which AI capability complements existing
human labour, whereas
\[
(L,S)
\]
may represent a configuration in which some of the same economic activity is
subject to significant task-substitution potential.

Similarly, feasible product regions may represent combinations such as
\[
U_{\mathrm{output}}
\times
V_{\mathrm{deployment}},
\]
or
\[
U_{\mathrm{labour}}
\times
V_{\mathrm{substitution}}.
\]

The existence of such regions is a modelling assumption constrained by
empirical evidence. It should not be inferred merely from the existence of
an AI capability.

\subsection{Candidate labour-displacement region}

Let
\[
X_a,X_b
\in
C_{\times}
\]
be two joint economic--AI concepts for which
\[
X_a\sim_{\times}X_b.
\]

Let
\[
p_{\mathrm{Lab}}
=
(R,U,V_a,V_b)
\]
be a product-space interaction profile.

Its remainder is
\[
R_{E,\times}(p_{\mathrm{Lab}})
=
\operatorname{Int}_{\mathcal T_{\times}}
\bigl(
U\setminus(V_a\cup V_b)
\bigr).
\]

Choose admissible finite cores
\[
K_a\subseteq U\cap V_a,
\qquad
K_b\subseteq U\cap V_b,
\]
such that
\[
N_{\times}(K_a\cup\{X_a\})
\neq\varnothing,
\]
and
\[
N_{\times}(K_b\cup\{X_b\})
\neq\varnothing.
\]

Choose anchored witnesses
\[
W_a
\in
N_{\times}(K_a\cup\{X_a\}),
\qquad
W_b
\in
N_{\times}(K_b\cup\{X_b\}).
\]

If
\[
W_a
\cap
W_b
\cap
R_{E,\times}(p_{\mathrm{Lab}})
\neq\varnothing,
\]
define
\[
E_{\mathrm{Lab}}
:=
W_a
\cap
W_b
\cap
R_{E,\times}(p_{\mathrm{Lab}}).
\]

Then
\[
E_{\mathrm{Lab}}
\]
is a cross-space candidate emergent region.

By the ordinary CCER theory applied to the product concept space,
\[
E_{\mathrm{Lab}}
\in
\mathcal T_{\times}
\]
and
\[
E_{\mathrm{Lab}}
\in
I^{\times}_{ab}.
\]

No global assumption
\[
(N_{\mathrm{fin}})
\]
is required.

\subsection{Interpretation}

A candidate
\[
E_{\mathrm{Lab}}
\]
may represent configurations in which several structures coexist, for
example:

\[
\text{increased AI capability},
\]

\[
\text{organizational deployment},
\]

\[
\text{reduced demand for some human tasks},
\]

and

\[
\text{distributional or occupational adjustment}.
\]

The topology itself does not determine whether the resulting adjustment takes
the form of unemployment, shorter working time, new occupations, increased
output, higher wages, lower wages, redistribution, or complementary
human--machine production.

Those outcomes depend on additional economic and institutional structure.

If the candidate is selected for reification, specify a next-stage economic
concept set \(C_{\mathrm{Econ}}^{(1)}\) and an injective reification map
\[
\rho_{\mathrm{Econ}}:
\mathscr E_{\mathrm{Lab}}^{\mathrm{sel}}
\longrightarrow
C_{\mathrm{Econ}}^{(1)}.
\]
One may then define
\[
L_{\mathrm{disp}}^{(1)}
:=
\rho_{\mathrm{Econ}}\bigl(E_{\mathrm{Lab}}\bigr)
\in C_{\mathrm{Econ}}^{(1)},
\]
where \(L_{\mathrm{disp}}^{(1)}\) denotes AI-associated labour displacement
as a next-stage economic concept.

The earlier region and the later concept remain different objects:
\[
E_{\mathrm{Lab}}
\subseteq
C_{\mathrm{Econ}}
\times
C_{\mathrm{AI}},
\]
whereas
\[
L_{\mathrm{disp}}^{(1)}
\in
C_{\mathrm{Econ}}^{(1)}.
\]

\subsection{The notion of ``surplus population''}

The term \emph{surplus population} has a history in political economy and may
itself be represented as a concept within an HDCS model. It should not,
however, be treated as an intrinsic property of persons or as an empirically
established consequence of contemporary AI.

If an application explicitly studies the formation of this discourse, one may
introduce a next-stage concept
\[
S_{\mathrm{pop}}
\]
representing
\[
\text{the social or economic concept of a surplus population}.
\]

Its reification would require additional economic structures beyond
task-substitution alone, such as persistent labour displacement,
distributional arrangements, income institutions, ownership structures, and
the availability or absence of alternative forms of socially valued activity.

Thus
\[
\boxed{
\text{AI exposure}
\not\Rightarrow
\text{human redundancy}.
}
\]

The HDCS framework can represent the formation of such a concept under stated
conditions, but does not predict it as a necessary future outcome.

\subsection{Empirical limits of the application}

Current empirical evidence requires caution.

Exposure measures estimate which occupational tasks could potentially be
performed or transformed using AI. They are not direct predictions of job
loss. Actual employment effects depend on adoption, costs, complementary
capital, worker adaptation, regulation, demand responses, and the creation of
new tasks and occupations.

The application should therefore be interpreted as a structural scenario:
HDCS represents how particular configurations of AI capability, deployment,
labour demand, and economic distribution could support a new conceptual
region.

It does not establish that widespread technological unemployment or a
``redundant population'' currently exists or will necessarily emerge.

\section{General Facts on Cross-Space Emergence}
\label{sec:cross-space-general}

Let
\[
\mathfrak C_1
=
(C_1,\mathcal F_1,N_1),
\qquad
\mathfrak C_2
=
(C_2,\mathcal F_2,N_2)
\]
be concept spaces with external topologies
\[
\mathcal T_1
=
\tau_{\mathcal F_1},
\qquad
\mathcal T_2
=
\tau_{\mathcal F_2}.
\]

Let
\[
\mathfrak C_{\times}
=
(C_1\times C_2,\mathcal F_{\times},N_{\times})
\]
be their product concept space, with external topology
\[
\mathcal T_{\times}
=
\mathcal T_1\times\mathcal T_2.
\]

Since the product is itself a concept space, cross-space emergence requires no
independent emergence axiom. The ordinary CCER construction is applied inside
\[
\mathfrak C_{\times}.
\]

\subsection{Product-space CCER}

\begin{proposition}[Cross-space candidate emergence]
\label{prop:cross-candidate}
Let
\[
X_a,X_b\in C_1\times C_2
\]
and let
\[
p_{\times}
=
(R,U,V_a,V_b)
\]
be a product-space profile.

Let
\[
K_a\subseteq U\cap V_a,
\qquad
K_b\subseteq U\cap V_b
\]
be admissible finite cores and let
\[
W_a
\in
N_{\times}(K_a\cup\{X_a\}),
\qquad
W_b
\in
N_{\times}(K_b\cup\{X_b\})
\]
be anchored witnesses.

If
\[
E_{\times}
:=
W_a
\cap
W_b
\cap
R_{E,\times}(p_{\times})
\neq\varnothing,
\]
then:

\[
E_{\times}
\in
\mathcal T_{\times},
\]
and
\[
E_{\times}
\in
I^{\times}_{ab}.
\]

Moreover, every open subset
\[
G\subseteq E_{\times}
\]
also lies in
\[
I^{\times}_{ab}.
\]
\end{proposition}

\begin{proof}
The product-space neighbourhoods
\[
W_a,\qquad W_b
\]
are open, and by definition
\[
R_{E,\times}(p_{\times})
\]
is open. Hence
\[
E_{\times}
=
W_a\cap W_b\cap R_{E,\times}(p_{\times})
\]
is a nonempty open set.

Since
\[
E_{\times}
\subseteq
W_a\cap W_b
\]
and
\[
W_a\in N_{\times}(X_a),
\qquad
W_b\in N_{\times}(X_b),
\]
we have
\[
W_a\cap W_b
\in
O_{\times}(X_a,X_b).
\]

The downward closure of the channel ideal therefore gives
\[
E_{\times}
\in
I^{\times}_{ab}.
\]

The same argument applies to every open
\[
G\subseteq E_{\times}.
\]
\end{proof}

\subsection{Finite multi-profile candidates}

Cross-space applications may involve several distinct profiles. The correct
construction intersects already-supported candidate regions rather than raw
remainders.

\begin{definition}[Multi-profile cross-space candidate]
\label{def:cross-multi-profile}
Let
\[
E_1,\dots,E_m
\]
be cross-space CCER candidate regions, possibly arising from different parent
pairs.

If
\[
E_{\mathrm{multi}}
:=
\bigcap_{\ell=1}^{m}E_\ell
\neq\varnothing,
\]
then
\[
E_{\mathrm{multi}}
\]
is called a \emph{finite multi-profile cross-space candidate}.
\end{definition}

\begin{proposition}
\label{prop:cross-multi-profile}
A finite multi-profile cross-space candidate is open.

If
\[
E_\ell
\in
I^{\times}_{a_\ell b_\ell}
\]
for each \(\ell\), then
\[
E_{\mathrm{multi}}
\in
I^{\times}_{a_\ell b_\ell}
\]
for every \(\ell\).
\end{proposition}

\begin{proof}
Each \(E_\ell\) is open, and a finite intersection of open sets is open.

Furthermore,
\[
E_{\mathrm{multi}}
\subseteq
E_\ell
\]
for every \(\ell\). Since each channel ideal is downward closed,
\[
E_{\mathrm{multi}}
\in
I^{\times}_{a_\ell b_\ell}.
\]
\end{proof}

\subsection{Coordinate projections and internal bases}

\begin{proposition}[Projections of cross-space candidates]
\label{prop:cross-proj}
Let
\[
E\subseteq C_1\times C_2
\]
be an open cross-space candidate.

For \(j=1,2\), let
\[
\pi_j:C_1\times C_2\to C_j
\]
be the coordinate projection.

Then
\[
\pi_j|_E:E\to\pi_j(E)
\]
is continuous and open onto its image.

If
\[
\mathcal B_1
\quad\text{and}\quad
\mathcal B_2
\]
are bases of the factor topologies, then
\[
\mathcal B_E
:=
\left\{
(B_1\times B_2)\cap E
\;\middle|\;
B_1\in\mathcal B_1,\;
B_2\in\mathcal B_2,\;
(B_1\times B_2)\cap E\neq\varnothing
\right\}
\]
is a basis for the subspace topology on \(E\).
\end{proposition}

\begin{proof}
Coordinate projections of a product space are continuous and open.

Their restrictions to \(E\) are continuous. Since \(E\) is open, every open
subset of \(E\) is also open in the ambient product, and its projection is
open in the appropriate factor. Thus the restricted maps are open onto their
images.

The family
\[
\{B_1\times B_2\}
\]
is the standard product basis. Restricting that basis to \(E\) yields a basis
for the subspace topology.
\end{proof}

\subsection{Connectedness}

A cross-space candidate need not be connected.

The connected-core criterion developed earlier applies without modification.

\begin{corollary}[Connected cross-space candidate]
\label{cor:cross-connected}
Let
\[
E\subseteq C_1\times C_2
\]
be a cross-space candidate.

Suppose there is a nonempty connected set
\[
K\subseteq E
\]
such that every point of \(E\) can be connected to \(K\) by a finite chain of
connected subsets of \(E\) with successive nonempty intersections.

Then \(E\) is connected.
\end{corollary}

\begin{proof}
Apply Proposition~\ref{prop:hdcs-connected-core} in the product space.
\end{proof}

If the corresponding connecting sets are path connected, the analogous
path-connectedness conclusion follows under the appropriate linked-witness
hypotheses.

Connectedness is therefore an additional property to be verified; it is not
a consequence of cross-space CCER alone.

\subsection{Stability of inherited internal topology}

The previous formulation of ``stability under localized edits'' was stronger
than the axioms justify. The correct statement is conditional.

\begin{proposition}[Internal-topology stability]
\label{prop:cross-internal-stability}
Let
\[
E\subseteq C_1\times C_2
\]
be a candidate region at one stage, with ambient topology
\[
\mathcal T_{\times}.
\]

Suppose a later modelling edit produces an ambient topology
\[
\mathcal T'_{\times}
\]
on the same underlying region and satisfies
\[
\mathcal T_{\times}|_E
=
\mathcal T'_{\times}|_E.
\]

Then the internal topology of \(E\) is unchanged.
\end{proposition}

\begin{proof}
The internal topology is, by definition, the restriction of the ambient
topology to \(E\). Hence equality of the two restricted topologies gives the
claim immediately.
\end{proof}

\begin{remark}
Disjointness of an edit region from \(E\), or even from
\(\overline E\), does not by itself prove the hypothesis above unless the
transition mechanism includes an explicit locality rule guaranteeing that the
restricted topology on \(E\) remains unchanged.
\end{remark}

\section{What HDCS Gives Us: Derived Structural Consequences}
\label{sec:what-hdcs-gives}

The preceding sections define the HDCS architecture and illustrate how it may
be instantiated. The present section asks a different question: what follows
\emph{from the formalism itself}, once a concept space, profile, core data, and
stage structure have been specified?

The results below are not additional axioms. They are consequences of the
existing definitions. Their role is important for the explanatory status of
HDCS: they identify canonical structures, monotonicity laws, invariants, and
impossibility criteria. In this sense the framework does more than provide a
vocabulary for redescribing change; it also constrains what can occur inside a
specified model.

\subsection{Neighbourhood support and direct interaction}

\begin{definition}[Neighbourhood support]
\label{def:neighbourhood-support}
For every nonempty subset \(A\subseteq C\), define its \emph{neighbourhood
support} by
\[
\operatorname{Supp}_N(A)
:=
\bigcup_{U\in N(A)}U.
\]
If \(N(A)=\varnothing\), this union is understood to be empty. For a concept
\(x\in C\), write
\[
S(x):=\operatorname{Supp}_N(\{x\})
=
\bigcup_{U\in N(x)}U.
\]
\end{definition}

\begin{proposition}[Support characterization of adjacency]
\label{prop:support-adjacency}
For nonempty subsets \(A,B\subseteq C\),
\[
A\sim B
\quad\Longleftrightarrow\quad
\operatorname{Supp}_N(A)
\cap
\operatorname{Supp}_N(B)
\neq\varnothing.
\]
In particular, for concepts \(x,y\in C\),
\[
x\sim y
\quad\Longleftrightarrow\quad
S(x)\cap S(y)\neq\varnothing.
\]
\end{proposition}

\begin{proof}
If \(A\sim B\), choose \(U\in N(A)\) and \(V\in N(B)\) with
\(U\cap V\neq\varnothing\). Every point of \(U\cap V\) belongs to both
neighbourhood supports.

Conversely, suppose
\[
z\in
\operatorname{Supp}_N(A)
\cap
\operatorname{Supp}_N(B).
\]
Then there exist \(U\in N(A)\) and \(V\in N(B)\) such that
\(z\in U\cap V\). Hence \(U\cap V\neq\varnothing\), so \(A\sim B\).
\end{proof}

\begin{proposition}[Direct channel support]
\label{prop:direct-channel-support}
For concepts \(C_i,C_j\in C\), the support of the direct channel ideal satisfies
\[
|I_{ij}|
:=
\bigcup I_{ij}
=
\bigcup O(C_i,C_j)
=
S(C_i)\cap S(C_j).
\]
Consequently,
\[
C_i\sim C_j
\quad\Longleftrightarrow\quad
|I_{ij}|\neq\varnothing.
\]
\end{proposition}

\begin{proof}
Every member of \(I_{ij}\) is contained in a finite union of overlap regions,
so
\[
\bigcup I_{ij}
\subseteq
\bigcup O(C_i,C_j).
\]
Conversely, every overlap region belongs to \(I_{ij}\), giving the reverse
inclusion.

Finally,
\[
\bigcup_{U\in N(C_i),\,V\in N(C_j)}(U\cap V)
=
\left(\bigcup_{U\in N(C_i)}U\right)
\cap
\left(\bigcup_{V\in N(C_j)}V\right),
\]
which is \(S(C_i)\cap S(C_j)\). The last equivalence follows from
Proposition~\ref{prop:support-adjacency}.
\end{proof}

Thus the apparently higher-level direct channel has a simple underlying
support: it is exactly the part of the stage lying in the neighbourhood reach
of both concepts.

\subsection{Mediated components}

Recall that \(C_i\sim^{(*)}C_j\) means that there exists a finite adjacency
walk from \(C_i\) to \(C_j\).

\begin{theorem}[Mediated adjacency is an equivalence relation]
\label{thm:mediated-equivalence}
The relation \(\sim^{(*)}\) is an equivalence relation on the concept set \(C\).
\end{theorem}

\begin{proof}
Reflexivity follows because every concept is adjacent to itself, so a
length-one walk \((C_i,C_i)\) exists. Symmetry follows by reversing any finite
adjacency walk, since direct adjacency is symmetric. For transitivity, if
\[
C_i=X_0\sim X_1\sim\cdots\sim X_k=C_j
\]
and
\[
C_j=Y_0\sim Y_1\sim\cdots\sim Y_m=C_\ell,
\]
then concatenating the walks at the common endpoint \(C_j\) yields a finite
adjacency walk from \(C_i\) to \(C_\ell\).
\end{proof}

\begin{corollary}[Mediated-interaction decomposition]
\label{cor:mediated-components}
Every concept space is canonically partitioned into equivalence classes
\[
[C]_{*}
:=
\{D\in C:C\sim^{(*)}D\}.
\]
Two concepts lie in the same class if and only if they are connected by a
finite adjacency walk.
\end{corollary}

These classes are graph-theoretic interaction components. They need not be
connected subsets in the topological sense; no such identification is made.

\subsection{The canonical CCER envelope}

Fix an interaction profile
\[
p=(R,U,V_i,V_j)
\]
for distinct concepts \(C_i,C_j\), and fix an admissible core pair
\((K_i,K_j)\).

\begin{definition}[Anchored support]
\label{def:anchored-support}
Define the anchored supports
\[
A_i(K_i)
:=
\bigcup N(K_i\cup\{C_i\}),
\qquad
A_j(K_j)
:=
\bigcup N(K_j\cup\{C_j\}).
\]
Define the associated \emph{CCER envelope} by
\[
E^{\max}(p;K_i,K_j)
:=
R_E(p)
\cap
A_i(K_i)
\cap
A_j(K_j).
\]
\end{definition}

The adjective ``max'' refers to the fact proved below: this set is the union
of all candidate regions obtainable from the fixed profile and fixed core pair.

\begin{theorem}[Exact CCER existence criterion]
\label{thm:ccer-envelope}
Let \(p\) and \((K_i,K_j)\) be as above. Then the following are equivalent:

\begin{enumerate}[label=(\roman*)]
\item there exists a productive anchored witness pair for \((K_i,K_j)\);
\item
\[
E^{\max}(p;K_i,K_j)\neq\varnothing.
\]
\end{enumerate}

Moreover, if \(\mathscr C(p;K_i,K_j)\) denotes the family of all candidate
emergent regions produced from productive anchored witness pairs for this
profile and core pair, then
\[
E^{\max}(p;K_i,K_j)
=
\bigcup_{E\in\mathscr C(p;K_i,K_j)}E.
\]
In particular, \(E^{\max}(p;K_i,K_j)\) is open.
\end{theorem}

\begin{proof}
Suppose \((W_i,W_j)\) is productive and choose
\[
z\in W_i\cap W_j\cap R_E(p).
\]
Since \(W_i\in N(K_i\cup\{C_i\})\) and
\(W_j\in N(K_j\cup\{C_j\})\), the point \(z\) belongs to both anchored
supports. Hence \(z\in E^{\max}(p;K_i,K_j)\).

Conversely, suppose
\[
z\in E^{\max}(p;K_i,K_j).
\]
Because \(z\in A_i(K_i)\), some
\[
W_i\in N(K_i\cup\{C_i\})
\]
contains \(z\). Likewise, some
\[
W_j\in N(K_j\cup\{C_j\})
\]
contains \(z\). Since also \(z\in R_E(p)\),
\[
z\in W_i\cap W_j\cap R_E(p),
\]
so \((W_i,W_j)\) is productive.

Every candidate produced from the fixed data is contained in the three sets
defining \(E^{\max}\), hence in \(E^{\max}\). The converse pointwise argument
above shows that every point of \(E^{\max}\) belongs to at least one such
candidate. Therefore the stated union equality holds.

Each anchored support is a union of open neighbourhood regions, and \(R_E(p)\)
is open. Their finite intersection is therefore open.
\end{proof}

\begin{remark}[Envelope versus individual candidate]
The envelope \(E^{\max}(p;K_i,K_j)\) need not itself be produced by one witness
pair, and it need not belong to \(I_{ij}\) or \(\mathscr E_{ij}\), because those
ideals are required to be closed only under \emph{finite} unions, whereas the
envelope may be an infinite union of candidate regions.
\end{remark}

\begin{corollary}[Compact envelope]
\label{cor:compact-envelope}
If \(E^{\max}(p;K_i,K_j)\) is compact, then
\[
E^{\max}(p;K_i,K_j)
\in
I_{ij}\cap\mathscr E_{ij}.
\]
\end{corollary}

\begin{proof}
By Theorem~\ref{thm:ccer-envelope}, the candidate regions generated from the
fixed profile and core pair form an open cover of \(E^{\max}\). Compactness
gives a finite subcover
\[
E^{\max}=E_1\cup\cdots\cup E_n.
\]
Every \(E_r\) lies in both \(I_{ij}\) and \(\mathscr E_{ij}\). Since both are
closed under finite unions, the result follows.
\end{proof}

The envelope therefore removes part of the apparent witness-choice
nondeterminism: once the profile and core pair are fixed, there is a canonical
open region containing exactly the points that can participate in some
productive witness realization.

\subsection{Monotonicity and impossibility}

\begin{proposition}[Core monotonicity]
\label{prop:core-monotonicity}
Fix a profile \(p\). Suppose
\[
K_i\subseteq K_i',
\qquad
K_j\subseteq K_j',
\]
and both core pairs are admissible. Then
\[
E^{\max}(p;K_i',K_j')
\subseteq
E^{\max}(p;K_i,K_j).
\]
\end{proposition}

\begin{proof}
Antitonicity of common-neighbourhood families gives
\[
N(K_i'\cup\{C_i\})
\subseteq
N(K_i\cup\{C_i\}),
\]
and similarly for the \(j\)-side. Taking unions yields
\[
A_i(K_i')\subseteq A_i(K_i),
\qquad
A_j(K_j')\subseteq A_j(K_j).
\]
Intersecting with the same remainder \(R_E(p)\) gives the result.
\end{proof}

Thus stronger core requirements can only preserve or reduce the canonical
emergence envelope; they cannot create new points of possible emergence.

\begin{proposition}[Parent-coverage monotonicity]
\label{prop:parent-coverage-monotonicity}
Let
\[
p=(R,U,V_i,V_j),
\qquad
p'=(R,U,V_i',V_j')
\]
be profiles for the same parent pair, with
\[
V_i\subseteq V_i',
\qquad
V_j\subseteq V_j'.
\]
Then
\[
R_E(p')\subseteq R_E(p).
\]
If a fixed core pair \((K_i,K_j)\) is admissible for both profiles, then
\[
E^{\max}(p';K_i,K_j)
\subseteq
E^{\max}(p;K_i,K_j).
\]
\end{proposition}

\begin{proof}
The parent inclusion assumptions give
\[
U\setminus(V_i'\cup V_j')
\subseteq
U\setminus(V_i\cup V_j).
\]
Monotonicity of the interior operator therefore yields
\[
\operatorname{Int}_{\mathcal T}
\bigl(U\setminus(V_i'\cup V_j')\bigr)
\subseteq
\operatorname{Int}_{\mathcal T}
\bigl(U\setminus(V_i\cup V_j)\bigr).
\]
The envelope inclusion follows by intersecting both sides with the same two
anchored supports.
\end{proof}

This formalizes an intuitive constraint: the more of the background context
already covered by the selected parent neighbourhoods, the less uncovered
structure remains available to CCER.

\begin{definition}[Profile interaction zone]
\label{def:profile-interaction-zone}
For a profile \(p\) between \(C_i\) and \(C_j\), define
\[
Z_{ij}(p)
:=
R_E(p)\cap |I_{ij}|.
\]
By Proposition~\ref{prop:direct-channel-support}, equivalently
\[
Z_{ij}(p)
=
R_E(p)\cap S(C_i)\cap S(C_j).
\]
\end{definition}

\begin{proposition}[Necessary interaction zone]
\label{prop:interaction-zone}
Every CCER candidate \(E\) arising from the profile \(p\) satisfies
\[
E\subseteq Z_{ij}(p).
\]
Consequently,
\[
Z_{ij}(p)=\varnothing
\quad\Longrightarrow\quad
\text{the profile \(p\) admits no CCER candidate.}
\]
\end{proposition}

\begin{proof}
Every candidate satisfies \(E\subseteq R_E(p)\). Also, its anchored witnesses
satisfy
\[
W_i\in N(K_i\cup\{C_i\})\subseteq N(C_i),
\qquad
W_j\in N(K_j\cup\{C_j\})\subseteq N(C_j).
\]
Hence
\[
E\subseteq W_i\cap W_j
\subseteq S(C_i)\cap S(C_j)
=|I_{ij}|.
\]
Intersecting the two inclusions gives \(E\subseteq Z_{ij}(p)\). If the latter
set is empty, no nonempty candidate can exist.
\end{proof}

The converse is not asserted: a nonempty interaction zone is a necessary
screening condition, not by itself a sufficient condition for admissible cores
and productive anchored witnesses.

\subsection{Factorization in product concept spaces}

Let
\[
\mathfrak C_1=(C_1,\mathcal F_1,N_1),
\qquad
\mathfrak C_2=(C_2,\mathcal F_2,N_2)
\]
be concept spaces and let \(\mathfrak C_\times\) be their product concept
space.

\begin{proposition}[Factorization of direct product adjacency]
\label{prop:product-adjacency-factorization}
For
\[
X=(x_1,x_2),
\qquad
Y=(y_1,y_2)
\]
in \(C_1\times C_2\),
\[
X\sim_\times Y
\quad\Longleftrightarrow\quad
x_1\sim_1 y_1
\quad\text{and}\quad
x_2\sim_2 y_2.
\]
\end{proposition}

\begin{proof}
A product neighbourhood of \(X\) has the form \(U_1\times U_2\) with
\(U_r\in N_r(x_r)\), and a product neighbourhood of \(Y\) has the form
\(V_1\times V_2\) with \(V_r\in N_r(y_r)\). Their intersection is
\[
(U_1\cap V_1)\times(U_2\cap V_2),
\]
which is nonempty exactly when both coordinate intersections are nonempty.
\end{proof}

\begin{theorem}[Factorization of mediated product components]
\label{thm:product-mediated-factorization}
For \(X=(x_1,x_2)\) and \(Y=(y_1,y_2)\),
\[
X\sim_\times^{(*)}Y
\quad\Longleftrightarrow\quad
x_1\sim_1^{(*)}y_1
\quad\text{and}\quad
x_2\sim_2^{(*)}y_2.
\]
Consequently,
\[
[(x_1,x_2)]_{*,\times}
=
[x_1]_{*,1}\times[x_2]_{*,2}.
\]
\end{theorem}

\begin{proof}
If a finite product-space adjacency walk joins \(X\) to \(Y\), projecting each
step to either coordinate and applying
Proposition~\ref{prop:product-adjacency-factorization} gives adjacency walks in
both factors.

Conversely, suppose there are factor walks of lengths \(k\) and \(m\). Since
every concept is adjacent to itself, terminal repetitions may pad the shorter
walk until both have the same length \(L=\max\{k,m\}\). Pairing the two
coordinate walks step by step then gives a length-\(L\) product-space adjacency
walk by Proposition~\ref{prop:product-adjacency-factorization}.

The component identity follows immediately.
\end{proof}

Thus cross-space mediated reachability is not arbitrary: under the rectangular
product construction, it occurs precisely when mediated reachability occurs in
both factors.

\subsection{Structural isomorphism and genuine stage change}

Literal inequality of two stage triples can arise merely from relabelling. To
distinguish this from intrinsic structural change, HDCS admits the following
notion.

\begin{definition}[Concept-space isomorphism]
\label{def:concept-space-isomorphism}
Let
\[
\mathfrak C=(C,\mathcal F,N),
\qquad
\mathfrak C'=(C',\mathcal F',N').
\]
A \emph{concept-space isomorphism} is a bijection
\[
f:C\to C'
\]
such that
\[
\mathcal F'
=
\{f[U]:U\in\mathcal F\}
\]
and, for every \(x\in C\),
\[
N'(f(x))
=
\{f[U]:U\in N(x)\}.
\]
When such an \(f\) exists, write
\[
\mathfrak C\cong\mathfrak C'.
\]
\end{definition}

\begin{theorem}[Invariance of the single-stage HDCS structure]
\label{thm:hdcs-isomorphism-invariance}
Let \(f:\mathfrak C\to\mathfrak C'\) be a concept-space isomorphism. Then:

\begin{enumerate}[label=(\roman*)]
\item \(f\) is a homeomorphism between the external topologies;
\item adjacency is preserved:
\[
A\sim B
\quad\Longleftrightarrow\quad
f[A]\sim' f[B]
\]
for nonempty \(A,B\subseteq C\);
\item overlap families and direct channel ideals are transported by \(f\):
\[
O'(f(C_i),f(C_j))
=
\{f[R]:R\in O(C_i,C_j)\},
\]
\[
I'_{f(C_i)f(C_j)}
=
\{f[W]:W\in I_{ij}\};
\]
\item if
\[
p=(R,U,V_i,V_j)
\]
is a profile, then
\[
f(p)
:=
(f[R],f[U],f[V_i],f[V_j])
\]
is a profile and
\[
R_E'(f(p))
=
f[R_E(p)];
\]
\item admissible cores, productive witness pairs, candidate regions, and CCER
envelopes are carried to their corresponding images under \(f\).
\end{enumerate}
\end{theorem}

\begin{proof}
Because \(\mathcal F'\) is exactly the image of \(\mathcal F\), arbitrary
unions of feasible regions are transported to arbitrary unions of feasible
regions. Hence
\[
\mathcal T'
=
\{f[U]:U\in\mathcal T\},
\]
so the bijection \(f\) is a homeomorphism.

The neighbourhood identity in the definition transports common-neighbourhood
families and therefore preserves adjacency. It also transports every overlap
\(U\cap V\) to
\[
f[U\cap V]=f[U]\cap f[V],
\]
which proves the overlap statement. Finite unions and inclusion are preserved
by a bijection, giving the channel-ideal statement.

Since homeomorphisms preserve interiors,
\[
f\!\left[\operatorname{Int}_{\mathcal T}
\bigl(U\setminus(V_i\cup V_j)\bigr)\right]
=
\operatorname{Int}_{\mathcal T'}
\bigl(f[U]\setminus(f[V_i]\cup f[V_j])\bigr),
\]
which proves the remainder identity. Finiteness, containment, common
neighbourhood conditions, witness membership, nonempty intersections, and the
candidate formula are all preserved under the bijection. The envelope identity
then follows from Definition~\ref{def:anchored-support}.
\end{proof}

\begin{definition}[Intrinsic stage change]
\label{def:intrinsic-stage-change}
A transition from \(\mathfrak C^{(i)}\) to \(\mathfrak C^{(i+1)}\) exhibits
\emph{intrinsic stage-structural change} when
\[
\mathfrak C^{(i)}
\not\cong
\mathfrak C^{(i+1)}.
\]
\end{definition}

This notion complements the nonstationarity condition of
Section~\ref{sec:hdcs-stages}. Two stages may be literally unequal but
isomorphic, in which case the single-stage structural difference is only a
relabelling. Conversely, even isomorphic stage spaces may participate in a
nontrivial HDCS history through different carry, reification, or provenance
data. Thus HDCS distinguishes at least three levels:
\[
\text{literal change},
\qquad
\text{intrinsic structural change},
\qquad
\text{historical change through stage relations}.
\]

\subsection{What the framework can rule out}

The preceding results yield several immediate negative conclusions. Once the
model data are fixed, HDCS can rule out the following possibilities:

\begin{enumerate}[label=\arabic*.]
\item If \(R_E(p)=\varnothing\), then the profile admits no CCER candidate.

\item If
\[
E^{\max}(p;K_i,K_j)=\varnothing,
\]
then the fixed core pair admits no productive anchored witnesses.

\item If
\[
Z_{ij}(p)=\varnothing,
\]
then the profile admits no CCER candidate for \emph{any} admissible core pair.

\item Enlarging a core cannot enlarge its CCER envelope, and enlarging the
selected parent coverage cannot enlarge the uncovered remainder.

\item In a product concept space, product adjacency is impossible whenever
adjacency fails in either factor; mediated product reachability is impossible
whenever mediated reachability fails in either factor.

\item Under the injective carry maps of HDCS, two distinct concepts at one
stage cannot become one history identity merely by later carry. Any genuine
convergence must be represented as a new concept with explicit provenance.
\end{enumerate}

These are modest results, but they mark an important boundary. HDCS does not
predict which empirically available transition must occur. It does, however,
derive constraints on which interactions, candidates, identities, and
cross-space relations are possible within a specified model.

\subsection{Summary: what HDCS gives us}

At its present level, the framework yields five kinds of mathematical output:

\begin{enumerate}[label=\arabic*.]
\item \textbf{Canonical interaction structure:} neighbourhood supports,
direct-channel supports, and mediated-interaction components.

\item \textbf{Canonical emergence bounds:} for fixed profile and core data,
the CCER envelope gives the entire open region in which some productive
witness realization can occur.

\item \textbf{Monotonicity:} stronger core constraints and larger parent
coverage cannot enlarge the corresponding emergence possibilities.

\item \textbf{Impossibility tests:} empty interaction zones, failed factor
adjacency, and failed mediated reachability exclude corresponding forms of
candidate emergence or cross-space interaction.

\item \textbf{Structural invariants:} concept-space isomorphisms preserve the
single-stage HDCS constructions, allowing intrinsic structural change to be
distinguished from mere relabelling.
\end{enumerate}

The resulting picture is therefore narrower than a causal law of historical
change but stronger than unrestricted description. Once the modelling data
are specified, the theory produces derived regions, equivalence classes,
monotonicity relations, invariants, and no-go conditions that are not separately
inserted by the modeller.

\begin{figure}[ht]
\centering
\begin{tikzpicture}[
  box/.style={draw, rounded corners, align=center, inner sep=5pt,
              text width=4.0cm, minimum height=9mm, font=\small},
  arr/.style={-{Latex[length=2.2mm]}, thick}
]
\node[box] (stage) at (0,3.3) {stage data\\$(C,\mathcal F,N)$};
\node[box] (support) at (6.1,3.3) {neighbourhood supports\\and mediated components};
\node[box] (profile) at (0,1.1) {profile + admissible cores};
\node[box] (env) at (6.1,1.1) {canonical CCER envelope\\$E^{\max}$};
\node[box] (prod) at (0,-1.1) {product concept spaces};
\node[box] (fact) at (6.1,-1.1) {factorized adjacency\\and mediated reachability};
\node[box] (hist) at (0,-3.3) {injective carry + reification};
\node[box] (ident) at (6.1,-3.3) {identity preservation\\and provenance};

\draw[arr] (stage) -- (support);
\draw[arr] (stage) -- (profile);
\draw[arr] (profile) -- (env);
\draw[arr] (support) -- node[right, font=\scriptsize, align=left]{bounds the\\interaction zone} (env);
\draw[arr] (stage.west) to[out=180,in=180,looseness=1.15] (prod.west);
\draw[arr] (prod) -- (fact);
\draw[arr] (stage.west) to[out=190,in=170,looseness=1.25] (hist.west);
\draw[arr] (hist) -- (ident);
\end{tikzpicture}
\caption{Some derived consequences of the HDCS architecture. The arrows indicate
logical dependence or derived constraint, not historical causation.}
\label{fig:hdcs-derived-consequences}
\end{figure}

\section{Conclusion}

HDCS provides a stagewise topological language for separating several forms of
structural change that are easily conflated: persistence of an existing
concept, alteration of its neighbourhood relations, disappearance from a
later stage, formation of a candidate emergent region, and reification of a
selected region as a new atomic concept. The formal construction does not
identify these processes with one another.

At a single stage, the feasible family generates the external topology while
the neighbourhood assignment carries additional relational information used
for adjacency, overlaps, and channels. CCER combines an uncovered remainder
with locally admissible anchored witnesses to produce nonempty open candidate
regions. The construction is local: no global finite-common-neighbourhood
axiom is required. Candidate regions remain regions until an explicit
reification map assigns selected candidates to later-stage atoms.

Across stages, injective carry maps record declared persistence while the
transition mechanism permits the feasible family and neighbourhood assignment
to change. This makes it possible for a concept to persist while its structural
role is reorganized. Distinct persistence trajectories are never merged by the
history quotient; convergence of earlier concepts into a new structure is
represented instead by reification with explicit provenance. The history
quotient therefore records identity through carry without identifying a
candidate region with its reified atom. Product concept spaces extend the same
mechanism to cross-domain configurations without introducing a second
emergence rule.

The applications are intended to test expressive adequacy rather than to make
topology a causal explanation of history, evolution, or economic change. In
particular, an HDCS stage may represent the renewed prominence of a historically
old practice, such as barter under monetary deformation, just as readily as it
may represent the stabilization of a newly modelled concept. What is stagewise
new is therefore not necessarily new in world history; it may instead be a new
structural role within the current concept space.

The derived results of Section~\ref{sec:what-hdcs-gives} sharpen the status of
the framework. Once model data are fixed, HDCS yields canonical interaction
supports and mediated components, an exact existence criterion and canonical
CCER envelope for fixed core data, monotonicity laws, no-go conditions,
product-space factorization, and invariance under concept-space isomorphism.
These consequences do not turn topology into a causal law, but they show that
the formalism can constrain a specified model rather than merely redescribe it.

The present framework remains deliberately incomplete as a dynamical theory.
The transition mechanism is abstract, empirical applications must justify the
chosen feasible and neighbourhood structures, and continuity, connectedness,
compactness, or branching require additional hypotheses or extensions when
needed. Within those limits, the framework supplies a consistent vocabulary
for asking which structures persist, which relations change, which candidate
regions become available, and which of those regions are subsequently treated
as concepts.

\appendix
\section{Logical Dependency Structure of HDCS}
\label{app:dependency}

This appendix summarizes the logical organization of the revised framework.
The purpose is to distinguish primitive assumptions from derived properties
and optional analytical conditions.

Hard-coded proposition numbers are avoided because the numbering may change as
the manuscript is revised.

\subsection{Foundational structure}

The primitive single-stage data are
\[
(C,\mathcal F,N),
\]
where \(\mathcal F\) satisfies
\[
(F1),\qquad
(F\cap),\qquad
(F\varnothing),
\]
and \(N\) satisfies
\[
(N0),
\qquad
(N\downarrow).
\]

The logical dependence is
\[
\begin{array}{c}
\text{Feasible family }\mathcal F
\quad+\quad
\text{neighbourhood assignment }N
\\[1mm]
\downarrow
\\[1mm]
\mathcal F\text{ is a basis}
\\[1mm]
\downarrow
\\[1mm]
\mathcal T
=
\tau_{\mathcal F}
\\[1mm]
\downarrow
\\[1mm]
\text{every feasible region and every assigned neighbourhood is open.}
\end{array}
\]

Equivalently,
\[
\mathcal T
=
\left\{
U\subseteq C:
\forall x\in U\
\exists V\in N(x)
\text{ with }V\subseteq U
\right\}.
\]

\subsection{Derived neighbourhood properties}

The following are derived rather than independently assumed:

\[
U\in N(x)
\Longrightarrow
U\neq\varnothing,
\]

and
\[
U,V\in N(x)
\Longrightarrow
U\cap V\in N(x).
\]

Thus the former axiom
\[
(N\cap)
\]
is redundant.

Likewise, the cross-point property formerly written
\[
(N\Rightarrow)
\]
follows from the feasible-intersection axiom together with
\((N0)\) and \((N\downarrow)\), and is not an additional axiom.

\subsection{Structural interaction}

The structural interaction layer depends directly on the concept-space data:
\[
\begin{array}{c}
(C,\mathcal F,N)
\\
\downarrow
\\
\text{adjacency}
\\
\downarrow
\\
\text{overlap families}
\\
\downarrow
\\
\text{direct channel ideals}
\\
\downarrow
\\
\text{mediated channel constructions}.
\end{array}
\]
The external topology is not needed to define adjacency itself; it enters the
profile and remainder constructions below.

Adjacency is symmetric but need not be transitive.

Multi-step channel structure records finite witnessed adjacency walks.
Cumulative ideals
\[
I^{(\leq k)}
\]
are monotone in \(k\).

\subsection{Profiles and remainders}

Profiles depend on overlap structure:
\[
\begin{array}{c}
\text{overlap family}
+
\text{external topology}
\\
\downarrow
\\
\text{interaction profile }
p=(R,U,V_i,V_j)
\\
\downarrow
\\
R_E(p)
=
\operatorname{Int}_{\mathcal T}
\bigl(
U\setminus(V_i\cup V_j)
\bigr)
\\
\downarrow
\\
\text{remainder ideals and related support constructions}.
\end{array}
\]

The remainder is an open uncovered part of a selected interaction context.
No interpretation of contradiction, insufficiency, or causal pressure follows
from topology alone.

\subsection{Candidate emergence}

The revised CCER dependence is
\[
\begin{array}{c}
\text{profile}
+
\text{admissible finite cores}
\\
\downarrow
\\
\text{anchored witnesses}
\\
\downarrow
\\
\text{productive witness pair}
\\
\downarrow
\\
E
=
W_i\cap W_j\cap R_E(p)
\\
\downarrow
\\
\text{candidate emergent region}
\\
\downarrow
\\
\mathcal T|_E
\\
\downarrow
\\
\text{optional selection and reification}.
\end{array}
\]

The global condition
\[
(N_{\mathrm{fin}})
\]
is not required.

Instead, admissible cores satisfy the local conditions
\[
N(K_i\cup\{C_i\})
\neq\varnothing,
\]
and
\[
N(K_j\cup\{C_j\})
\neq\varnothing.
\]

\subsection{Reification and stage dynamics}

Candidate regions and concepts have different types.

At stage \(i\),
\[
E\subseteq C^{(i)}.
\]

If \(E\) is selected for reification,
\[
e_E
=
\rho_i(E)
\in
C^{(i+1)}.
\]

Persistence of existing concepts is represented separately by the injective
carry map
\[
\sigma_i:
D_i\to C^{(i+1)}.
\]
For every stage transition, persistence and reification have disjoint images:
\[
\sigma_i(D_i)
\cap
\rho_i\bigl(\mathscr E_i^{\mathrm{sel}}\bigr)
=
\varnothing.
\]
Thus a next-stage atom cannot be both the declared continuation of an earlier
identity and a newly reified concept.

Thus:
\[
\begin{array}{c}
\text{candidate regions}
\\
\downarrow\ \rho_i
\\
\text{new stage concepts}
\end{array}
\qquad\qquad
\begin{array}{c}
\text{persisting concepts}
\\
\downarrow\ \sigma_i
\\
\text{successor concepts}.
\end{array}
\]

The transition mechanism
\[
\Phi_i
\]
must output another valid concept space but is otherwise left abstract unless
additional domain-specific dynamics are supplied.

\subsection{History space}

The history construction depends on the carry maps:
\[
\begin{array}{c}
\text{stage concept spaces}
+
\text{carry maps }\sigma_i
\\
\downarrow
\\
\bigsqcup_i C^{(i)}
\\
\downarrow
\\
\text{equivalence relation generated by carry identifications}
\\
\downarrow
\\
C_{\mathrm{hist}}
\\
\downarrow
\\
\text{final history topology}.
\end{array}
\]

Carry maps are injective because they represent persistence of identity.
Consequently, every canonical stage map into the history set is injective and
later convergence does not retroactively identify distinct earlier concepts.
Continuity of the carry maps remains optional.

Continuity is required if one wishes to regard the stage sequence as a diagram
of continuous maps, but it is not required for the history quotient or final
topology to exist.

\subsection{Cross-space construction}

Cross-space structure branches directly from the single-stage theory:

\[
\begin{array}{c}
(C_1,\mathcal F_1,N_1)
\qquad
(C_2,\mathcal F_2,N_2)
\\
\downarrow
\\
(C_1\times C_2,\mathcal F_{\times},N_{\times})
\\
\downarrow
\\
\mathcal T_1\times\mathcal T_2
\\
\downarrow
\\
\text{ordinary profiles and remainders in the product space}
\\
\downarrow
\\
\text{ordinary CCER applied to the product concept space}.
\end{array}
\]

Thus ``cross-space CCER'' is not a second emergence mechanism. It is ordinary
CCER operating in a product concept space.

\subsection{Analytical properties}

The following properties are optional and require their own hypotheses:

\[
\text{continuity},
\qquad
\text{connectedness},
\qquad
\text{path connectedness},
\qquad
\text{compactness}.
\]

In particular:

\[
\text{adjacency}
\not\Rightarrow
\text{topological connectedness},
\]

\[
\text{candidate emergence}
\not\Rightarrow
\text{connectedness},
\]

and
\[
\text{openness}
\not\Rightarrow
\text{compactness}.
\]

\subsection{Derived structural consequences}

Section~\ref{sec:what-hdcs-gives} derives several consequences without adding
new foundational axioms. Their dependencies are:

\[
(C,\mathcal F,N)
\longrightarrow
\text{neighbourhood supports}
\longrightarrow
\text{direct-channel supports and mediated components},
\]

\[
\text{profile + admissible cores}
\longrightarrow
E^{\max}(p;K_i,K_j)
\longrightarrow
\text{exact existence and impossibility tests},
\]

\[
\text{product concept spaces}
\longrightarrow
\text{factorized direct and mediated adjacency},
\]

and
\[
\text{concept-space isomorphisms}
\longrightarrow
\text{invariance of the single-stage HDCS constructions}.
\]

These results distinguish constraints derived from a specified model from
application-specific causal or historical interpretation.

\subsection{Dependency diagram}

\begin{center}
\begin{tikzpicture}[
  box/.style={
    draw,
    rectangle,
    rounded corners,
    align=center,
    minimum width=3.7cm,
    minimum height=0.9cm,
    font=\small
  },
  arrow/.style={->, thick},
  node distance=1.15cm and 1.1cm
]

\node[box] (foundation)
{Feasible family \(\mathcal F\)\\
Neighbourhood assignment \(N\)};

\node[box, below left=of foundation] (topology)
{External topology\\
\(\mathcal T=\tau_{\mathcal F}\)};

\node[box, below right=of foundation] (interaction)
{Adjacency, overlaps\\
and channel ideals};

\node[box, below=2.0cm of foundation] (profiles)
{Profiles and remainders};

\node[box, below=of profiles] (ccer)
{Finite cores, witnesses\\
and candidate emergence};

\node[box, below left=of ccer] (reification)
{Reification \(\rho_i\)\\
and stage dynamics};

\node[box, below right=of ccer] (product)
{Product concept spaces\\
and cross-space CCER};

\node[box, below=of reification] (history)
{Carry maps \(\sigma_i\)\\
and history space};

\node[box, below=of product] (analysis)
{Optional analysis:\\
continuity, connectedness, compactness};

\draw[arrow] (foundation) -- (topology);
\draw[arrow] (foundation) -- (interaction);
\draw[arrow] (topology) -- (profiles);
\draw[arrow] (interaction) -- (profiles);
\draw[arrow] (profiles) -- (ccer);
\draw[arrow] (ccer) -- (reification);
\draw[arrow] (ccer) -- (product);
\draw[arrow] (reification) -- (history);
\draw[arrow] (product) -- (analysis);

\end{tikzpicture}
\end{center}

\subsection*{Assumptions and optional conditions}

The globally imposed structural axioms are

\[
(F1),
\qquad
(F\cap),
\qquad
(F\varnothing),
\]
for feasible families, and
\[
(N0),
\qquad
(N\downarrow),
\]
for neighbourhood assignments.

The properties formerly denoted
\[
(N\cap)
\qquad\text{and}\qquad
(N\Rightarrow)
\]
are derived and are not additional assumptions.

The condition
\[
(N_{\mathrm{fin}})
\]
is not imposed globally.

When a CCER construction uses a finite core \(K\), the required condition is
local:
\[
N(K\cup\{C\})
\neq\varnothing.
\]

Continuity of carry maps, linkedness of witness walks, connectedness,
path-connectedness, and compactness are optional analytical hypotheses used
only when the corresponding conclusions are required.

\subsection*{Terminology and ontological status}

A \emph{concept} is an atomic element
\[
x\in C^{(i)}.
\]

A \emph{region} is a subset
\[
U\subseteq C^{(i)}.
\]

A \emph{candidate emergent region} is a nonempty open region produced by a
productive CCER witness pair:
\[
E
=
W_i\cap W_j\cap R_E(p).
\]

A candidate region is not automatically a concept.

A \emph{reified concept} is an element of the next stage produced from a
selected candidate:
\[
e_E
=
\rho_i(E)
\in
C^{(i+1)}.
\]

Its earlier regional structure may be retained as provenance:
\[
\operatorname{Real}_{i+1}(e_E)
=
\bigl(
E,\mathcal T^{(i)}|_E
\bigr).
\]

A \emph{carry map}
\[
\sigma_i:
D_i\to C^{(i+1)}
\]
tracks persistent concepts and is logically distinct from reification.

\vspace{8em}
\begin{verse}
Karanfil sokağında bir camlı bahçe\\
Camlı bahçe içre bir çini saksı\\
Bir dal süzülür mavide\\
Al - al bir yangın şarkısı\\
Bakmayın saksıda boy verdiğine\\
Kökü Altındağ'da, İncesu'dadır.
\end{verse}

\begin{flushleft}
A. Arif
\end{flushleft}

\end{document}